%% file: which-effect-of-race-arxiv.tex
\documentclass[12pt]{article}
\usepackage[T1]{fontenc}
\usepackage[utf8]{inputenc}
\usepackage{lmodern}
\usepackage{setspace}
\usepackage{graphicx}
\usepackage{enumerate}
\usepackage{enumitem}
\usepackage[natbib,backend=biber,sorting=nyt,maxcitenames=2,mincitenames=1,maxbibnames=4,uniquelist=false, doi=false, style = authoryear-comp]{biblatex}
\usepackage{url} %
\usepackage{graphicx}
\usepackage{float}
\usepackage{booktabs}
\usepackage{tikz}
\usepackage[table]{xcolor}
\usepackage{colortbl}
\usepackage{caption}
\usepackage[normalem]{ulem}
\usepackage{diagbox}

\usepackage{amsfonts,amsthm,amsmath,amssymb} 
\usepackage{mathtools}
\usepackage{bbm} %
\usepackage{bm} %
\usepackage{csquotes}
\usepackage{multicol}
\usepackage{hyperref}
\usepackage[noabbrev,capitalise]{cleveref} 

\newtheoremstyle{newstyle}
{10pt} %
{10pt} %
{\itshape\onehalfspacing} %
{} %
{\bfseries} %
{.} %
{ } %
{} %

\theoremstyle{newstyle}

\newtheorem{prop}{Proposition}
\newtheorem{defin}{Definition}
\newtheorem{lem}{Lemma}
\newtheorem{cor}{Corollary}

\newtheorem{assm}{Assumption}

\DeclareMathOperator{\E}{\rm{E}}
\DeclareMathOperator{\R}{\mathbbm{R}}

\DeclareMathOperator{\Var}{\rm{Var}}
\DeclareMathOperator{\Cov}{\rm{Cov}}

\DeclareMathOperator{\diagop}{diag}
\newcommand{\diag}[1]{\diagop\left(#1\right)}

\DeclarePairedDelimiter\abs{\lvert}{\rvert}
\newcommand\given[1][]{\:#1\vert\:}

\renewcommand{\arraystretch}{.6} %

\newcommand{\blind}{0}

\ifnum\blind=1
  \hypersetup{pdfauthor={}, pdftitle={Which Effect of Race? Causal Inference without Holding All Else Equal}, pdflang={en-US}}
\else
  \hypersetup{pdfauthor={Thomas Leavitt}, pdftitle={Which Effect of Race? Causal Inference without Holding All Else Equal}, pdflang={en-US}}
\fi

\usepackage{titlesec}
\titlespacing*{\section}{0pt}{0\baselineskip}{0\baselineskip}

\input{estimands_macros}
\input{nrec_macros}
\input{language_macros}
\input{perceived_race_macros}

\newtheoremstyle{exampstyle}
{10pt} %
{10pt} %
{\onehalfspacing} %
{} %
{\bfseries} %
{.} %
{ } %
{} %
\theoremstyle{exampstyle}
\newtheorem{example}{Example}
\begin{document}

\addtocontents{toc}{\protect\setcounter{tocdepth}{-10}}

\begin{refsection}
\begin{titlepage}
\title{Which Effect of Race? Causal Inference without Holding All Else Equal}
\ifnum\blind=1
  \author{}
\else
  \author{Thomas Leavitt}
\fi
\date{}
\maketitle
\abstract{\noindent \input{abstract-text}}
\end{titlepage}

\doublespacing

\section{Introduction}

Whether people face discrimination on the basis of race --- as well as on other bases such as religion and ethnolinguistic identity --- is among the oldest and most contested questions in political life. Consider the long-running debate over discrimination against Muslims in France: Would a French Muslim receive a colder response than a French Christian to a job application or a request to a public office? As it figures in public debate, the comparison sets a Muslim of Maghrebi origin --- along with associated traits such as Arabic language and others that characterize a ``real'' Muslim in popular French perception \citep{diop1988,adidaetal2010} --- against a Christian of European origin who typically bears none of these traits.

The standard approach is to pull these features apart, treating geographic origin and its associated traits as confounders of what \citet[][p.~22386]{adidaetal2010} call the ``Muslim effect.'' Their correspondence study in the French labor market holds the applicant's Senegalese origin fixed while varying the religion the r\'{e}sum\'{e} signals, and finds that fewer employers would respond positively to the ostensibly Muslim applicant. Yet isolating religion in this way excludes the groups at the center of French public debate --- Maghrebi-origin Muslims and European-origin Christians --- putting in their place a Senegalese Muslim and a Senegalese Christian, each atypical of its category as it figures in that debate. Which effect of religion, then, is the one to study: the effect of religious affiliation separable from origin, language, and other traits, or the effect of the category ``Muslim'' with the traits the category indexes?

The same all-else-equal move --- erasing rather than preserving differences in attributes across categories --- underwrites much of what we know about the effects of race and ethnicity. As \citet[][p.~512]{senwasow2016} observe, ``\textit{all} studies employing exposure to a racial or ethnic signal share a common experimental design,'' one that ``present[s] a subject with information that differs only with respect to signals or cues about race or ethnicity.'' This design ranges from the canonical audit and correspondence studies of hiring \citep{bertrandmullainathan2004,pager2003,quillianetal2017} to the conjoint experiments now common in political science \citep{hainmuellerhopkins2015}. Because the literature defends the all-else-equal design as what credible inference requires \citep[see, e.g.,][]{heckman1998,butlercrabtree2021}, the choice of estimand embedded in the design is rarely recognized as a choice at all.

This defense of the all-else-equal design informs a practice that genuinely removes confounding but, in the same act, silently fixes \textit{which} estimand a study recovers (\Cref{sec: differential exposure}). The literature's recovery conditions thus bundle two claims this paper separates: an estimand commitment, which concerns the contrast a study should target, and a recovery condition, which concerns whether a design recovers that contrast. Once the two claims are separated, the all-else-equal contrast defines only one member of a broader family of counterfactual race effects --- including members that preserve the race-conditional structure of nonracial attributes. A single randomization recovers every member of that family of estimands. The choice among that family's members is then substantive rather than statistical: Differences in nonracial attributes ``may not be \textit{threats to inference}, but, rather, an important part of the \textit{story of racial discrimination}'' \citep[][p.~2, emphasis added]{crabtreeetal2023}. What has been lacking is a framework that makes these claims precise.

This paper supplies one, consisting of three parts. First, I derive a family of precisely defined race estimands. Second, I develop substantive grounds for choosing among the estimands. Third, I provide conditions, weaker than the literature assumes, under which any member of the estimand family can be recovered with unbiased or consistent estimation and valid inference.

The framework applies in either of the two dominant regimes for studying discrimination, \textit{exposure-based} and \textit{perception-based} \citep[][pp.~259--260]{hukohler-hausmann2025}. In the exposure-based regime, race is a fixed feature of what the decision-maker encounters --- a person or a racial signal --- regardless of what the decision-maker perceives \citep[e.g.,][]{piper1992}. For example, an officer may encounter a Black or a White civilian whose race remains fixed even if the officer cannot perceive (or misperceives) the civilian's race, as under a veil of darkness \citep{groggerridgeway2006,piersonetal2020}. Likewise, an employer may encounter a r\'{e}sum\'{e} whose name the researcher conceives as a Black or a White signal, whether or not the signal registers with the employer. In the perception-based regime, race is instead the racial interpretation the decision-maker forms, so the regime requires a cue to induce the interpretation --- the r\'{e}sum\'{e}'s name, say --- and the racial variable is the cue's uptake, not the cue.

The family of race estimands is bounded by two extremes, with mixtures between them. At one extreme, the All-Else-Equal Effect (AEE) --- of which the average marginal component effect (AMCE) from conjoint experiments \citep{hainmuelleretal2014} is a special case --- forces the distribution of nonracial attributes to be identical across racial conditions. At the other, the All-Else-Within-Race Effect (AEWR) lets that distribution differ across racial conditions in ways that reflect the social structure under study. The mixtures hold some attributes fixed and let others vary by race. Every member of the estimand family is built from the same unit-level race effect --- a counterfactual contrast for a fixed decision-maker, not a difference across people (\Cref{sec: formal setup}) --- and the members differ only in which features are held fixed across racial conditions and how the remaining features are weighted, specifications the researcher sets under the exposure-based regime. No member of the estimand family is privileged independently of the substantive question.

Under the perception-based regime, the cue's uptake, rather than the researcher, determines which nonracial perceptions stay fixed and which move with perceived race: A cue that would shift a decision-maker's perceived race may shift nonracial perceptions with it, so the effect the data recover is whichever member of the estimand family those joint movements compose. Insofar as moving those perceptions together is partly what it means for a cue to be racialized, as \citet{hukohler-hausmann2025} argue, the member the joint movements compose is not a failed manipulation but the effect of interest --- the perception-based counterpart of the within-race effect. Although much research seeks cues that shift perceived race and nothing else \citep{dafoeetal2018,landgraveweller2022,elderhayes2023}, this paper shows that an isolated shift in perceived race need be neither a condition for credible inference nor a privileged target.

The framework as a whole enters a debate at the center of the empirical literature on discrimination: whether, and how, a counterfactual effect of race can be defined and recovered. Two prominent accounts, one explicitly and one by presupposition, proceed from a constructivist conception of race on which the category consists in nominal membership together with the nonracial attributes the category indexes. Both hold one further premise: A counterfactual effect of race must hold all else equal. Granting the premise, one response studies one racial cue at a time and sets aside the effect of the category as constituted \citep[e.g.,][]{senwasow2016}; the other retains the category and objects that an all-else-equal contrast, by holding the indexed attributes equal, cannot be an effect of the category as constituted \citep{hu2023,hukohler-hausmann2025}, an objection whose strongest form dispenses with the potential outcomes framework for detecting discrimination \citep{kohler-hausmann2019}. This paper's contribution to that debate is to show that the premise is a choice rather than a requirement of credible inference (\Cref{sec: debate}): Because membership in a racial category is fixed independently of the attributes the category indexes, contrasts that hold nothing else equal are well-defined effects of race. One can therefore keep the category as constituted without abandoning causal inference, and keep causal inference without reducing the category to a cue.

I illustrate the framework by reanalyzing a candidate-evaluation experiment \citep{zarateetal2024} under both regimes, asking whether Hispanic voters reward a coethnic candidate --- a form of discrimination that favors rather than penalizes. Read through the exposure-based regime, the all-else-equal effect is indistinguishable from zero, while the within-race effect is positive and distinguishable from zero. On one experiment, then, the choice of estimand is the difference between finding coethnic preference and finding none, a divergence that reflects the different questions the two estimands answer rather than any artifact of estimation. Read through the perception-based regime, which member of the estimand family the data recover is no longer the researcher's to choose; auxiliary name-rating data \citep{elderhayes2023} suggest the recovered effect is unlikely to be the all-else-equal member, though the recovered effect remains a well-defined estimand of discrimination either way.

In \Cref{sec: conclusion}, I close with three questions that emerge once the choice of estimand is recognized as a choice: (i) how normative theories of discrimination might name the contrast a study should target; (ii) how a study can report across the family of estimands or instead adopt and defend a single member; and (iii) what it means to represent a racial group when some individuals belong to both that group and another defined by, say, income, so that, unlike race and income on a r\'{e}sum\'{e}, which a study can pair freely, the groups' preferences cannot move independently. Proofs and additional formal material appear in the Supplementary Material.

\section{Must an Effect of Race Hold All Else Equal?} \label{sec: debate}

Following arguments that an individual's race cannot be manipulated \citep{holland1986,holland2001,holland2003,zuberi2001,dinardo2007}, studies of discrimination define the effect of race not on an individual conceived as someone who could be, say, either Black or White, but on a decision-maker who could encounter either an individual racialized as Black or a different individual racialized as White \citep[following][]{greinerrubin2011}. The effect is the difference in how the decision-maker would respond to each individual, whether the encounter is direct, through a r\'{e}sum\'{e}, or through another representation of the individual.

Defining the effect through a decision-maker's responses to two differently racialized individuals sidesteps the manipulation of race but still depends on a conception of race itself: who counts as Black or as White. Among the many conceptions of race and ethnicity, nearly all agree that the categories rest on membership classifications justified in terms of \textit{descent}, whether grounded in biological essence or social convention \citep{chandra2006,chandra2012,andreasen1998,andreasen2000,weber1978,fearon2003,hardimon2003,hardimon2017,spencer2014,appiahgutmann1998,mills1998,cornellhartmann1997,haslanger2000,taylor2003}.

On a classical conception of categories, membership exhausts what a category is \citep[][pp.~9--10]{monk2022}: A racial category consists only in nominal membership as the descent-based classifications assign it, so a comparison of a Black and a White individual identical in nonracial features can omit nothing the category is. By contrast, on what \citet{hu2023} calls a ``thick constructivist'' conception of race, a racial category consists in nominal membership together with the nonracial attributes that membership probabilistically indexes --- the category's social content \citep{haslanger2000,taylor2003}. Race, so conceived, is not reducible to nominal membership alone.

Insofar as a racial category indexes a distribution of nonracial attributes, some configurations are more typical of the category than others \citep{wittgenstein1953,rosch1973,rosch1975,roschmervis1975}. Typicality increases toward a prototype at the center of the category-conditional distribution --- a graded conception of race already present in empirical work \citep{maetal2015,eberhardtetal2006,maddox2004,ocamporoland2025}. To call a configuration typical of the category is to describe the social world --- the distribution of nonracial attributes within the category, or the expectations society attaches to the category's members --- not to judge how authentically Black or White anyone is \citep{appiahgutmann1998,gooding-williams1998,appiah2005}.

Two prominent recent accounts of race as a cause proceed from such a constructivist conception, one explicitly and one by presupposition. \citet{kohler-hausmann2019} embraces the conception explicitly, as do \citet[][p.~12]{hu2023}, for whom the marking of racial categories runs through ``phenotype and supposed ancestry,'' and \citet[][p.~250]{hukohler-hausmann2025}, for whom racial and nonracial features ``co-constitute the category of race.'' \citet{senwasow2016} presuppose the conception: Their designs exploit within-group variation, and a single ``stick'' serves as a proxy that ``conveys information about race'' \citep[][p.~509]{senwasow2016} --- formulations that require a racial category distinct from, and fixed independently of, any one cue. On either account, the descent that grounds membership --- actual or presumed --- can stay fixed when neighborhood or record moves, so who counts as Black or White can hold still while the attributes vary. In the image of race as a bundle of sticks \citep{senwasow2016}, the attributes are the sticks, and nominal membership is the string that ties them into a bundle (\Cref{sec: substantive motivations}).

Yet from this shared conception a division emerges over what the two racialized individuals must differ in, and what they may share, for the contrast between them to count as an effect of race. The division's two sides move in opposite directions, yet rest on a single premise: A causal effect in the potential outcomes framework must be an all-else-equal contrast. Grant the premise, and the category as constituted admits no effect defined in terms of potential outcomes: An effect of the category would let the indexed attributes move with membership, and that movement is precisely not holding all else equal. As \citet[][p.~1171]{kohler-hausmann2019} concludes, ``race cannot be conceptualized as an isolated treatment in the counterfactual causal model, and accordingly, racial discrimination cannot be defined as the treatment effect of race.''

One response --- the cue-effects response --- retains potential outcomes and studies one cue at a time, holding the remaining attributes fixed, while setting aside the effect of the category as constituted \citep{senwasow2016}. The other --- the incompatibility response --- retains the category as constituted and denies that a contrast holding the indexed attributes equal --- the all-else-equal contrast --- can be an effect of race \citep{kohler-hausmann2019,kohler-hausmanndembroff2022,hu2023,hukohler-hausmann2025}; \citet{dembroffetal2020} make the same denial for sex. The denial has resulted in two conclusions. One dispenses with the framework for detecting discrimination on the grounds that ``[t]he counterfactual causal model of discrimination is based on a flawed theory of what the category of race references'' \citep[][p.~1163]{kohler-hausmann2019}, a conclusion \citet{dembroffetal2020} also reach. The other holds that positive inquiry into the effects of race cannot be separated from normative inquiry, a claim the accounts advance with varying scope \citep{dembroffetal2020,kohler-hausmanndembroff2022,hukohler-hausmann2025,hu2026}.

I challenge the premise that the cue-effects response and the incompatibility response share. As the sections below show, a contrast in which the nonracial attributes move with racial membership across the two individuals is a well-defined effect in the potential outcomes framework. \citet[][p.~253]{hukohler-hausmann2025} come closest to relaxing the premise, expressing one such contrast within potential outcomes, though as the estimand the constructivist conception calls for rather than as one member of a family of estimands the researcher can choose among. Once the premise falls, the two responses are no longer exhaustive alternatives: A study can infer cue effects in the sense of the cue-effects response and effects of the category as constituted in the sense of the incompatibility response. Which effect the study targets becomes a substantive choice rather than a consequence of the premise.

\section{Formal Setup} \label{sec: formal setup}

I now formalize the common structure of studies of discrimination and, within it, the two regimes introduced above. I take the exposure-based regime as the baseline, building through it the shared apparatus --- the configuration space, the potential-outcome function on it, and the family of race effects; for that reason the exposure-based regime has no subsection of its own here. The perception-based regime, treated separately (\Cref{sec: perception-based}) because it recasts the configuration as itself a potential outcome of a cue, borrows that apparatus intact. Later sections, which treat the two regimes in parallel, give each regime its own heading. 

I develop the apparatus on a running example of prosecutorial charging --- a setting the causal inference literature itself studies \citep{gaebleretal2022} and the case its critics treat as paradigmatic \citep{hukohler-hausmann2025}. A prosecutor reads an arrest file and decides whether to file charges. The file records the arrestee's race --- Black or White --- or only cues of race, such as the arrestee's name, which may induce the prosecutor's perception of the arrestee as one or the other. Alongside the stated race or the cue, the file contains two nonracial features: the location of the arrest and the extent of the arrestee's prior record \citep[][pp.~256--257]{hukohler-hausmann2025}. Neither feature is neutral background: Black, relative to White, indexes a greater likelihood of residing in --- and so of being arrested in --- neighborhoods like Brownsville, Brooklyn, a low-income, largely Black neighborhood \citep[][p.~256]{hukohler-hausmann2025}, and a greater likelihood of an extensive prior record, if only because such neighborhoods are subject to greater policing.

The example thus confronts the researcher with a choice. Holding location fixed compares a Black and a White arrestee both arrested in, say, Brownsville --- a shared location that makes \enquote{the Black arrestee `expected'} and \enquote{the white arrestee `unexpected'} \citep[][p.~257]{hukohler-hausmann2025}. Yet an \enquote{unexpected} White arrestee is still White. Who counts as White holds fixed while an attribute the category indexes varies. Letting location vary instead compares the arrestees at locations more likely within each arrestee's category: say, a Black arrestee in Brownsville and a White arrestee on the Upper East Side, a high-income, largely White neighborhood. An analogous choice arises for the prior record.

\subsection{Decision-Makers, Decisions, and Configurations}

There is a set of decision-makers, prosecutors in the running example, and each faces one or more decisions; indexing the resulting decision-maker--decision pairs as \textit{units} $i \in \{1, \ldots, N\}$, a unit is one prosecutor reading one case file (\Cref{app: units sutva} gives the explicit indexing by decision-makers and decisions). The pair, rather than the prosecutor, is the unit because the same prosecutor need not respond alike across decisions (\Cref{sec: differential exposure}). Under the exposure-based regime, unit $i$ encounters a \textit{configuration} $\bm{t}_i \coloneqq (z, \bm{v})$: the case file, formally. The feature $z \in \{0, 1\}$ is a \textit{racially constituted attribute} --- a feature specified in terms of race, such as the race the file states or a name conceived as a racial signal. The remaining \textit{nonracial features} $\bm{v} \coloneqq (v_1, \ldots, v_L)$, such as the prior record and the neighborhood, take values in $\mathcal{V} \coloneqq \mathcal{V}_1 \times \cdots \times \mathcal{V}_L$, where each $\mathcal{V}_\ell$ is finite, and their joint distribution with $z$ may itself be part of how race manifests in the social world. Each configuration $\bm{t}_i$ thus lies in the common space $\mathcal{T} \coloneqq \{0, 1\} \times \mathcal{V}$.

The set $\mathcal{T}$ is a product space --- either racial value can be paired with any nonracial profile --- and this pairing lets every member of the estimand family be defined. What licenses the pairing is the distinction between nominal membership and the attributes that membership indexes: For the racial categories this paper treats, membership runs through descent, actual or presumed, a criterion outside the nonracial profile (\Cref{sec: debate}). A file that pairs one racial category with a profile typical of the other is therefore atypical, not ill-defined. Which nonracial features compose $\bm{v}$ is a substantive choice the researcher makes in defining the estimands, not a claim about measurement. The framework takes the set as given and does not derive it.

\subsection{Potential Outcomes and the Race Effect} \label{sec: PO}

In principle, a unit's potential outcome could depend on the configurations that other units encounter --- including the other files the same prosecutor reads. A stable unit treatment value assumption (SUTVA), stated formally in \Cref{app: units sutva}, rules out such interference and carryover: The potential outcome for unit $i$ depends only on unit $i$'s own configuration, defining a function $y_i: \mathcal{T} \to \R$, where $\R$ is the set of real numbers, with $y_i(\bm{t}_i) = y_i(z, \bm{v})$ for $\bm{t}_i = (z, \bm{v}) \in \mathcal{T}$. In the running example, the function $y_i$ is the charging response unit $i$'s prosecutor would give to each possible file.

A unit-level effect is a contrast between two configurations:
\begin{align*}
\tau_i(\bm{t}, \bm{t}^{\prime}) \coloneqq y_i(\bm{t}) - y_i(\bm{t}^{\prime}), \qquad \bm{t}, \bm{t}^{\prime} \in \mathcal{T}.
\end{align*}
For race effects, the two configurations differ in the value of $z$. Writing $\bm{t} = (1, \bm{v})$ and $\bm{t}^{\prime} = (0, \bm{v}^{\prime})$, the \textit{race effect} is
\begin{align} \label{eq: race effect}
\tau_i\left((1, \bm{v}), (0, \bm{v}^{\prime})\right) = y_i(1, \bm{v}) - y_i(0, \bm{v}^{\prime}), \qquad \bm{v}, \bm{v}^{\prime} \in \mathcal{V}.
\end{align}
Because $\bm{v}$ and $\bm{v}^{\prime}$ each range over $\mathcal{V}$, the quantity in \eqref{eq: race effect} defines a family of effects, one for each pair $(\bm{v}, \bm{v}^{\prime})$ of nonracial profiles. The all-else-equal contrast is the special case $\bm{v} = \bm{v}^{\prime}$; the opposite extreme leaves $\bm{v}$ and $\bm{v}^{\prime}$ free to vary across racial conditions. In the running example, setting $\bm{v} = \bm{v}^{\prime}$ compares the Black and the White arrestee at one shared prior record and neighborhood, while setting $\bm{v} \neq \bm{v}^{\prime}$ lets the two files differ not only in race but also in prior record and neighborhood.

\subsection{Perception-Based Regime} \label{sec: perception-based}

Return once more to the running example, and take the file in its other form, recording race only through the arrestee's name --- DeShawn or Connor. Suppose the file contains nothing else, so that neighborhood and prior record, like race, are left for the prosecutor to infer. The arrestee's own race is fixed however anyone reads the file, but what now drives the decision is the race the prosecutor perceives, which need not match the arrestee's own, and the same cue may move nonracial perceptions too; for some files the name would move perception not at all. The perception-based regime formalizes this reading, treating the perceived race and features as potential outcomes of the name and confining the race contrast to the files whose perceived race the name would actually shift.

Formally, the manipulation is a \textit{cue} $w_i \in \{0, 1\}$ --- a name \citep{bertrandmullainathan2004}, photo \citep{terkildsen1993}, accent \citep{purnelletal1999}, or other stimulus --- not the configuration itself, with $w_i = 1$ intended to induce $z = 1$ and $w_i = 0$ intended to induce $z = 0$. The configuration is the perception the cue actually would induce, not necessarily the perception the cue is intended to induce. The racial feature, the nonracial features, and the response are therefore all potential outcomes of the cue.

As in the exposure-based regime, a cue-level SUTVA --- stated formally in \Cref{app: units sutva} --- rules out interference, so unit $i$'s response, racial perception, and nonracial perceptions depend only on unit $i$'s own cue. Write $z_i(w)$ and $\bm{v}_i(w)$ for the racial and nonracial components of the perception that cue value $w$ would induce for unit $i$, and define the cue-induced configuration
\begin{align*}
\bm{t}_i(w) \coloneqq \left(z_i(w), \bm{v}_i(w)\right) \in \mathcal{T}.
\end{align*}
The component functions $z_i$ and $\bm{v}_i$ take values in the same spaces as in the exposure-based regime, but here the components are unit-level potential outcomes of the cue. As under the exposure-based regime, the researcher specifies which perceptions compose $\bm{v}$; the cue's uptake, not the researcher, determines the values $z_i(w)$ and $\bm{v}_i(w)$. The configuration realized for unit $i$ is $\bm{t}_i(w_i)$.

The potential response to a cue may in principle depend on the cue through channels other than the cue-induced perception --- through, say, an aesthetic or affective reaction to the cue that operates whatever race and nonracial features the decision-maker would read from the cue. To rule out such direct channels --- channels that would make part of the cue effect an effect of the stimulus itself, a part no effect of perceived race in \eqref{eq: race effect} could represent --- I assume the potential outcome depends on the cue only through that unit's perceived configuration, an assumption that holds only if every channel through which the cue moves the outcome runs through the perceived race or the nonracial perceptions $\bm{v}$.
\begin{assm}[Perception sufficiency] \label{assm: PS}
For all units $i \in \{1, \ldots, N\}$ and all $w, w^{\prime} \in \{0, 1\}$,
\begin{align*}
\bm{t}_i(w) = \bm{t}_i(w^{\prime}) \implies y_i(w) = y_i(w^{\prime}).
\end{align*}
\end{assm}
Together, the cue-level SUTVA and perception sufficiency (\Cref{assm: PS}) do for the perception-based regime what SUTVA alone did for the exposure-based one (\Cref{sec: PO}): The two assumptions reduce the potential responses to the cue to the same function $y_i$, evaluated at the cue-induced configuration,
\begin{align*}
y_i(w) = y_i\left(\bm{t}_i(w)\right) = y_i\left(z_i(w), \bm{v}_i(w)\right), \qquad w \in \{0, 1\}.
\end{align*}
The unit-level cue effect is thus the contrast between the two configurations induced by the alternative cue values:
\begin{align} \label{eq: cue effect}
\tau_i\left(\bm{t}_i(1), \bm{t}_i(0)\right) = y_i\left(z_i(1), \bm{v}_i(1)\right) - y_i\left(z_i(0), \bm{v}_i(0)\right).
\end{align}

Under the cue-level SUTVA and \Cref{assm: PS}, whether the cue effect in \eqref{eq: cue effect} corresponds to a race effect in \eqref{eq: race effect} depends on how the cue changes perceived race. Following \citet{frangakisrubin2002}, classify units into principal strata by $(z_i(0), z_i(1))$: \textit{compliers} $(0,1)$, \textit{always-takers} $(1,1)$, \textit{never-takers} $(0,0)$, and \textit{defiers} $(1,0)$. For always-takers and never-takers, the cue would induce no race contrast because $z$ would not change across cue values. For defiers, the cue would move $z$ opposite the intended direction, so the cue effect equals the negative of a race effect in \eqref{eq: race effect}, evaluated at the perceptions the cue would induce. The intended directions rest on shared conventions of racial signification --- the coding that makes DeShawn read as Black and Connor as White --- so a defier would read the coding in reverse, against the conventions the cue's design presupposes. For compliers, by contrast, the cue would induce the race contrast in the intended direction.

Let $\mathcal{C} \coloneqq \{i : z_i(0) = 0, z_i(1) = 1\}$ denote the set of compliers. On $\mathcal{C}$, the cue effect equals the race effect at the configurations $\bm{t}_i(0) = (0, \bm{v}_i(0))$ and $\bm{t}_i(1) = (1, \bm{v}_i(1))$ induced by the cue:
\begin{align} \label{eq: complier race effect}
\tau_i\left(\bm{t}_i(1), \bm{t}_i(0)\right)
= y_i\left(1, \bm{v}_i(1)\right) - y_i\left(0, \bm{v}_i(0)\right), \qquad i \in \mathcal{C},
\end{align}
where $\bm{v}_i(0), \bm{v}_i(1) \in \mathcal{V}$ are the nonracial perceptions that would be induced under each cue value. Unlike in \eqref{eq: race effect}, where the researcher can in principle set the values of $\bm{v}, \bm{v}^{\prime} \in \mathcal{V}$, the perception-based regime allows the researcher to set only the value of the cue: Perceived race moves from $0$ to $1$ by the definition of $\mathcal{C}$, while the nonracial perceptions $\bm{v}_i(0)$ and $\bm{v}_i(1)$ are governed by the decision-maker's potential responses to that cue.

\section{Anatomy of an Estimand} \label{sec: anatomy estimand}

An \textit{estimand} collapses the family of unit-level race effects in \eqref{eq: race effect} into a scalar through three components: A \textit{contrast} holds certain nonracial features fixed across racial conditions and lets the others vary by race; a \textit{weighting} on $\mathcal{V}$, either common across racial conditions or race-specific, weights the nonracial profiles; and a \textit{subset} restricts which units enter the average. Each such estimand is therefore a linear functional of the potential outcomes --- a difference-scale weighted average of unit-level race effects.

The weighting is either a distribution $\rho$ on $\mathcal{V}$ common across racial conditions --- which weights nonracial features without regard to race --- or race-specific distributions $q_z$ on $\mathcal{V}$ --- which weight nonracial features conditional on each racial condition $z \in \{0,1\}$. Both can be derived from a single joint PMF on $\mathcal{T}$, $q(z, \bm{v}) \coloneqq \Pr\left((Z, \bm{V}) = (z, \bm{v})\right)$, via the marginal $\rho(\bm{v}) \coloneqq \sum_z q(z, \bm{v})$ and, where the denominator is positive, the conditional $q_z(\bm{v}) \coloneqq q(z, \bm{v}) / \sum_{\bm{v}^{\prime} \in \mathcal{V}} q(z, \bm{v}^{\prime})$.
These distributions are part of the estimand's definition --- quantities the researcher specifies, not a model of how the data were generated or an empirical distribution measured in the study. Because the researcher places these weights, a distribution may concentrate on any subset of $\mathcal{V}$, and the researcher can put little or no weight on a profile of no substantive interest. I take the weighting to be common across units (not to be confused with common across racial conditions); unit-specific distributions would require only added notation.

\subsection{The Estimand: Contrast, Weighting, and Subset} \label{sec: choosing contrast}

\subsubsection{Exposure-Based Regime}

Specifying an estimand begins with the contrast. Contrasts can be ordered by how much of $\bm{v}$ they hold fixed across racial conditions, running between two extremes. The \textit{all-else-equal} contrast holds every feature fixed and the \textit{within-race} contrast holds none; \textit{mixed} contrasts, holding some features fixed while leaving others free to vary by race, lie in between. Across all three contrasts, I average over all $N$ units, though the framework allows narrower subsets.

\paragraph{All-else-equal (AEE) effect.} The all-else-equal contrast sets $\bm{v} = \bm{v}^{\prime}$, so the unit-level race effect at each $\bm{v} \in \mathcal{V}$ is $y_i(1, \bm{v}) - y_i(0, \bm{v})$. Because both conditions sit at that shared profile, a single distribution $\rho$, common across racial conditions, supplies the weights. Weighting these profile-specific effects by $\rho$ and averaging over all $N$ units yields the \textit{All-Else-Equal Effect} (AEE),
\begin{align} \label{eq: AEE}
\textrm{AEE} \coloneqq \dfrac{1}{N} \sum_{i = 1}^N \sum_{\bm{v} \in \mathcal{V}} \left[y_i(1, \bm{v}) - y_i(0, \bm{v})\right] \rho(\bm{v}),
\end{align}
which, with the weighting common across units, coincides with the average marginal component effect (AMCE) of conjoint experiments \citep{hainmuelleretal2014}.

\paragraph{All-else-within-race (AEWR) effect.} The within-race contrast leaves $\bm{v}$ and $\bm{v}^{\prime}$ free to differ, so each side of the racial contrast is weighted by its own race-specific conditional $q_z$ --- the distribution of $\bm{v}$ within racial condition $z$ --- and averaging over all $N$ units yields the \textit{All-Else-Within-Race Effect} (AEWR),
\begin{align} \label{eq: AEWR}
\textrm{AEWR} \coloneqq \dfrac{1}{N} \sum \limits_{i = 1}^N \left[\sum_{\bm{v} \in \mathcal{V}} y_i(1, \bm{v})\, q_1(\bm{v}) - \sum_{\bm{v}^{\prime} \in \mathcal{V}} y_i(0, \bm{v}^{\prime})\, q_0(\bm{v}^{\prime})\right].
\end{align}

The AEE and the AEWR reverse the order in which they difference and average: The AEE differences the racial conditions at a shared profile and then averages the within-profile differences over the common $\rho$, while the AEWR averages each racial condition over its own $q_z$ and then differences the race-specific averages, with no shared profile anchoring the comparison. %

The two estimands differ except in special cases: The difference between the AEWR and the AEE vanishes when $q_1 = q_0$, so that each racial condition's own distribution is already the common one, or when the average potential outcome $\bar{y}(z, \bm{v}) \coloneqq (1/N) \sum_{i = 1}^N y_i(z, \bm{v})$ is constant in $\bm{v}$ for each racial condition $z$, so that every distribution returns the same race-specific average. \Cref{prop: gap} in the Supplementary Material records the gap between any two members of the estimand family, mixed effects included, and shows that the same logic determines when the members differ.

\paragraph{Mixed effects.} The choice of contrast need not be all-or-nothing. Consider racial discrimination in officer use of force \citep{fryer2019,knoxetal2020}: Officers are the decision-makers, civilians bear the configurations, and the researcher may want to use a legal standard to sort the features of $\bm{v}$. Conduct such as openly carrying a weapon bears on that standard and may be held fixed, while dress and neighborhood, no legally justified basis for force, may be left free to vary by race.

A mixed contrast divides the nonracial features into a fixed set and a free set. The fixed features are held equal across racial conditions, while the free features vary by race and are averaged over within each racial condition. The weighting follows the partition: A common distribution weights the fixed features, as $\rho$ weights all of $\bm{v}$ for the AEE, and race-specific distributions weight the free features given the fixed profile, as the $q_z$ do for the AEWR. In the use-of-force example, the two racial conditions share one weighting over whether the civilian carries a weapon, while dress and neighborhood follow each racial group's own distribution given that conduct. A mixed effect is thus all-else-equal on the fixed features and within-race on the free ones; \Cref{sec: mixed contrast formalism} of the Supplementary Material gives the formal partition and the resulting partially marginalized average effect.

\subsubsection{Perception-Based Regime}

Perceived configurations lie in the same space $\mathcal{T}$ as in the exposure-based regime (\Cref{sec: perception-based}), so the cue effect belongs to the same family of unit-level race effects. What the perception-based regime withholds is the choice among that family's members: Under the cue-level SUTVA (\Cref{app: units sutva}), each cue value induces a single perceived configuration $(z_i(w),\bm{v}_i(w))$ for each unit. For compliers, the cue fixes the nonracial perception at whatever value of $\bm{v}_i(w)$ the cue would induce. This degenerate distribution is common across racial conditions when $\bm{v}_i(0) = \bm{v}_i(1)$ and race-specific on at least some nonracial features otherwise.

The cue effect in \eqref{eq: cue effect} corresponds to the race effect in \eqref{eq: race effect}, now read as the effect of perceived race, only for the compliers $\mathcal{C}$ (\Cref{sec: perception-based}). For compliers, and for compliers alone, the cue would move perceived race from $0$ to $1$ --- the contrast in \eqref{eq: race effect}. The subset therefore narrows from all $N$ units to those in $\mathcal{C}$, and under perception sufficiency (\Cref{assm: PS}) the average cue effect on $\mathcal{C}$ is what I call the \textit{natural race effect among compliers} (NREC),
\begin{align} \label{eq: complier effect}
\textrm{NREC} \coloneqq \dfrac{1}{\abs{\mathcal{C}}} \sum_{i \in \mathcal{C}} \left[y_i\left(1, \bm{v}_i(1)\right) - y_i\left(0, \bm{v}_i(0)\right)\right],
\end{align}
where $\abs{\mathcal{C}}$ is the number of compliers.

The NREC averages the mix of contrast types the cue induces among compliers. If $\bm{v}_i(0) = \bm{v}_i(1)$ for a complier, that unit contributes an all-else-equal contrast; if the cue changes every nonracial feature, a within-race contrast; and if the cue changes some nonracial features but not others, a mixed contrast. Only if $\bm{v}_i(0) = \bm{v}_i(1)$ for every $i \in \mathcal{C}$ does the NREC reduce to the all-else-equal estimand among compliers, the standard target in the perception-based literature \citep{dafoeetal2018,landgraveweller2022,elderhayes2023}. As the introduction argues, the co-movement does not disqualify the NREC as an effect of race: Each induced contrast is a unit-level race effect in \eqref{eq: race effect}, and insofar as shifting nonracial perceptions together with perceived race is partly what it means for a cue to be racialized \citep{hukohler-hausmann2025}, the NREC is not a failed manipulation but the effect of interest.

The label \textit{natural} follows mediation usage \citep{pearl2001,robinsgreenland1992}, under which a mediator takes the value the treatment would induce. The simultaneous ordering matches how \citet[][p.~251]{hukohler-hausmann2025} describe reading a cue racially: activating the category together with the nonracial attributes the category indexes. %

\begin{figure}[H]
\centering
\includegraphics[width=\linewidth]{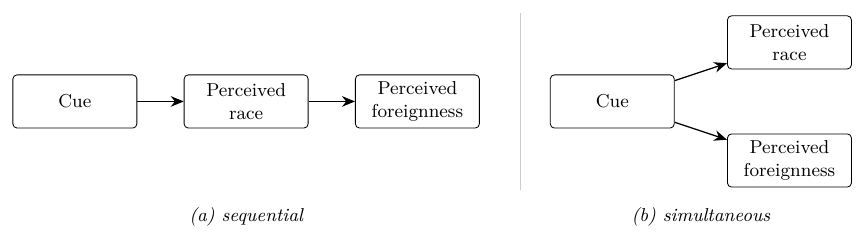}
\caption{Sequential versus simultaneous mediation, illustrated with perceived foreignness: In panel \textit{(a)}, the cue would move the nonracial perception only through perceived race; in panel \textit{(b)}, alongside perceived race.}
\label{fig: mediation ordering}
\end{figure}

Note that, although the simultaneous ordering is the one \citet{hukohler-hausmann2025} describe, the sequential ordering is no less compatible with the thick constructivist conception: Inference from perceived race to nonracial attribute can be one way the category's indexing operates \citep{vanderweele2015}. The orderings differ instead in where the movement sits (\Cref{fig: mediation ordering}). Downstream of perceived race, the inferred attribute moves as part of the effect of race itself, so nothing co-moves that an all-else-equal contrast must hold fixed.

\section{All-Else-Equal as a Substantive Commitment} \label{sec: substantive motivations}

Restricting the estimand family of \Cref{sec: anatomy estimand} to the all-else-equal contrast involves two tradeoffs. The first is conceptual: On some conceptions of race, holding nonracial attributes fixed across racial categories may strip away part of the category's social content. The second is structural: Insofar as the distribution of nonracial attributes differs by racial category, no single weighting over profiles can simultaneously represent the distribution of those attributes within each racial group. Neither tradeoff condemns the all-else-equal contrast; together they make adopting it a substantive decision rather than a default.

\subsection{The Conceptual Tradeoff}

On the classical conception (\Cref{sec: debate}), nominal membership exhausts the category, so a contrast that moves membership while holding the indexed attributes fixed compares the categories whole. On the thick constructivist conception, the same contrast compares the categories stripped of their social content. With the indexed attributes held fixed, what varies across the two individuals is nominal membership alone --- the string without the bundle. The contrast can therefore refer to the effect of membership as such \citep{hu2023}, but not to the effect of the category as constituted \citep{hukohler-hausmann2025}. On either conception, the contrast is well-defined; the conceptions differ only in what the contrast gives up.

The conception sets what the contrast costs, not whether to adopt it. Even on the thick constructivist conception, a study may hold the indexed attributes fixed for other reasons: because the study seeks the mechanisms through which race exerts an effect, or because differences in some attributes normatively justify differential treatment. The study then pays the cost deliberately rather than by default (\Cref{sec: contrast substantive}).

\subsection{The Structural Tradeoff}

Even if the goal is only to represent each racial group well, the all-else-equal contrast requires one weighting over nonracial profiles applied to both racial categories. When the two groups' distributions of nonracial profiles overlap only slightly, as the arrestees' neighborhoods and prior records do in the running example, the common weighting faces a tradeoff. Confined to the small region of overlap, the weighting captures only a sliver of each group rather than the profiles where most group members lie. Spread out to represent both groups more faithfully, the weighting must place mass on profiles improbable under one group's distribution or the other's. The race-specific weightings $q_z$ escape the tradeoff, evaluating each racial condition where its own mass lies, though the escape costs the common distribution and, with it, the all-else-equal contrast. \Cref{fig: structural tradeoff} depicts the two choices of common weighting, together with the race-specific weighting, in the attribute space of the running example; \Cref{prop: tradeoff} in \Cref{sec: tradeoff} formalizes the structural tradeoff.

\begin{figure}[H]
\centering
\includegraphics[width=\linewidth]{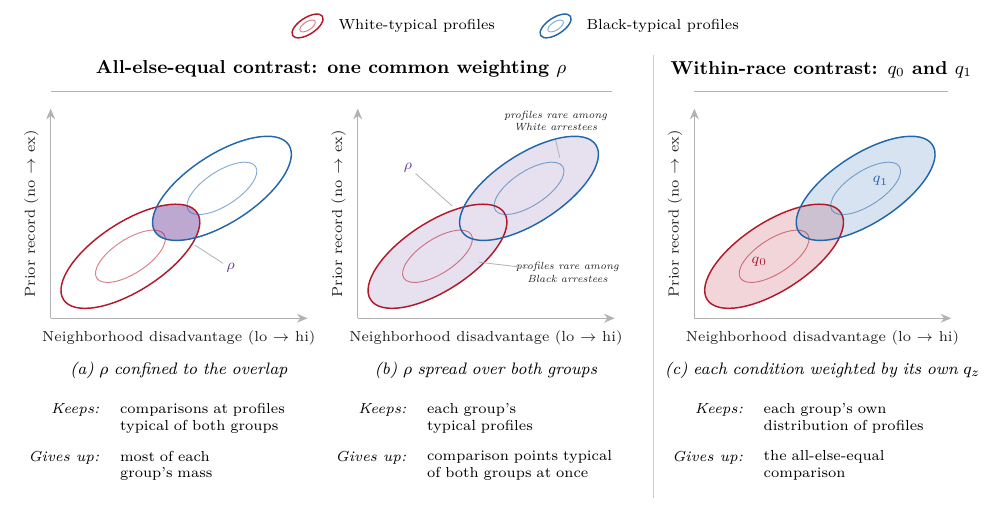}
\caption{The structural tradeoff in the attribute space of the charging example. Each closed curve is a contour of a group's distribution of nonracial profiles (outer curve: at least a $1-\varepsilon$ share of the group's mass, \Cref{prop: tradeoff}; inner curve: the denser core), and shading marks the support of the weighting. Panels \textit{(a)} and \textit{(b)} show the two choices of common weighting $\rho$; panel \textit{(c)}, the race-specific weightings $q_z$.}
\label{fig: structural tradeoff}
\end{figure}

\subsection{Grounds for Choosing the Contrast} \label{sec: contrast substantive}

A classical conception of the racial category points to the all-else-equal effect, and a thick constructivist conception points to the within-race effect. Yet even on that conception, whether race indexes a feature does not alone settle whether to hold the feature fixed. The further grounds for the contrast sort the features on criteria of their own. A study that embraces the thick constructivist conception may therefore reasonably hold fixed attributes that race indexes, yielding a mixed or all-else-equal contrast (\Cref{sec: choosing contrast}). The grounds that follow, for choosing first the contrast and then the weighting, are illustrative rather than exhaustive, and several may apply in one study; \Cref{tab: grounds} in the Supplementary Material collects the grounds, ordered by which component of the estimand each ground bears on.

\subsubsection{Analyzing Causal Mechanisms}

Under the sequential ordering of \Cref{fig: mediation ordering}, panel \textit{(a)}, decision-makers infer nonracial attributes from race, and a substantial literature asks whether such statistical, as opposed to taste-based, discrimination \citep{arrow1973,becker1957,phelps1972} is less morally objectionable \citep[e.g.,][]{lippert-rasmussen2011,hellman2008,alexander1992}. Apart from those normative debates, identifying the pathways through which racial effects operate is a positive question in its own right, an inquiry that may itself be motivated by a normative commitment to understanding how to reduce discrimination.

Holding an attribute fixed across the racial conditions can block a pathway through which race, or a signal of race, affects the outcome. An inferred attribute, however, is itself a potential outcome of race, so it cannot be held fixed directly; what a design controls is the information the configuration provides. Providing the attribute aims to preempt the inference the decision-maker could otherwise draw by fixing the corresponding belief at the provided value. If successful, the design achieves ``perfect'' --- as opposed to imperfect --- ``manipulation of the mediator'' \citep[][p.~359]{acharyaetal2018}. The racial contrast at that value then blocks the mediating path and isolates the remaining channels \citep{vanderweelerobinson2014} --- isolation that implies neither that the ``true'' effect of race is what remains once one blocks particular pathways nor that any one pathway is a privileged effect of race. For example, by stating a partisan signal alongside a racially distinctive name in constituency service requests to state legislators --- who may favor copartisans \citep{fenno1978} --- \citet{butlerbroockman2011} supply the party identification that legislators could otherwise infer from race \citep{mcdermott1998}. %

\subsubsection{Normative and Legal Commitments}

In much empirical work, normative commitments about when a decision-maker is justified in discriminating are often implicit \citep{kohler-hausmann2019,hukohler-hausmann2025,hu2026}. The example of racial discrimination in officer use of force (\Cref{sec: choosing contrast}) made one such commitment explicit by sorting the conduct that may justify force from the features that may not. The study of racial disparities in health care, an important area of inquiry in political science and related fields \citep{michener2018,michener2020,olveraetal2023,groganpark2017}, makes such commitments systematic, building them into the definition of disparity itself.

Drawing on a conception of certain disparities as ``avoidable, unfair, and unjust'' \citep[][p.~107]{whitehead1992}, the landmark report of the Institute of Medicine (IOM) defined disparity in health care as ``racial or ethnic differences in the quality of health care that are not due to access-related factors or clinical needs, preferences, and appropriateness of intervention'' \citep[][pp.~3--4]{iom2002}. On this standard, differences in care between, say, Black and White patients are justified only insofar as those differences reflect needs or preferences \citep{mcguireetal2006,cooketal2012,cooketal2009a,cooketal2009b}.

The IOM's needs-and-preferences criterion has generated work distinguishing ``allowable'' from ``non-allowable'' sources of racial differences in health care \citep{cooketal2012,duanetal2008,jackson2021}. \citet{mcguireetal2006} implement that distinction by holding health status fixed across racial conditions while leaving race-specific distributions of non-allowable features, such as socioeconomic status, unchanged. The resulting estimand marginalizes over the allowable feature using a distribution common across racial groups and over the non-allowable features using each group's own distribution. In the vocabulary of the mixed contrast (\Cref{sec: choosing contrast}), the normative distinction sorts allowable features into the held-fixed set and the non-allowable features into the free set.

\subsection{Grounds for Choosing the Weighting} \label{sec: weighting substantive}

The choice of weighting is specific to the exposure-based regime since, as \Cref{sec: choosing contrast} explains, in the perception-based regime each complier's configurations are fixed potential responses to the cue, so no $\rho$ or $q_z$ enters the estimand. Like the choice of contrast, how to specify the conditionals $q_z$ and the marginal $\rho$ --- set by the researcher, not estimated from the data (\Cref{sec: anatomy estimand}) --- is a substantive choice that turns on what the comparison is meant to represent. I give two motivations: external validity and typicality.

\subsubsection{External Validity}

A researcher who wants the estimand to speak to a specific target population can calibrate the joint distribution $q(z, \bm{v})$ to match that population's distribution of profiles. For the AMCE, \citet{delacuestaetal2022} calibrate the common marginal $\rho$ to external benchmarks; because the calibrated distribution is common across racial conditions, the reweighting presupposes the all-else-equal contrast. The family of race estimands in \Cref{sec: anatomy estimand} permits the same calibration on the race-conditional side, with each $q_z$ matched to the target population's within-race distribution. \Cref{app: distribution formalism} shows that calibrating each $q_z$ requires reference information about how attributes co-vary within race, which single-attribute marginals do not supply.

\subsubsection{Typicality}

\Cref{sec: debate} described a graded conception of typicality, on which racial categories are organized around prototypes and typicality may vary continuously with perceived similarity to the prototype. The conditional distribution $q_z$ can encode that gradation by assigning greater weight to profiles that are more typical of group $z$. \citet{maetal2015}, for instance, operationalize such typicality with racial-prototypicality scores derived from respondents' photograph ratings.

One way, among many, to translate graded typicality into the conditional distribution $q_z(\bm{v})$ builds on the construction of \citet{gardenfors2000,gardenfors2014}. In this construction, how typical a profile is of group $z$ --- for example, how typically Black or White the profile is --- corresponds to the profile's closeness to the group's prototype, with each nonracial attribute weighted by how diagnostic it is of group membership. Profiles near a Black prototype are therefore more probable under $q_1$, and profiles near a White prototype are more probable under $q_0$. In a name-based audit study, for example, the Black prototype might center on names like DeShawn or Jamal, with Connor or Greg at the periphery. The rate at which typicality declines with distance from the prototype may differ across groups, so DeShawn may register as more atypically White than Greg does atypically Black.

\section{Recovering the Estimand Family without Holding All Else Equal} \label{sec: inferential concerns}

\Cref{sec: substantive motivations} treated the all-else-equal contrast as a choice of estimand; whether data recover a chosen estimand is a separate question. The literature's answer to each recovery problem --- confounding under exposure, excludability under perception --- does double duty, securing recovery while silently restricting the estimand to the all-else-equal contrast. This section shows that, in each regime, a weaker condition secures recovery alone, so the choice of contrast and weighting remains substantive.

\subsection{Confounding} \label{sec: differential exposure}

A unit's race effect contrasts two potential outcomes --- the responses to a Black-marked and to a White-marked file --- and at most one can be observed; nor does the same prosecutor at a later decision supply the other. A docket that fills toward the end of a term, attention that flags over a day, or an intervening high-profile case may each change the prosecutor, so the response either file receives need not equal the response the same file would have received at the other decision. When chance alone determines which unit encounters which file, the units under different configurations differ only by chance, so systematic differences in their responses reflect the files. Studies of discrimination therefore remove confounding through the assignment of configurations to decision-makers. Experiments randomize that assignment directly \citep[e.g.,][]{bertrandmullainathan2004, senwasow2016}, and observational studies are judged by how closely their assignment processes approximate such a randomization \citep{cochran1965,rubin2008}.

By two conventions of practice, the assignment also fixes the estimand. First, the association between race and the nonracial attributes it indexes is itself called confounding, and a problem so named invites a remedy in the assignment of configurations. Second, the weights that define the estimand are read off the assignment rather than chosen in their own right.

The literature's guidance prescribes one form of assignment: The two racial conditions share one distribution over nonracial profiles --- so that whatever neighborhood and record a file contains, its chance of being marked Black is the same --- with identical probabilities for every unit. %
This one act of random assignment is chosen to satisfy two conditions in the name of removing confounding, and only one of the two matters for unconfounded inference. 

The two conditions do two jobs --- selecting the estimand and securing valid inference --- and because both are properties of one mathematical object, the researcher's single choice of assignment settles both at once. Write the assignment as the array of unit-level probabilities $p_i(z, \bm{v}) \coloneqq \Pr(\bm{T}_i = (z, \bm{v}))$ that unit $i$ encounters configuration $(z, \bm{v})$, and define the two conditions as two invariances of that array, in deliberately parallel form.
\begin{defin}[Unit-blind and profile-blind assignments] \label{def: blindnesses}
An assignment with unit-level probabilities $p_i(z, \bm{v})$ is \textit{unit-blind} if $p_i(z, \bm{v})$ does not vary with $i$: Every unit faces the same probabilities, so the assignment cannot track a unit's attributes or potential outcomes. The assignment is \textit{profile-blind} if $\Pr(Z = z \given \bm{V} = \bm{v})$ does not vary with $\bm{v}$: Race falls independently of the nonracial profile, so neither racial condition is tied to particular profiles.
\end{defin}
\noindent Unit-blindness is an invariance across units, holding when all rows of the array are identical, and profile-blindness is an invariance across profiles, holding when the share of probability on each racial condition is the same within every nonracial profile.

A stripped-down charging example makes the two conditions --- unit-blindness and profile-blindness --- concrete (\Cref{fig: assignment arrays}): Three prosecutors each read their own case file, and each file pairs the arrestee's race with a nonracial profile of neighborhood disadvantage (lo/hi) and prior record (no/ex), and the assignment's probabilities form a $3 \times 8$ array, one row per prosecutor and one column per configuration. The uniform default of \citet{hainmuelleretal2014} (panel \textit{a}) and a population-calibrated design of \citet{delacuestaetal2022} (panel \textit{b}) satisfy both invariances: The rows are identical, and the chance that a file is marked Black, $\Pr(Z = 1 \given \bm{V} = \bm{v})$, is constant across nonracial profiles; the only difference between the two panels is the weighting $\rho$ common to the two racial conditions.

The two conditions can fail separately. Panel \textit{(c)} keeps the rows identical but lets the chance that a file is marked Black vary across nonracial profiles, from $1/4$ at (lo, no) to $3/4$ at (hi, ex), so unit-blindness holds while profile-blindness fails --- and because every configuration retains positive probability, panel \textit{(c)} supports the same inferences as panels \textit{(a)} and \textit{(b)}. Unit-blindness would fail instead if some prosecutor's row differed from the rest. If, say, prosecutors in heavily policed boroughs were both more likely to receive Black-marked files and more likely to charge, the Black-file cell means would draw disproportionately on those prosecutors, mixing the effect of the file with differences among the prosecutors who read it.

\begin{figure}[H]
\centering
\includegraphics[width=\linewidth]{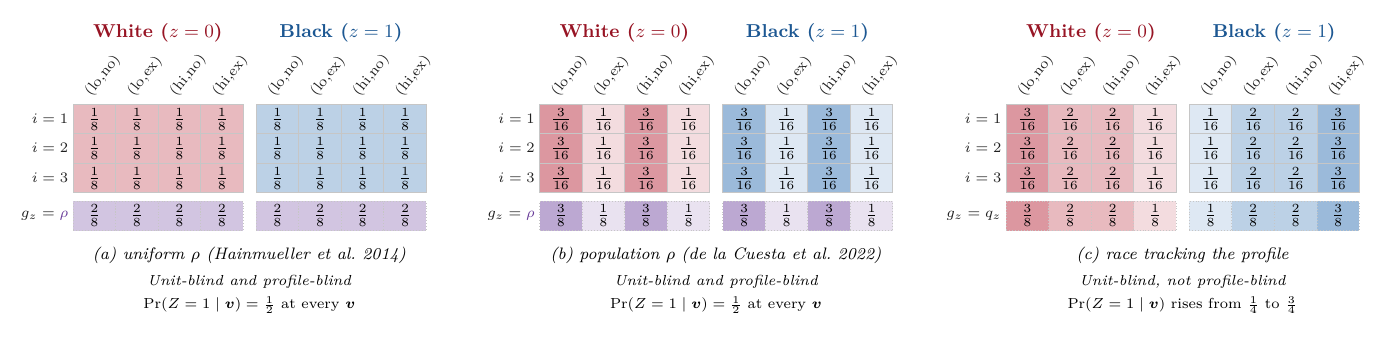}
\caption{Three configuration assignments for the charging example: one row per prosecutor $i$, one column per configuration $(z, \bm{v})$, cell shading proportional to the cell's probability. The dotted bottom margin reports the weights read off each assignment, $\Pr(\bm{V} = \bm{v} \given Z = z)$: one common weighting $\rho$ in panels \textit{(a)} and \textit{(b)}, the race-specific $q_z$ in panel \textit{(c)}.}
\label{fig: assignment arrays}
\end{figure}

Unit-blindness settles the inferential problem by itself. Every member of the estimand family is a difference between two race-specific weighted averages of cell means (\Cref{sec: anatomy estimand}). Part \textit{(i)} of \Cref{prop: two blindnesses} states that a unit-blind assignment recovers every member of the family without bias, so long as the assignment places positive probability on every profile with positive weight under the target (\Cref{sec: estimation}); the proof never invokes profile-blindness. Profile-blindness is neither necessary nor sufficient for recovering any member of the estimand family; profile-blindness contributes nothing to whether inference is confounded. If profile-blindness buys no inference, what is profile-blindness buying?

The answer is the second of the two conventions. The field reads the estimand's weights off the assignment, so the distribution the design places on nonracial profiles within each racial condition becomes the distribution over which the estimand averages, $\Pr(\bm{V} = \bm{v} \given Z = z)$. The convention is nearly universal: As \citet[][pp.~20--21]{delacuestaetal2022} report, 88\% of reviewed conjoint articles randomize every factor with the uniform distribution, and fewer than 4\% theoretically motivate the randomization distribution in the main text.
\begin{prop} \label{prop: two blindnesses}
Suppose the configuration-level SUTVA (\Cref{assm: SUTVA}), so that the estimand family of \Cref{sec: anatomy estimand} is well defined, and let the assignment be unit-blind, with fixed cell counts and $p(z, \bm{v}) > 0$ on every profile with positive weight under the target. \textit{(i)} For every member of the estimand family, the difference in cell means reweighted to that member's weighting is unbiased for that member; profile-blindness plays no role. \textit{(ii)} If instead the researcher reads the weighting off the assignment, so that nonracial profiles within each racial condition are weighted by $\Pr(\bm{V} = \bm{v} \given Z = z)$, then the weighting is common across the two racial conditions if and only if the assignment is profile-blind.
\end{prop}
\noindent \Cref{app: double duty proof} gives the proof of part \textit{(ii)}; part \textit{(i)} is the unbiasedness of the plug-in estimator of the estimand family, developed in \Cref{sec: estimation} and proved in \Cref{app: randomization properties}. Part \textit{(i)} is what randomization buys; part \textit{(ii)} is what profile-blindness buys. A weighting common to the two racial conditions is the all-else-equal structure of \eqref{eq: AEE}, so profile-blindness, read for its weighting, is the selection of the all-else-equal member. Profile-blindness was never doing inference; profile-blindness was choosing the estimand.

Both concerns travel under one name, \textit{confounding} --- applied by the statistical literature on causal inference to dependence between the assignment and the potential outcomes, and by critiques of discrimination experiments \citep{heckman1998} to the association between race and the nonracial attributes it indexes. One name applied to two conditions presents their joint removal as a single requirement, which is how the choice of estimand came to arrive inside a credibility condition: To remove confounding, in the undifferentiated sense, was to run a clean experiment and, in the same act, to choose the all-else-equal effect.

Reserving \textit{confounding} for failures of unit-blindness breaks the fusion and matches the word's formal definition: An assignment mechanism is unconfounded when its probabilities do not depend on units' potential outcomes \citep{rosenbaumrubin1983, imbensrubin2015}. On this usage, the association between race and the nonracial attributes it indexes is never a confounder. Sorting those attributes into the held-fixed and the free is the choice of contrast, a choice made on substantive and often normative grounds (\Cref{sec: substantive motivations}). Confounding remains a property of the assignment, so normative considerations enter the choice of estimand, not the inference that recovers the chosen estimand.

The distinction between choosing and recovering the estimand also explains the conclusion, noted in \Cref{sec: debate}, that positive inquiry into the effects of race cannot be separated from normative inquiry \citep{hukohler-hausmann2025,hu2026}. When one act of assignment must both select the estimand and secure inference, the normative grounds of the estimand present themselves as normative grounds of inference. Separate the two jobs, and the appearance dissolves: The estimand can be, and often is, normative; inference, once the estimand is fixed, is not.

Two connections between assignment and estimand remain, and both run from the estimand to the design: support --- which profiles the assignment can produce at all --- since the assignment must place positive probability on every profile with positive weight under the chosen member of the estimand family, and precision, taken up in \Cref{sec: estimation}. The design secures the estimand's recovery; the design never selects the estimand. \citet{lundbergetal2021} recommend the same order, estimand before design, for empirical research generally, and \citet{blairetal2019} formalize the order for political science.

\subsection{Excludability} \label{sec: excludability}

In the perception-based regime, the researcher assigns the cue $w_i \in \{0, 1\}$; the perceptions the cue would induce, $z_i(w)$ and $\bm{v}_i(w)$, are potential outcomes of the cue (\Cref{sec: perception-based}). The target is the NREC of \eqref{eq: complier effect}, the average over the compliers $\mathcal{C}$ of the outcome contrasts the cue would induce. Complier status is counterfactual --- defined by both potential racial perceptions while any realized cue reveals at most one --- so the NREC cannot be formed by subsetting. Recovery proceeds instead through two averages over all $N$ units, the cue's average effect on the outcome and the cue's average effect on perceived race; under assumptions given below, the ratio of the first to the second equals the NREC:
\begin{align} \label{eq: NREC iden ratio}
\dfrac{\frac{1}{N} \sum_{i = 1}^N \left[y_i\left(z_i(1), \bm{v}_i(1)\right) - y_i\left(z_i(0), \bm{v}_i(0)\right)\right]}{\frac{1}{N} \sum_{i = 1}^N \left[z_i(1) - z_i(0)\right]},
\end{align}
the cue's average effect on the outcome divided by the cue's average effect on perceived race.

Suppose that no unit would invert the cue, perceiving the Black cue as White and the White cue as Black, so that every unit is a complier, an always-taker, or a never-taker (\Cref{sec: perception-based}). The two averages in the ratio then run over the compliers, whose contrasts the NREC averages, and the non-compliers, whose contrasts the NREC excludes. Any difference between the ratio in \eqref{eq: NREC iden ratio} and the NREC in \eqref{eq: complier effect} therefore comes from the non-compliers alone. For the ratio to equal the NREC, it thus suffices that the non-compliers contribute nothing, and the following assumption states that condition.
\begin{assm}[No isolated nonracial shift] \label{assm: no isolated nonracial shift}
For all units $i \in \{1, \ldots, N\}$ with $z_i(0) = z_i(1)$, $\bm{v}_i(0) = \bm{v}_i(1)$.
\end{assm}
\noindent \Cref{assm: no isolated nonracial shift} asserts that the cue would not move nonracial perceptions on its own: A decision-maker who would read ``Lakisha Washington'' and ``Emily Walsh'' \citep{bertrandmullainathan2004} as the same race would, under the assumption, pick up none of the names' nonracial content.

The assumption restricts non-compliers alone. A complier's perceptions remain entirely free: Racial perception may shift with nonracial perceptions in tow, or shift by itself. The NREC averages whatever joint movement the cue would induce among compliers.

\begin{prop} \label{prop: NREC identification}
Suppose the cue-level SUTVA (\Cref{assm: SUTVA cue}), no defiers, the existence of at least one complier --- conditions stated formally in \Cref{app: NREC identification} --- and perception sufficiency (\Cref{assm: PS}). Under \Cref{assm: no isolated nonracial shift}, the ratio in \eqref{eq: NREC iden ratio} equals the NREC in \eqref{eq: complier effect}.
\end{prop}

The literature secures the equality of the ratio and the NREC through conditions stronger than \Cref{assm: no isolated nonracial shift} and frames the stronger conditions in terms of excludability. \citet[][pp.~122--123]{butlerhomola2017} assume the outcome does not depend on perceived nonracial features, $y_i(z, \bm{v}) = y_i(z, \bm{v}^{\prime})$. \citet[][eq.~6, p.~404]{dafoeetal2018} require that the cue move no unit's nonracial perceptions, $\bm{v}_i(0) = \bm{v}_i(1)$ for all $i$ --- \textit{information equivalence} in their vocabulary, \textit{cue stability} hereafter --- and relate the requirement to the exclusion restriction in instrumental-variable analysis \citep{landgraveweller2022,elderhayes2023}. %
With defiers ruled out, every unit is a complier or a non-complier, so a condition imposed on every unit says two separate things: what the condition requires of non-compliers and what the condition requires of compliers.
\begin{prop}[Decomposition of cue stability] \label{prop: cue stability decomposition}
Suppose the cue-level SUTVA (\Cref{assm: SUTVA cue}), no defiers, the existence of at least one complier, and perception sufficiency (\Cref{assm: PS}). Cue stability is the conjunction of \textit{(i)} \Cref{assm: no isolated nonracial shift} and \textit{(ii)} complier cue stability: $\bm{v}_i(0) = \bm{v}_i(1)$ for all $i \in \mathcal{C}$. Conjunct \textit{(i)} is necessary and sufficient for the ratio in \eqref{eq: NREC iden ratio} to equal the NREC, and, given conjunct \textit{(i)}, conjunct \textit{(ii)} is necessary and sufficient for the NREC to equal the AEE among compliers.
\end{prop}

The two conjuncts restrict disjoint sets of units and do different work (\Cref{fig: assumption coverage}). Given conjunct \textit{(i)}, the ratio equals the NREC, so no recovery remains for conjunct \textit{(ii)} to help; what conjunct \textit{(ii)} determines is the estimand. Holding compliers' nonracial perceptions fixed across the cue values leaves each complier's induced contrast differing in perceived race alone, and the NREC collapses into the AEE among compliers (\Cref{sec: anatomy estimand}). I draw from the decomposition a distinction the literature leaves implicit: Conditions on the units outside the target subgroup bear on whether the ratio equals the NREC; conditions on the units inside the target subgroup bear on which member of the estimand family the ratio equals. The distinction is the perception-based counterpart of unit-blindness and profile-blindness in \Cref{sec: differential exposure}: Conjunct \textit{(i)} does the recovery work of unit-blindness, and conjunct \textit{(ii)} the estimand work of profile-blindness.

\begin{figure}[H]
\centering
\includegraphics[width=\linewidth]{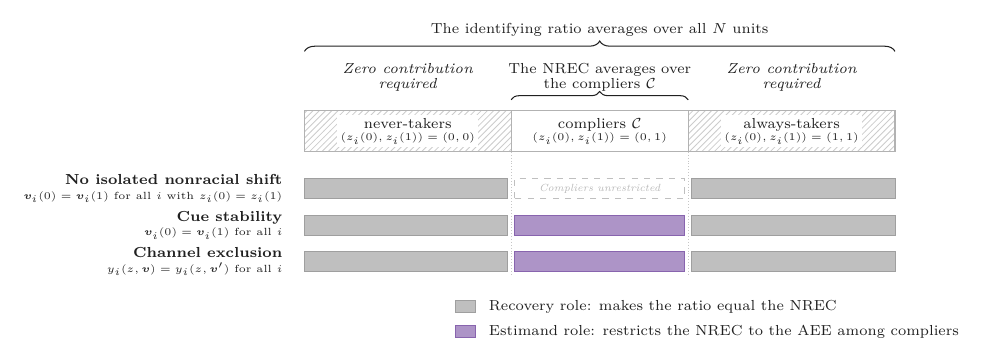}
\caption{Excludability conditions as coverage over the compliance types (defiers ruled out; segment widths illustrative). Brackets: The ratio in \eqref{eq: NREC iden ratio} averages over all $N$ units, while the NREC averages over the compliers only, so the hatched flanks must contribute zero. Rows: Each assumption covers the units it restricts; color records each block's role (\Cref{prop: cue stability decomposition}).}
\label{fig: assumption coverage}
\end{figure}

Reserving \textit{excludability} for conditions on the units the target excludes gives the word one target-invariant meaning: Conjunct \textit{(i)} exhausts excludability's content whichever member a study seeks, while conjunct \textit{(ii)} is a condition of a different kind --- not recovery but the selection of the all-else-equal member. A researcher who wants the AEE among compliers may therefore design cues with conjunct \textit{(ii)} as the aim; a researcher who wants the NREC needs conjunct \textit{(i)} alone, and on either target the movement of compliers' nonracial perceptions does not violate excludability.

The ratio in \eqref{eq: NREC iden ratio} is the Wald ratio \parencites[pp.~127--128]{angristpischke2008}[p.~180]{gerbergreen2012}, and the recovery of the NREC in \Cref{prop: NREC identification} has the same structure as the identification of local average treatment effects \citep{imbensangrist1994, angristetal1996}, with the cue as the instrument and the perceived configuration --- racial and nonracial components together --- as the treatment. \Cref{app: NREC identification} of the Supplementary Material develops the parallel, including why conjunct \textit{(ii)} has no counterpart in the classical setting of instrumental variables.

In practice, many cue-based studies measure perceived race, if at all, in a separate sample rather than on the decision-makers who supply the outcome, so such studies can estimate only the numerator of \eqref{eq: NREC iden ratio}; the numerator shares the NREC's sign and understates the NREC's magnitude (\Cref{cor: NREC lower bound} of the Supplementary Material), and a researcher can treat the unknown complier share as a sensitivity parameter, analogous to the sensitivity analyses in \citet{leavittrivera-burgos2024,leavittrivera-burgos2026}.

\section{Estimation and Inference across the Estimand Family} \label{sec: estimation}

\subsection{Exposure-Based Regime}

Under the exposure-based regime, the researcher's design choice is the set of cell counts --- how many units encounter each profile --- and random assignment allocates units to profiles given those counts.
\begin{assm}[Complete random assignment of configurations] \label{assm: CRA}
Given cell counts $n_{z, \bm{v}}$, one positive integer per profile $(z, \bm{v}) \in \mathcal{T}$, the configuration assignment $\bm{T} = (\bm{T}_1, \ldots, \bm{T}_N)$ takes each value in the resulting set $\Omega$ with probability $1/\abs{\Omega}$, so $\Pr(\bm{T}_i = (z, \bm{v})) = n_{z, \bm{v}}/N$.
\end{assm}
\noindent Because the probabilities $n_{z, \bm{v}}/N$ do not vary with the unit, the assignment is unit-blind (\Cref{def: blindnesses}), unable to track a unit's attributes or potential outcomes; the assumption leaves profile-blindness unconstrained. Unit-blindness suffices to recover whichever member of the estimand family the researcher targets (\Cref{sec: differential exposure}).

For each profile $(z, \bm{v}) \in \mathcal{T}$, the one quantity estimated from data is the cell mean
\begin{align} \label{eq: cell mean estimator}
\hat{\mu}(z, \bm{v}) \coloneqq \dfrac{1}{n_{z, \bm{v}}} \sum_{i = 1}^N Y_i \, \mathbbm{1}\{\bm{T}_i = (z, \bm{v})\},
\end{align}
the average outcome among units assigned to profile $(z, \bm{v})$, where $Y_i$ is unit $i$'s observed outcome and $\mathbbm{1}\{\cdot\}$ is the indicator function, equal to $1$ when the condition in braces holds and $0$ otherwise. Every estimand in \Cref{sec: anatomy estimand} is a difference between two race-specific weighted averages of these cell means, so one plug-in estimator covers the entire estimand family. Writing $g_0$ and $g_1$ for the weighting distributions the target estimand places on the profiles with $z = 1$ and $z = 0$, the plug-in estimator is
\begin{align} \label{eq: family estimator main}
\widehat{\theta} \coloneqq \sum_{\bm{v} \in \mathcal{V}} g_1(\bm{v})\, \hat{\mu}(1, \bm{v}) - \sum_{\bm{v} \in \mathcal{V}} g_0(\bm{v})\, \hat{\mu}(0, \bm{v}).
\end{align}
The AEE takes the common marginal, $g_0 = g_1 = \rho$; the AEWR takes the race-specific conditionals, $g_z = q_z$; and a mixed contrast takes a weight that is common on the held-fixed features and race-specific on the free ones (\Cref{sec: mixed contrast formalism}). These weights carry no hat: The researcher fixes them in advance, so \eqref{eq: family estimator main} targets the chosen member of the estimand family, and changing the weights changes the estimand, not the estimator's ability to recover the estimand.

Under \Cref{assm: CRA}, $\widehat{\theta}$ is unbiased and consistent for its estimand, whatever the weights $g_0$ and $g_1$. A single construction yields valid confidence intervals for every member of the estimand family (\Cref{app: randomization properties}). Unbiasedness holds under any design that assigns at least one unit to every profile cell: The estimator reweights the cell means to the target distribution whatever the cell counts $n_{z, \bm{v}}$, so the weights define the estimand and the counts do not.

The variance of $\widehat{\theta}$ grows as the weights $g_z(\bm{v})$ depart from the assigned shares $n_{z, \bm{v}} / N$, so a researcher who controls the design can increase precision by assigning profiles in proportion to the target weights --- the strategy \citet{delacuestaetal2022} develop for the AMCE, which in the estimand family of \Cref{sec: anatomy estimand} means assigning by the common $\rho$ for an all-else-equal target and by the race-specific $q_z$ for a within-race target. Neither weighting is inherently more precise. As in \Cref{sec: differential exposure}, the order runs from estimand to design: A researcher who specifies the member of the estimand family to target can then choose the assignment that recovers the member precisely.

\subsection{Perception-Based Regime}

In the perception-based regime the target is the NREC, and the researcher assigns only the cue.
\begin{assm}[Complete random assignment of cues] \label{assm: CRA cue}
Given cue counts $n_w$, one positive integer per value $w \in \{0, 1\}$, the cue assignment $\bm{W} = (W_1, \ldots, W_N)$ takes each value in the resulting set $\Omega_W$ with probability $1/\abs{\Omega_W}$, so $\Pr(W_i = w) = n_w/N$.
\end{assm}
\noindent Randomizing the cue does not randomize the racial contrast, which is a potential outcome of the cue. Recovery therefore needs, beyond randomization, the conditions of \Cref{sec: excludability}: no defiers and at least one complier (\Cref{app: NREC identification}), and \Cref{assm: no isolated nonracial shift}; nothing stronger is required.

Under \Cref{assm: CRA cue}, unit $i$'s cue is now the random variable $W_i$, and $Z_i \coloneqq z_i(W_i)$ is the perceived race that $W_i$ induces (\Cref{sec: perception-based}). Two difference-in-means estimators compare the outcome and perceived race across the two cue values --- the reduced-form estimator $\widehat{\Delta}_Y$ and the first-stage estimator $\widehat{\Delta}_Z$:
\begin{align}
\widehat{\Delta}_Y &\coloneqq \dfrac{1}{n_1} \sum_{i = 1}^N Y_i \, \mathbbm{1}\{W_i = 1\} - \dfrac{1}{n_0} \sum_{i = 1}^N Y_i \, \mathbbm{1}\{W_i = 0\}, \label{eq: rf estimator} \\
\widehat{\Delta}_Z &\coloneqq \dfrac{1}{n_1} \sum_{i = 1}^N Z_i \, \mathbbm{1}\{W_i = 1\} - \dfrac{1}{n_0} \sum_{i = 1}^N Z_i \, \mathbbm{1}\{W_i = 0\}. \label{eq: fs estimator}
\end{align}

When perceived race is measured on the same respondents who supply the outcome, $\widehat{\mathrm{NREC}} \coloneqq \widehat{\Delta}_Y / \widehat{\Delta}_Z$ estimates the NREC: By \Cref{prop: NREC identification}, the ratio is the instrumental-variables estimator (\Cref{sec: excludability}), with the cue as the instrument and perceived race as the treatment taken up. Absent a first stage, as in most cue-based audits, only $\widehat{\Delta}_Y$ is available, and $\widehat{\Delta}_Y$ unbiasedly estimates the lower bound in magnitude of the NREC (\Cref{cor: NREC lower bound} of the Supplementary Material). A confidence interval for the ratio can be derived from a delta-method approximation that, like the usual instrumental-variables variance, degrades when the cue moves perceived race only weakly; \Cref{app: NREC inference} gives the justification and alternatives for weak instruments.

\section{Application: Race and Language in Candidate Evaluation} \label{sec: app_1}

Would Hispanic voters evaluate a Hispanic candidate more favorably than an Anglo candidate? Coethnic preference is the expectation of the literature on linked fate --- the belief that one's own life chances are tied to those of the group --- and group consciousness \citep{dawson1994,sanchez2008,sanchezmasuoka2010}. I take the question to the candidate-evaluation experiment of \citet{zarateetal2024}, read under both regimes of \Cref{sec: formal setup}. Across both regimes, which conclusion is right comes down to which effect of race the study targets, not to which estimand the data better recover.

In a $2 \times 3$ between-subjects experiment, \citet{zarateetal2024} vary a candidate's race, Anglo ($z = 0$) or Hispanic ($z = 1$), and the language of a recorded campaign speech. The candidate, a hypothetical state representative seeking reelection, is named either Josh Martin or Josu\'{e} Mart\'{i}nez, and the written vignette that introduces the recording gives the full name and states the candidate's race outright: ``Representative [Martin/Mart\'{i}nez], who is [White/Latino]'' \citep[][p.~368]{zarateetal2024}. Language is the sole nonracial feature, so $\mathcal{V} = \{\textrm{English}, \textrm{Non-Native Spanish}, \textrm{Native Spanish}\}$: The speech is in English throughout or includes two sentences in Spanish, spoken with an American-English or a native accent, and both racial conditions use the same recordings, made by a single voice actor. Each respondent in a national sample of Hispanic adults listens to one randomly assigned recording and rates the candidate on an index combining vote intention, trust, liking, and feeling represented.

The design can be read under both regimes. The exposure reading takes the stated race as assigned, and the perception reading --- since a respondent may perceive a race other than the one stated --- takes the name and label together as a single cue. \citet{zarateetal2024} also measure the race each respondent perceives --- a measurement many cue-based studies lack --- so the perception-based regime's target is directly estimable (\Cref{sec: app perception}).

\citet{zarateetal2024} ask whether Spanish-language appeals raise Hispanic respondents' evaluations and whether respondents evaluate the Hispanic candidate more favorably than the Anglo candidate. On the second question, the authors compare the two candidates at one language at a time and find coethnic preference under English but no detectable preference under either Spanish condition. Those comparisons hold language fixed; the reanalysis asks the same question without deciding in advance whether language is held fixed across the racial conditions (the AEE) or varies with race (the AEWR).

Whether the AEE or the AEWR answers the question of coethnic preference depends on whether one treats language as separable from the category Hispanic --- the choice of \Cref{sec: debate} between the classical conception, on which a racial category consists in nominal membership alone, and the thick constructivist conception, on which the category also probabilistically indexes nonracial attributes. Commenting on the seminal Urban Institute audit of Hispanic and Anglo job applicants \citep{crossetal1990}, \citet[][pp.~217--218]{heckmansiegelman1993} exemplify the classical conception: They read the Hispanic accent as separable from the category, to be held fixed to isolate discrimination against ``Hispanics per se.''

For the most part, the literature on Latino politics instead treats language as part of what the category is, on both sides of the encounter --- the candidate's and the respondent's. On the candidate's side, Hispanic candidates court Hispanic voters \citep{barreto2007} in Spanish \citep{abrajano2010}, so the two racial conditions differ in their language distributions as part of the category's social content. On the respondent's side, race is perceived partly through speech --- ``sounding like a race'' \citep{rosa2018} --- and to be read as Hispanic is in part to be read as foreign \citep{chavez2013}, among Latino perceivers as well as White ones \citep{zoucheryan2017}.

Each side of the encounter matters under one regime. Under the exposure reading, the stated race is taken as assigned, and the open choice is the weighting over language; the weights come from the candidate's side. Weighting each racial condition by its own campaign distribution, the AEWR answers the constructivist question; weighting both by a common distribution, the AEE answers the classical one (\Cref{sec: app exposure}). Under the perception reading, the candidate's race is what the respondent perceives, and the respondent's side determines what the cue induces. Speech shapes which race the respondent perceives, foreignness may move alongside perceived race, and the target is the NREC (\Cref{sec: app perception}).

Throughout the application, the estimands average over Hispanic respondents' heterogeneity in origin, generation, and language \citep{masuokajunn2013,dowling2014,sanchezetal2012}. On the candidate's side, the campaign distributions represent how candidates speak, not how Hispanic voters do. Narrower subsets or different reference distributions would define other members of the estimand family (\Cref{sec: anatomy estimand}).

\subsection{The Exposure-Based Reading: All-Else-Equal versus Within-Race over Language} \label{sec: app exposure}

The exposure reading turns on the contrast over language, held fixed (the AEE) or varying with race (the AEWR). The AEE weights language by a common distribution $\rho$. The AEWR instead weights language by each racial condition's own $q_z$.

I ground both distributions in how candidates campaign. \citet{zarateetal2024} analyze every Spanish-language U.S.\ House advertisement aired between 2010 and 2018. Among the candidates who campaigned in Spanish --- \nSpanAnglo{} Anglo and \nSpanHisp{} Hispanic --- \pctAngloNonNative{} percent of Anglos used an American accent (Non-Native Spanish) and \pctHispNative{} percent of Hispanics a native one (\Cref{tab: lang dist}). Being Hispanic rather than Anglo thus flips which accent is typical. The campaign frequencies are the race-specific distributions $q_z$; weighting them by each racial condition's share of the Spanish-language campaigners gives the common distribution $\rho$.

These weights enact both grounds of \Cref{sec: weighting substantive}: The weights calibrate the comparison to the population of real candidacies (external validity) and evaluate each racial condition at the accent typical of its own campaigners (typicality). The two grounds coincide because typicality is here defined by the target population itself, though a researcher could define typicality against a different reference distribution. Because the same two race-conditional distributions barely overlap, the common weighting faces the structural tradeoff of \Cref{sec: substantive motivations}: Any one distribution either confines the comparison to the sliver of overlap or places weight on an accent rare under one of the two distributions (\Cref{fig: structural tradeoff}).

\begin{table}[H]
\vspace{10pt}
\input{tab_language_dist}
\end{table}

The plug-in estimator of \Cref{sec: estimation} gives an AEWR of $\AEWRest$ (95\% CI $[\AEWRlo, \AEWRhi]$), distinguishable from zero, and an AEE of $\AEEest$ (95\% CI $[\AEElo, \AEEhi]$), not (\Cref{fig: aewr aee res}). On the same experiment, then, the choice of estimand is the difference between finding coethnic preference and finding none. %

\begin{figure}[H]
\centering
\includegraphics[width=\linewidth]{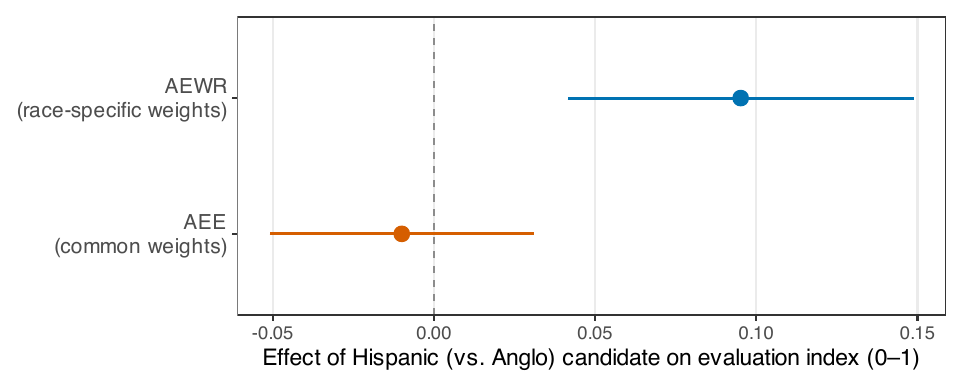}%
\caption{AEWR and AEE estimates of the effect of a Hispanic (versus Anglo) candidate on Hispanic respondents' evaluations, restricted to the two Spanish-language conditions. Bars are 95\% confidence intervals based on consistent estimators of the variance upper bound and a Normal approximation.}
\label{fig: aewr aee res}
\end{figure}

The source of the divergence between the AEWR and AEE estimates is visible in the race-by-language cell means (\Cref{tab: cell means}). Within each language, the Anglo and Hispanic mean evaluations are nearly identical. However, across languages, respondents rate Native Spanish more positively than Non-Native Spanish, for candidates of both races.

\begin{table}[H]
\vspace{10pt}
\input{tab_cell_means}
\end{table}

The AEE averages the two near-zero within-language race contrasts and is essentially zero: The common weight $\rho$ puts $\rhoNat$ on native Spanish for Anglo candidates, far above the \pctAngloNative{} percent who campaign that way, crediting them with evaluations they rarely earn. Because both within-language race contrasts are near zero, any common weighting would leave the AEE near zero --- whether the uniform default \citep[][p.~12]{hainmuelleretal2014} or the target population's distribution for external validity \citep{delacuestaetal2022}, based here on the campaign data. The AEWR instead compares the typical Hispanic candidate (Native Spanish, $\muHnat$) with the typical Anglo candidate (Non-Native, $\muAnn$), and is positive.

If the category Anglo indexes language, so that being Anglo consists partly in rarely speaking native Spanish, then the AEE weights an atypical profile as though it were common, while the AEWR does not (\Cref{sec: substantive motivations}). Both estimands are recovered from the same assignment under the same assumptions: The experiment randomizes race and language jointly, so every race-by-language cell is filled by random assignment, and each estimand is a different weighting of the same cell means (\Cref{sec: differential exposure}).

\subsection{The Perception-Based Reading: The Natural Race Effect among Compliers} \label{sec: app perception}

The name and label form the cue $w$, and the perceived-race item measures the cue's potential outcome, perceived race $z_i(w)$, at the assigned cue value (\Cref{sec: perception-based}). The item directly follows the recording and precedes the evaluation items. Answering the item could raise the salience of race for the evaluations that follow, but only salience that the item raises more under one cue value than the other would enter the contrast. The Hispanic name and label might raise the salience of race more than the White ones do, a possibility the design cannot rule out.

I estimate the NREC within each language condition. At a fixed language, respondents hear the same recording and differ only in the name and label, so every difference in their perceptions, racial or nonracial, is induced by the cue. Pooled across the language conditions, perceived race would move with the assigned language as well as with the cue, and the compliers would no longer be defined at a single language.

The instrumental-variables ratio of the evaluation on perceived race, with the cue as the instrument, estimates the NREC \eqref{eq: complier effect} (\Cref{sec: estimation}). \Cref{fig: nrec by lang} reports, by language condition, estimates of the first stage --- the cue's effect on the share perceiving the candidate as Hispanic --- and of the NREC.

\begin{figure}[H]
\centering
\includegraphics[width=\linewidth]{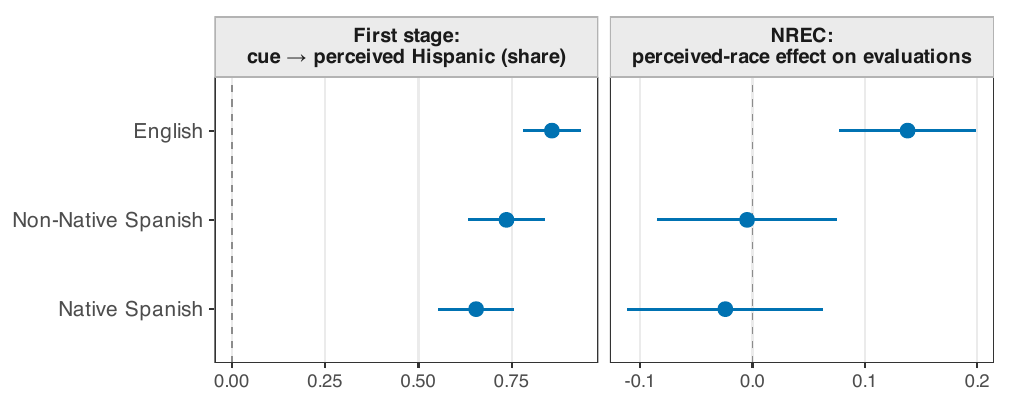}
\caption{Perception-based estimates by language condition (Study 1, Hispanic respondents): the first stage (\textit{left}) and the NREC (\textit{right}). Bars are 95\% confidence intervals; the NREC interval uses a delta-method standard error \citep{pashley2022} with a Normal approximation.}
\label{fig: nrec by lang}
\end{figure}

The NREC answers the question of coethnic preference in the perception reading's terms, and the answer depends on the language the candidate speaks. Under English, the cue moves perceived race sharply (estimated first stage $\fsEng$) and the estimated NREC is positive ($\nrecEng$, 95\% CI $[\nrecEngLo, \nrecEngHi]$): Respondents evaluate a candidate they perceive as Hispanic more favorably than one they perceive as White.

In both Spanish conditions --- non-native and native --- the estimated NREC is indistinguishable from zero, and two patterns may together account for the null results. First, within a Spanish condition the race contrast in evaluations all but vanishes, as though speaking Spanish mutes the role of perceived race in respondents' judgments, consistent with accented language operating as an intergroup cue in its own right \citep{hopkins2015}. Second, the accent appears to reshape perceived race itself, weakening the cue's first stage: Under native Spanish, \pctAngloHispNat{} percent of respondents perceive the Anglo candidate as Hispanic, against \pctAngloHispEng{} percent under English, while under non-native Spanish, \pctHispWhiteNN{} percent perceive the Hispanic candidate as White, against \pctHispWhiteEng{} percent. The two patterns together suggest that respondents perceive race partly through the speech that conveys it --- ``sounding like a race'' \citep{rosa2018} registering in the first stage itself --- and leave a clear perceived-race contrast only under English.

Whether that perceived-race effect is also the all-else-equal effect of race turns on the cue moving compliers' perceived race and nothing else, which a name bearing a race label may fail to satisfy. \Cref{fig: name bundle} bears on whether the name signals more than race, with the first-name trait ratings of \citet{elderhayes2023}: ``Josu\'e,'' the Hispanic name \citet{zarateetal2024} use, tracks Hispanic names as a whole, rated markedly more foreign, modestly more working-class, and less Republican than White-coded names. %

\begin{figure}[H]
\centering
\includegraphics[width=\linewidth]{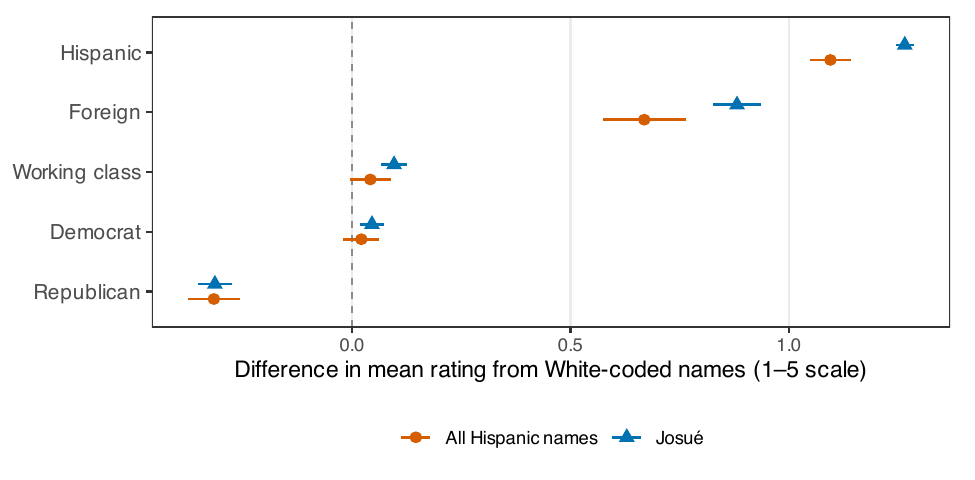}
\caption{Mean first-name trait ratings from \citet{elderhayes2023} (1--5 scale), as contrasts against White-coded names: all Hispanic-coded names (orange) and ``Josu\'e,'' the name \citet{zarateetal2024} use (blue). The ratings pool a separate respondent sample, include ``Josu\'e'' but not the Anglo name, and compare across names, so the figure shows that the name signals these traits, not how a respondent combines them. Intervals reflect variation across names --- for ``Josu\'e,'' the spread of the White-coded comparison names --- rather than rater-level sampling.}
\label{fig: name bundle}
\end{figure}

That the name conveys these traits suggests the cue moves the perceptions rather than leaving them fixed, so the NREC would bundle the movement and depart from the all-else-equal effect. The departure presumes that the name marks foreignness directly rather than that respondents infer foreignness from perceived race (\Cref{sec: perception-based}), and the two orderings cannot be told apart, not by ratings from a separate sample and not in principle, since the sequencing occurs inside the respondent's head \citep{robinsrichardson2010}. Nor need the orderings be exclusive: The departure follows if any part of the foreignness moves directly with the name.

\subsection{The Two Readings Together}

The reanalysis returns an answer about coethnic preference that the all-else-equal default conceals. Hispanic respondents do reward the Hispanic candidate, but the reward attaches to the candidate who is Hispanic in the way Hispanic candidates actually are --- speaking native-accented Spanish when they campaign in the language --- and the reward disappears when the comparison strips the language away. The common weighting evaluates the Anglo condition mainly at native Spanish, a profile few Spanish-campaigning Anglos fit; the Anglo candidate who speaks native Spanish is the counterpart of the Senegalese Muslim of the introduction, the atypical member of the category Anglo on whom an all-else-equal comparison must lean.

The two regimes reach the same lesson by different routes. What the audit literature might treat as a confound to hold fixed (\Cref{sec: differential exposure}), and what the cue-stability literature might treat as a failed manipulation (\Cref{sec: excludability}), may be the content of the category itself \citep{hukohler-hausmann2025}. The reanalysis kept the category as constituted and kept causal inference (\Cref{sec: debate}): The AEWR and the NREC let language move with the category, and both are well-defined causal estimands, with estimates and confidence intervals from the same experiment.

\section{Conclusion} \label{sec: conclusion}

This paper has argued that the literature's defenses of the all-else-equal design bundle an estimand commitment with a recovery condition. Unbundled, the same randomization recovers an entire family of race estimands --- each built from counterfactual effects at the unit level, not from disparities across people --- and the choice among that family's members is a claim about what a racial category is, not a concession on rigor. The ``Muslim effect'' of the opening pages was such a choice all along: Isolating religion from origin and language decided, before any data arrived, that the comparison would sit at the categories' edges rather than their centers. The application shows the choice governing the answer: On one experiment, Hispanic voters' coethnic preference is absent under the all-else-equal weighting and clear under the within-race weighting, because the preference lives in the language the category indexes. Three questions concern what the framework asks of empirical research once the estimand is a substantive choice.

The first question is how to choose among the members of the estimand family. The paper lays out grounds for the choice (\Cref{sec: substantive motivations}) but does not adjudicate among them, and develops the normative grounds least: The choice of contrast can be grounded in what one takes discrimination to be and in which differences across groups are invidious rather than benign \citep{hu2026}, a stance explicit in health-care policy research \citep{iom2002,mcguireetal2006} but rare in political science. The open task is the pairing itself: which conception of discrimination names which member of the estimand family, so that normative theory sets the target, empirical design recovers it, and each constrains the other rather than proceeding apart. %

The second question is what to do with a choice the data cannot settle. A study might report several members of the estimand family --- as the application does, setting the AEE beside the AEWR --- so readers see how a finding of discrimination depends on the conception of the category, much as a sensitivity analysis displays dependence on an unverifiable assumption. The study might instead commit to a single estimand and defend the conception that estimand embodies. Either path sets practical tasks: recovering the within-group distributions $q_z$ that the within-race and mixed estimands require (\Cref{app: distribution formalism}), deciding which attributes a category indexes and so belong in $\bm{v}$ (\Cref{sec: formal setup}), measuring perceived categories in the perception-based regime, and carrying the estimands to observational studies, where identification assumptions replace the randomization this paper assumes.

The third question reaches beyond discrimination to political representation \citep{pitkin1967,mansbridge1999}. When a representative is said to represent a racial group, is the group fixed by a category label or by the broader profile the label indexes? The question becomes concrete in defining the group's preference, for example the median opinion against which a legislator's votes are judged \citep{millerstokes1963,laxetal2019,rivera-burgos2025}: Is the median taken over the racial group with income and the rest held fixed, or over the group as constituted? No design can dissolve the choice. Some individuals, for example, are both Black and middle-class, so a counterfactual shift in Black opinion moves middle-class opinion with it; unlike race and income on a r\'{e}sum\'{e}, which a study can pair freely, the two opinions cannot be paired at will. The estimand family gives the choice precise form, with legislators as the decision-makers, roll-call votes as the decisions, and the constituency as the configuration, where $z$ encodes a position of the racial group of interest and $\bm{v}$ the analogous statistics for other groups.

This paper does not answer these questions. Its contribution is a framework for taking them up: The framework translates contested conceptions of race and other identity categories into distinct causal estimands that empirical studies can reliably recover. Each contested conception can then be combined with rigorous causal evidence rather than debated without it.

\newpage
\begin{singlespace}
\printbibliography
\end{singlespace}
\end{refsection}

\clearpage

\setcounter{section}{0}
\setcounter{equation}{0}
\setcounter{figure}{0}
\setcounter{table}{0}
\setcounter{footnote}{0}
\setcounter{assm}{0}
\setcounter{prop}{0}
\setcounter{defin}{0}
\setcounter{lem}{0}
\setcounter{cor}{0}
\setcounter{thm}{0}
\setcounter{algo}{0}
\setcounter{example}{0}
\renewcommand{\thesection}{S\arabic{section}}
\renewcommand{\theequation}{S\arabic{equation}}
\renewcommand{\thefigure}{S\arabic{figure}}
\renewcommand{\thetable}{S\arabic{table}}
\renewcommand{\theHsection}{S\arabic{section}}
\renewcommand{\theHequation}{S\arabic{equation}}
\renewcommand{\theHfigure}{S\arabic{figure}}
\renewcommand{\theHtable}{S\arabic{table}}
\renewcommand{\theHassm}{S\arabic{assm}}
\renewcommand{\theHprop}{S\arabic{prop}}
\renewcommand{\theHdefin}{S\arabic{defin}}
\renewcommand{\theHlem}{S\arabic{lem}}
\renewcommand{\theHcor}{S\arabic{cor}}
\renewcommand{\theHthm}{S\arabic{thm}}
\renewcommand{\theHalgo}{S\arabic{algo}}
\renewcommand{\theHexample}{S\arabic{example}}
\setcounter{tocdepth}{2}
\addtocontents{toc}{\protect\setcounter{tocdepth}{2}}

\begin{refsection}
\begin{center}
{\Large Supplementary Material for ``Which Effect of Race? Causal Inference without Holding All Else Equal''}\\[1.5ex]
\ifnum\blind=1\else Thomas Leavitt\fi
\end{center}

{\singlespacing
\setcounter{tocdepth}{2}
\tableofcontents}

\makeatletter
\renewcommand{\theassm}{S\arabic{assm}}
\renewcommand{\theprop}{S\arabic{prop}}
\renewcommand{\thecor}{S\arabic{cor}}
\renewcommand{\thelem}{S\arabic{lem}}
\renewcommand{\thethm}{S\arabic{thm}}
\renewcommand{\thedefin}{S\arabic{defin}}
\renewcommand{\theexample}{S\arabic{example}}
\makeatother

\doublespacing

\section{Formal Results}

\subsection{From decision-makers and decisions to units} \label{app: units sutva}

\Cref{sec: formal setup} of the main text indexes units directly and points here for the underlying structure. There is a set of decision-makers indexed by $j \in \{1, \ldots, J\}$, and each decision-maker $j$ faces $K_j$ decisions indexed by $k \in \{1, \ldots, K_j\}$, giving $N \coloneqq \sum_j K_j$ decision-maker--decision pairs. Under the exposure-based regime, decision-maker $j$ in decision $k$ encounters a configuration $\bm{t}_{jk} \in \mathcal{T}$. Let $\bm{t} \coloneqq (\bm{t}_{11}, \ldots, \bm{t}_{JK_J}) \in \mathcal{T}^N$ stack these configurations into a \textit{configuration assignment}. In principle, the potential outcome of the pair $(j,k)$ may depend on the entire assignment.

\begin{assm}[SUTVA] \label{assm: SUTVA}
For all decision-makers $j$, decisions $k$, and configuration assignments $\bm{t}, \bm{t}^{\prime} \in \mathcal{T}^N$, $y_{jk}(\bm{t}) = y_{jk}(\bm{t}^{\prime})$ whenever $\bm{t}_{jk} = \bm{t}^{\prime}_{jk}$.
\end{assm}

\noindent \Cref{assm: SUTVA} rules out interference between units and carryover across a decision-maker's own decisions, encompassing the ``stability and no carryover effects'' assumption of \citet[][p.~8]{hainmuelleretal2014}. This assumption licenses reindexing the pairs $(j,k)$ as units $i \in \{1, \ldots, N\}$ whose potential outcomes depend only on their own configurations, defining the function $y_i: \mathcal{T} \to \R$ used throughout the main text.

Under the perception-based regime, only the cue is assigned. Let $\bm{w} \coloneqq (w_{11}, \ldots, w_{JK_J})\allowbreak \in \{0, 1\}^N$ stack the cues into a \textit{cue assignment}. The analogue of \Cref{assm: SUTVA} applies to the response, the racial perception, and the nonracial perceptions at once.

\begin{assm}[SUTVA for the cue] \label{assm: SUTVA cue}
For all decision-makers $j$, decisions $k$, and cue assignments $\bm{w}, \bm{w}^{\prime} \in \{0, 1\}^N$,
\begin{align*}
y_{jk}(\bm{w}) = y_{jk}(\bm{w}^{\prime}), \quad z_{jk}(\bm{w}) = z_{jk}(\bm{w}^{\prime}), \quad \text{and} \quad \bm{v}_{jk}(\bm{w}) = \bm{v}_{jk}(\bm{w}^{\prime})
\end{align*}
whenever $w_{jk} = w_{jk}^{\prime}$.
\end{assm}

\noindent Reindexing as before yields the unit-level potential outcomes $z_i(w)$, $\bm{v}_i(w)$, and $y_i(w)$ of \Cref{sec: perception-based}.

\subsection{The gap between two estimands} \label{sec: gap}

A single identity, stated as \eqref{eq: gap identity} in \Cref{prop: gap} below, gives the gap between any two members of the estimand family. The identity is purely algebraic: \Cref{assm: SUTVA} makes the estimands well defined, but the identity itself invokes no assumption beyond the estimands' definitions.

\begin{prop}[Gap between two members of the estimand family] \label{prop: gap}
Write
\begin{align*}
\theta(g_0, g_1) \coloneqq \frac{1}{N} \sum_{i=1}^N \sum_{\bm{v} \in \mathcal{V}} \left[y_i(1, \bm{v})\, g_1(\bm{v}) - y_i(0, \bm{v})\, g_0(\bm{v})\right]
\end{align*}
for the member of the estimand family that weights nonracial profiles by $g_0$ under racial condition $z = 0$ and by $g_1$ under racial condition $z = 1$ --- the estimand targeted by the plug-in estimator of the estimand family, \eqref{eq: family estimator main} of the main text --- and $\bar{y}(z, \bm{v}) \coloneqq (1/N) \sum_{i=1}^N y_i(z, \bm{v})$. For any two members of the estimand family,
\begin{align} \label{eq: gap identity}
\theta(g_0, g_1) - \theta(g_0^{\prime}, g_1^{\prime}) = \sum_{\bm{v} \in \mathcal{V}} \left\{\bar{y}(1, \bm{v})\left[g_1(\bm{v}) - g_1^{\prime}(\bm{v})\right] - \bar{y}(0, \bm{v})\left[g_0(\bm{v}) - g_0^{\prime}(\bm{v})\right]\right\}.
\end{align}
\end{prop}

\begin{proof}
The proof subtracts the two estimands term by term and then exchanges the order of the two finite sums, so that each difference of weights multiplies an average potential outcome.

Applying the definition of $\theta$ to each pair of weights,
\begin{align*}
\theta(g_0, g_1) - \theta(g_0^{\prime}, g_1^{\prime}) & = \dfrac{1}{N} \sum_{i=1}^N \sum_{\bm{v} \in \mathcal{V}} \left[y_i(1, \bm{v})\, g_1(\bm{v}) - y_i(0, \bm{v})\, g_0(\bm{v})\right] \\
& \qquad - \dfrac{1}{N} \sum_{i=1}^N \sum_{\bm{v} \in \mathcal{V}} \left[y_i(1, \bm{v})\, g_1^{\prime}(\bm{v}) - y_i(0, \bm{v})\, g_0^{\prime}(\bm{v})\right] \\
& = \dfrac{1}{N} \sum_{i=1}^N \sum_{\bm{v} \in \mathcal{V}} \left\{y_i(1, \bm{v})\left[g_1(\bm{v}) - g_1^{\prime}(\bm{v})\right] - y_i(0, \bm{v})\left[g_0(\bm{v}) - g_0^{\prime}(\bm{v})\right]\right\}
\end{align*}
where the second equality combines the two sums term by term and, within each $(i, \bm{v})$ term, factors $y_i(1, \bm{v})$ and $y_i(0, \bm{v})$ out of their respective differences of weights.
Both index sets are finite, so the sums over $i$ and $\bm{v}$ may be exchanged, and the average $(1/N) \sum_{i}$ distributes over the potential outcomes because the differences of weights do not depend on $i$:
\begin{align*}
\theta(g_0, g_1) - \theta(g_0^{\prime}, g_1^{\prime}) & = \sum_{\bm{v} \in \mathcal{V}} \Bigg\{\left[\dfrac{1}{N} \sum_{i=1}^N y_i(1, \bm{v})\right]\left[g_1(\bm{v}) - g_1^{\prime}(\bm{v})\right] \\
& \qquad \qquad - \left[\dfrac{1}{N} \sum_{i=1}^N y_i(0, \bm{v})\right]\left[g_0(\bm{v}) - g_0^{\prime}(\bm{v})\right]\Bigg\}.
\end{align*}
Substituting $\bar{y}(z, \bm{v}) = (1/N)\sum_{i} y_i(z, \bm{v})$ then yields the identity \eqref{eq: gap identity}.
\end{proof}

The difference between the AEWR and the AEE is the instance of the identity \eqref{eq: gap identity} with $(g_0, g_1) = (q_0, q_1)$ and $(g_0^{\prime}, g_1^{\prime}) = (\rho, \rho)$. A mixed contrast enters the identity \eqref{eq: gap identity} through the mixed weights $g_z(\bm{v}) = \rho(\bm{v}_{\mathrm{fix}})\, q_z(\bm{v}_{\mathrm{free}} \given \bm{v}_{\mathrm{fix}})$, stated as \eqref{eq: mixed weights} in \Cref{sec: mixed contrast formalism}. Substituting the mixed weights on one side of \eqref{eq: gap identity} gives the gap between a mixed effect and the AEE or the AEWR, and substituting the mixed weights of two different partitions on the two sides gives the gap between two mixed effects.

Two members of the estimand family coincide when the right-hand side of the identity \eqref{eq: gap identity} vanishes. The right-hand side vanishes when the weights agree, $g_z = g_z^{\prime}$ for $z \in \{0, 1\}$, and also when $\bar{y}(z, \bm{v})$ does not depend on $\bm{v}$. Writing $\bar{y}(z)$ for the common value and factoring the common value out of the sum,
\begin{align*}
\sum_{\bm{v} \in \mathcal{V}} \bar{y}(z, \bm{v})\left[g_z(\bm{v}) - g_z^{\prime}(\bm{v})\right] = \bar{y}(z) \left[\sum_{\bm{v} \in \mathcal{V}} g_z(\bm{v}) - \sum_{\bm{v} \in \mathcal{V}} g_z^{\prime}(\bm{v})\right] = \bar{y}(z)\,(1 - 1) = 0,
\end{align*}
because $g_z$ and $g_z^{\prime}$ are probability distributions on $\mathcal{V}$ and therefore each normalizes to one: $\sum_{\bm{v} \in \mathcal{V}} g_z(\bm{v}) = \sum_{\bm{v} \in \mathcal{V}} g_z^{\prime}(\bm{v}) = 1$.

The AEWR-versus-AEE pairing gives the cleanest reading of the gap as a failure of two operations to commute, differencing across racial conditions and averaging over nonracial profiles: The AEE differences at each shared profile and then averages over $\rho$, whereas the AEWR averages each racial condition over its own conditional $q_z$ and then differences the two race-specific averages. The two estimands agree exactly when the two operations commute.

\subsection{The double duty of a profile-blind assignment} \label{app: double duty proof}

\Cref{sec: differential exposure} of the main text states \Cref{prop: two blindnesses} and points here for the proof of part \textit{(ii)}; part \textit{(i)} is the unbiasedness of the plug-in estimator for the targeted member of the estimand family, proved in \Cref{app: randomization properties}.

\begin{proof}[Proof of part \textit{(ii)} of \Cref{prop: two blindnesses}]
The weighting read off the assignment is $g_z(\bm{v}) \coloneqq \Pr(\bm{V} = \bm{v} \given Z = z)$ for each racial condition $z \in \{0, 1\}$, and profile-blindness (\Cref{def: blindnesses}) is the condition that $\Pr(Z = z \given \bm{V} = \bm{v})$ does not vary with $\bm{v}$. The weightings are well defined because of the proposition's hypothesis that $p(z, \bm{v}) > 0$ on every profile with positive weight under the target. The target's weights are probability distributions, so each racial condition contains at least one profile with positive weight, and therefore $\Pr(Z = z) = \sum_{\bm{v}} p(z, \bm{v}) > 0$ for each racial condition $z$. Conditioning on $\bm{V} = \bm{v}$ is understood at profiles with $\Pr(\bm{V} = \bm{v}) > 0$, and at profiles with $\Pr(\bm{V} = \bm{v}) = 0$ both weightings place zero weight, $g_1(\bm{v}) = g_0(\bm{v}) = 0$.

Part \textit{(ii)} asserts an equivalence, and the proof establishes the equivalence one implication at a time. The first implication is that profile-blindness forces the weightings read off the two racial conditions to coincide; the second is that coinciding weightings force profile-blindness. Both implications rest on two elementary facts. The first fact is that a joint probability factors through either of its conditionals,
\begin{align*}
\Pr(A, B) = \Pr(A \given B) \Pr(B) = \Pr(B \given A) \Pr(A),
\end{align*}
which is the definition of conditional probability rearranged; the proof below refers to the factorization as the multiplication rule. The second fact is the law of total probability. Each implication also establishes, along the way, the same intermediate fact --- derived within each argument, not presupposed --- that a weighting common to the two racial conditions can only be the marginal profile distribution $\Pr(\bm{V} = \bm{v})$.

Suppose first that the assignment is profile-blind, so that $\Pr(Z = z \given \bm{V} = \bm{v})$ equals a constant $c_z$ at every profile with $\Pr(\bm{V} = \bm{v}) > 0$. By the law of total probability, the constant equals the marginal probability of the racial condition,
\begin{align*}
\Pr(Z = z) = \sum_{\bm{v}} \Pr(Z = z \given \bm{V} = \bm{v}) \Pr(\bm{V} = \bm{v}) = c_z \sum_{\bm{v}} \Pr(\bm{V} = \bm{v}) = c_z.
\end{align*}
Then, at every profile with $\Pr(\bm{V} = \bm{v}) > 0$, writing the joint probability by the multiplication rule in each of its two orders,
\begin{align*}
g_z(\bm{v}) &= \dfrac{\Pr(Z = z, \bm{V} = \bm{v})}{\Pr(Z = z)} = \dfrac{\Pr(Z = z \given \bm{V} = \bm{v}) \Pr(\bm{V} = \bm{v})}{\Pr(Z = z)} \\
&= \dfrac{c_z \Pr(\bm{V} = \bm{v})}{c_z} = \Pr(\bm{V} = \bm{v}),
\end{align*}
where the constant $c_z$ cancels between the numerator and the denominator. The right-hand side does not depend on $z$, so $g_1(\bm{v}) = g_0(\bm{v}) = \Pr(\bm{V} = \bm{v})$ at every profile. The weightings read off the two racial conditions therefore coincide, and the common weighting equals the marginal profile distribution, which is the intermediate fact.

Suppose conversely that the weighting read off the assignment is common across the two racial conditions, $g_1 = g_0 \eqqcolon \rho$. The law of total probability over the two racial conditions shows that the common weighting $\rho$ is the marginal profile distribution,
\begin{align*}
\Pr(\bm{V} = \bm{v}) = \sum_{z \in \{0, 1\}} \Pr(\bm{V} = \bm{v} \given Z = z) \Pr(Z = z) = \rho(\bm{v}) \sum_{z \in \{0, 1\}} \Pr(Z = z) = \rho(\bm{v}),
\end{align*}
where the middle equality factors $\rho(\bm{v})$ out of the sum because the common weighting does not depend on the racial condition $z$. Then, at every profile with $\Pr(\bm{V} = \bm{v}) > 0$, again writing the joint probability by the multiplication rule in each of its two orders,
\begin{align*}
\Pr(Z = z \given \bm{V} = \bm{v}) &= \dfrac{\Pr(Z = z, \bm{V} = \bm{v})}{\Pr(\bm{V} = \bm{v})} = \dfrac{g_z(\bm{v}) \Pr(Z = z)}{\Pr(\bm{V} = \bm{v})} \\
&= \dfrac{\rho(\bm{v}) \Pr(Z = z)}{\rho(\bm{v})} = \Pr(Z = z),
\end{align*}
where the common weight $\rho(\bm{v})$ cancels between the numerator and the denominator. The right-hand side does not vary with $\bm{v}$, so the assignment is profile-blind.
\end{proof}

\noindent For a researcher who reads the estimand's weights off the assignment, as part \textit{(ii)} supposes, the intermediate fact of the proof has a design consequence. Each weighting $g_z(\bm{v}) = \Pr(\bm{V} = \bm{v} \given Z = z)$ is the distribution of nonracial profiles within racial condition $z$, and the marginal profile distribution mixes the two weightings in proportion to the probabilities of the racial conditions,
\begin{align*}
\Pr(\bm{V} = \bm{v}) = g_0(\bm{v})\, \Pr(Z = 0) + g_1(\bm{v})\, \Pr(Z = 1),
\end{align*}
by the law of total probability, because each configuration falls under exactly one of the two racial conditions. When $g_1 = g_0 = \rho$, the mixture of the two equal weightings is
\begin{align*}
\rho(\bm{v}) \left[\Pr(Z = 0) + \Pr(Z = 1)\right] = \rho(\bm{v}),
\end{align*}
so the common weighting coincides with the marginal profile distribution --- the distribution written as $\rho$ in \Cref{fig: assignment arrays} of the main text. For the researcher who reads the estimand's weights off the assignment, selecting the common weighting $\rho$ and selecting the assignment's marginal distribution over nonracial profiles are therefore one choice, not two.

\subsection{Identification of the NREC in the perception-based regime} \label{app: NREC identification}

\Cref{sec: excludability} of the main text states \Cref{prop: NREC identification} and \Cref{prop: cue stability decomposition} and points here for the proofs and for the two standard conditions on how the cue would move each unit's perceived race, that is, on the potential outcomes $z_i(0)$ and $z_i(1)$ of the cue.

The first condition rules out defiers, the units that would invert the cue. With Black and White as the two racial conditions of the running example, a defier would perceive the Black cue as White and the White cue as Black.
\begin{assm}[No defiers] \label{assm: no defiers}
For all units $i \in \{1, \ldots, N\}$, $(z_i(0), z_i(1)) \neq (1, 0)$.
\end{assm}
\noindent To illustrate what \Cref{assm: no defiers} permits and what the assumption rules out, consider the name-based racial cue of the running example, with DeShawn coded to read as Black and Connor as White (\Cref{sec: perception-based} of the main text). A decision-maker might read ``Connor'' as Black rather than White, but would then plausibly read ``DeShawn'' as Black as well; a decision-maker who read ``DeShawn'' as White would likewise plausibly read ``Connor'' as White. Either decision-maker would fail to register the cue's intended racial contrast, a failure that \Cref{assm: no defiers} permits; the assumption rules out only the decision-maker who would invert the contrast. Once defiers are ruled out, the denominator of \eqref{eq: NREC iden ratio} of the main text --- the cue's average effect on perceived race --- equals the share of compliers, with no defiers to offset them.

The second condition is a positive complier share. The cue must move at least one unit's perceived race in the direction the cue intends --- in the running example, toward White under the White cue and toward Black under the Black cue --- and movement alone does not suffice, because a defier's perceived race also moves, in the opposite direction.
\begin{assm}[Existence of compliers] \label{assm: compliers exist}
At least one of the $i \in \{1, \ldots, N\}$ units has $z_i(0) = 0$ and $z_i(1) = 1$, i.e., $\mathcal{C} \neq \emptyset$.
\end{assm}
\noindent Without such a complier the complier share is zero; the denominator of \eqref{eq: NREC iden ratio} then vanishes and the NREC, an average over $\mathcal{C}$, is undefined.

Under \Cref{assm: no defiers} and \Cref{assm: compliers exist}, together with the hypotheses that \Cref{prop: NREC identification} states in the main text --- the cue-level SUTVA (\Cref{assm: SUTVA cue}), perception sufficiency, and \Cref{assm: no isolated nonracial shift} --- the proposition follows.

\begin{proof}[Proof of \Cref{prop: NREC identification}]
The proof shows that non-compliers contribute zero to both the numerator and the denominator of the ratio in \eqref{eq: NREC iden ratio}, so that the numerator and the denominator each reduce to a sum over compliers, and the ratio of those two sums is the NREC.

By \Cref{assm: no defiers}, each unit belongs to one of three types: a complier ($z_i(0) = 0, z_i(1) = 1$), an always-taker ($z_i(0) = z_i(1) = 1$), or a never-taker ($z_i(0) = z_i(1) = 0$). For non-compliers, $z_i(0) = z_i(1)$, and by \Cref{assm: no isolated nonracial shift}, $\bm{v}_i(0) = \bm{v}_i(1)$ as well, so $y_i\left(z_i(1), \bm{v}_i(1)\right) - y_i\left(z_i(0), \bm{v}_i(0)\right) = 0$ and $z_i(1) - z_i(0) = 0$ for every non-complier. The numerator and denominator of \eqref{eq: NREC iden ratio} therefore reduce to sums over compliers:
\begin{align*}
\frac{1}{N} \sum_{i = 1}^N \left[y_i\left(z_i(1), \bm{v}_i(1)\right) - y_i\left(z_i(0), \bm{v}_i(0)\right)\right] & = \frac{1}{N} \sum_{i \in \mathcal{C}} \left[y_i\left(1, \bm{v}_i(1)\right) - y_i\left(0, \bm{v}_i(0)\right)\right], \\
\frac{1}{N} \sum_{i = 1}^N \left[z_i(1) - z_i(0)\right] & = \frac{\abs{\mathcal{C}}}{N}.
\end{align*}
By \Cref{assm: compliers exist}, $\abs{\mathcal{C}} \geq 1$, so the denominator $\abs{\mathcal{C}}/N$ is positive and the ratio is well defined. Taking the ratio of the two sums over compliers,
\begin{align*}
\dfrac{\frac{1}{N} \sum_{i \in \mathcal{C}} \left[y_i\left(1, \bm{v}_i(1)\right) - y_i\left(0, \bm{v}_i(0)\right)\right]}{\abs{\mathcal{C}}/N}
& = \dfrac{1}{\abs{\mathcal{C}}} \sum_{i \in \mathcal{C}} \left[y_i\left(1, \bm{v}_i(1)\right) - y_i\left(0, \bm{v}_i(0)\right)\right],
\end{align*}
and applying the definition of the NREC in \eqref{eq: complier effect} of the main text yields the result.
\end{proof}

The numerator of \eqref{eq: NREC iden ratio} --- the cue's average effect on the outcome over all $N$ units --- is a lower bound in magnitude on the NREC. In practice, many cue-based studies cannot compute the ratio: The denominator --- the cue's average effect on perceived race --- requires perceived race to be measured on the decision-makers who supply the outcome, and most studies measure perceived race, if at all, in a separate sample \citep{gaddis2017a, landgraveweller2022, elderhayes2023}. Such studies can estimate the numerator alone, and the numerator alone is still informative.
\begin{cor}[The reduced-form effect bounds the NREC] \label{cor: NREC lower bound}
Under the conditions of \Cref{prop: NREC identification}, the numerator of \eqref{eq: NREC iden ratio} equals $(\abs{\mathcal{C}}/N) \, \textrm{NREC}$, where the complier share $\abs{\mathcal{C}}/N$ lies in $(0, 1]$.
\end{cor}
\noindent \Cref{cor: NREC lower bound} follows from a single equation. Writing $\Delta_Y$ for the numerator of the ratio in \eqref{eq: NREC iden ratio}, the equation is $\Delta_Y = (\abs{\mathcal{C}}/N)\, \textrm{NREC}$. In words, the numerator $\Delta_Y$ averages the cue's effects on the outcome over all $N$ units, and every non-complier contributes zero to that average, so the average dilutes the compliers' average effect by the factor $\abs{\mathcal{C}}/N$, the complier share.

The proof below derives the equation $\Delta_Y = (\abs{\mathcal{C}}/N)\, \textrm{NREC}$ and records three of the equation's consequences. First, because the complier share is at most one, the dilution of the compliers' average effect cannot enlarge that effect, so $\abs{\Delta_Y} \leq \abs{\textrm{NREC}}$. Second, the two magnitudes $\abs{\Delta_Y}$ and $\abs{\textrm{NREC}}$ are equal exactly when the dilution changes nothing, either because every unit is a complier ($\abs{\mathcal{C}} = N$) or because there is no effect to dilute ($\textrm{NREC} = 0$). Third, because the complier share is positive, dividing $\Delta_Y$ by the complier share undoes the dilution and recovers the NREC, $\textrm{NREC} = (N / \abs{\mathcal{C}})\, \Delta_Y$, with the factor $N / \abs{\mathcal{C}}$ confined to the finite set $\{N/k : k = 1, \ldots, N\}$.

\begin{proof}[Proof of \Cref{cor: NREC lower bound}]
The proof shows that the denominator of the ratio in \eqref{eq: NREC iden ratio} equals the complier share and combines that value with \Cref{prop: NREC identification} to obtain $\Delta_Y = (\abs{\mathcal{C}}/N)\, \textrm{NREC}$; the three consequences then follow from $0 < \abs{\mathcal{C}}/N \leq 1$.

Write $\Delta_Y \coloneqq \frac{1}{N} \sum_{i = 1}^N \left[y_i\left(z_i(1), \bm{v}_i(1)\right) - y_i\left(z_i(0), \bm{v}_i(0)\right)\right]$ for the numerator of \eqref{eq: NREC iden ratio} and $\Delta_Z \coloneqq \frac{1}{N} \sum_{i = 1}^N \left[z_i(1) - z_i(0)\right]$ for its denominator. As established in the proof of \Cref{prop: NREC identification}, each unit is a complier, always-taker, or never-taker (\Cref{assm: no defiers}), contributing $z_i(1) - z_i(0)$ equal to $1$, $0$, and $0$ respectively, so $\Delta_Z = \abs{\mathcal{C}}/N$. \Cref{prop: NREC identification} gives $\Delta_Y / \Delta_Z = \textrm{NREC}$, hence $\Delta_Y = \Delta_Z\, \textrm{NREC} = (\abs{\mathcal{C}}/N)\, \textrm{NREC}$.

By \Cref{assm: compliers exist}, $\abs{\mathcal{C}} \geq 1$, and $\abs{\mathcal{C}} \leq N$ since $\mathcal{C} \subseteq \{1, \ldots, N\}$; therefore $\abs{\mathcal{C}}/N \in (0, 1]$.
Taking absolute values of both sides of $\Delta_Y = (\abs{\mathcal{C}}/N)\, \textrm{NREC}$,
\begin{align*}
\abs{\Delta_Y} = \dfrac{\abs{\mathcal{C}}}{N}\, \abs{\textrm{NREC}} \leq \abs{\textrm{NREC}},
\end{align*}
because $\abs{\mathcal{C}}/N \leq 1$.
The inequality holds with equality exactly when $\abs{\mathcal{C}}/N = 1$ --- equivalently, $\abs{\mathcal{C}} = N$ --- or when $\textrm{NREC} = 0$. Finally, because $\abs{\mathcal{C}}/N > 0$, dividing both sides of $\Delta_Y = (\abs{\mathcal{C}}/N)\, \textrm{NREC}$ by the complier share recovers the NREC from the numerator alone, $\textrm{NREC} = (N / \abs{\mathcal{C}})\, \Delta_Y$; as $\abs{\mathcal{C}}$ ranges over the integers $\{1, \ldots, N\}$, the factor $N / \abs{\mathcal{C}}$ ranges over $\{N/k : k = 1, \ldots, N\}$.
\end{proof}

An estimator that uses the numerator alone therefore shares the NREC's sign but understates the NREC's magnitude: The estimator is conservative, guaranteed never to overstate. \citet{kaufmanetal2026} suggest that the understatement may be substantial, and a researcher could treat the unknown complier share as a sensitivity parameter, reassessing the implied NREC across the share's possible values, analogous to the sensitivity analyses in \citet{leavittrivera-burgos2024,leavittrivera-burgos2026}.

Finally, \Cref{prop: cue stability decomposition} of the main text decomposes cue stability into two conjuncts and assigns a role to each conjunct of the decomposition.

\begin{proof}[Proof of \Cref{prop: cue stability decomposition}]
The proof proceeds in three parts: The partition of units into compliers and non-compliers splits cue stability into the two conjuncts; conjunct (i) is then shown to be necessary and sufficient for the ratio to equal the NREC; and, given conjunct (i), conjunct (ii) is shown to be necessary and sufficient for the NREC to equal the AEE among compliers. Each sufficiency claim holds for all potential outcomes, and each necessity claim is established by constructing potential outcomes under which the conjunct fails and the corresponding equality fails with it.

By \Cref{assm: no defiers}, the units partition into compliers ($z_i(0) = 0$, $z_i(1) = 1$) and non-compliers ($z_i(0) = z_i(1)$). Cue stability requires $\bm{v}_i(0) = \bm{v}_i(1)$ for every unit, so cue stability holds if and only if the equality $\bm{v}_i(0) = \bm{v}_i(1)$ holds on each cell of the partition: on non-compliers, which is \Cref{assm: no isolated nonracial shift} of the main text, conjunct (i), and on compliers, which is conjunct (ii).

Sufficiency of conjunct (i) for the ratio in \eqref{eq: NREC iden ratio} to equal the NREC is \Cref{prop: NREC identification}. \Cref{prop: NREC identification} concludes that the ratio equals the NREC for all potential outcomes satisfying the proposition's assumptions, so, to show that conjunct (i) is necessary for that conclusion, it suffices to construct one set of potential outcomes that satisfies every assumption except conjunct (i) and under which the ratio differs from the NREC. The argument below constructs such a set of potential outcomes.

Suppose that conjunct (i) fails for some non-complier $j$, so that $\bm{v}_j(0) \neq \bm{v}_j(1)$. Consider potential outcomes with $y_j\left(z_j(0), \bm{v}_j(1)\right) - y_j\left(z_j(0), \bm{v}_j(0)\right) = \kappa \neq 0$, with every complier's contrast equal to zero, and with the corresponding term $y_i\left(z_i(0), \bm{v}_i(1)\right) - y_i\left(z_i(0), \bm{v}_i(0)\right)$ equal to zero for every non-complier $i \neq j$. The denominator of \eqref{eq: NREC iden ratio} still equals $\abs{\mathcal{C}}/N > 0$, while the numerator equals $\kappa/N \neq 0$ because unit $j$'s term is the only nonzero term, so the ratio equals $\kappa/\abs{\mathcal{C}} \neq 0 = \textrm{NREC}$.

The third part of the proof takes conjunct (i) as given and shows that conjunct (ii) is necessary and sufficient for the NREC to equal the AEE among compliers. Given conjunct (i), the ratio in \eqref{eq: NREC iden ratio} equals the NREC by the sufficiency just established, so recovery is settled, and the remaining question is whether the NREC equals the AEE among compliers --- the average over $\mathcal{C}$ of the all-else-equal contrast evaluated, for each complier, at the nonracial perception the complier would hold under cue value $0$, $\frac{1}{\abs{\mathcal{C}}} \sum_{i \in \mathcal{C}} \left[y_i\left(1, \bm{v}_i(0)\right) - y_i\left(0, \bm{v}_i(0)\right)\right]$. Subtracting the AEE among compliers from the NREC in \eqref{eq: complier effect} term by term,
\begin{align*}
& \dfrac{1}{\abs{\mathcal{C}}} \sum_{i \in \mathcal{C}} \left[y_i\left(1, \bm{v}_i(1)\right) - y_i\left(0, \bm{v}_i(0)\right)\right] - \dfrac{1}{\abs{\mathcal{C}}} \sum_{i \in \mathcal{C}} \left[y_i\left(1, \bm{v}_i(0)\right) - y_i\left(0, \bm{v}_i(0)\right)\right] \\
& = \dfrac{1}{\abs{\mathcal{C}}} \sum_{i \in \mathcal{C}} \left\{\left[y_i\left(1, \bm{v}_i(1)\right) - y_i\left(0, \bm{v}_i(0)\right)\right] - \left[y_i\left(1, \bm{v}_i(0)\right) - y_i\left(0, \bm{v}_i(0)\right)\right]\right\} \\
& = \dfrac{1}{\abs{\mathcal{C}}} \sum_{i \in \mathcal{C}} \left[y_i\left(1, \bm{v}_i(1)\right) - y_i\left(1, \bm{v}_i(0)\right)\right],
\end{align*}
where the first equality combines the two sums term by term and the second cancels the terms $y_i\left(0, \bm{v}_i(0)\right)$ within each summand. The difference between the NREC and the AEE among compliers therefore equals the final sum, and both directions of the third part read off that sum.

To establish sufficiency, suppose conjunct (ii) holds. Every term of the final sum vanishes because $\bm{v}_i(1) = \bm{v}_i(0)$ for each complier $i \in \mathcal{C}$, so the NREC equals the AEE among compliers for all potential outcomes.

To establish necessity, suppose conjunct (ii) fails for some complier $j$, so that $\bm{v}_j(1) \neq \bm{v}_j(0)$; a counterexample analogous to the one for conjunct (i) then breaks the equality. Consider potential outcomes with $y_j\left(1, \bm{v}_j(1)\right) - y_j\left(1, \bm{v}_j(0)\right) \neq 0$ and with the corresponding terms of all other compliers equal to zero. The final sum is then nonzero, so the NREC differs from the AEE among compliers, and the equality cannot hold for all potential outcomes without conjunct (ii).
\end{proof}

\paragraph{Why conjunct (ii) has no classical counterpart.} The identification of the NREC in \Cref{prop: NREC identification} has the same structure as the identification of local average treatment effects \citep{imbensangrist1994, angristetal1996, gerbergreen2012}: The cue $w$ is the instrument, and the perceived configuration $(z_i(w), \bm{v}_i(w))$ --- racial and nonracial components together --- is the treatment. Yet the classical treatment of instrumental variables demands no condition like conjunct \textit{(ii)}, and the reason is the dimension of the treatment. In the classical setting the treatment is a single variable, the analogue of perceived race $z_i(w)$ alone, so compliance fully specifies a complier's treatment at both values of the instrument, $z_i(0) = 0$ and $z_i(1) = 1$, and nothing about what compliers would receive is left to assume. Under the perception-based regime, compliance is defined by the racial component alone, so the same specification, $z_i(0) = 0$ and $z_i(1) = 1$, leaves the nonracial component $\bm{v}_i(w)$ of what compliers would perceive unspecified. Conjunct \textit{(ii)} supplies the specification, $\bm{v}_i(0) = \bm{v}_i(1)$ for each complier $i \in \mathcal{C}$.

A condition that specifies which treatment the compliers would receive selects the estimand rather than recovers it. Under conjunct \textit{(ii)}, each complier's induced contrast $y_i\left(1, \bm{v}_i(1)\right) - y_i\left(0, \bm{v}_i(0)\right)$ becomes $y_i\left(1, \bm{v}_i(0)\right) - y_i\left(0, \bm{v}_i(0)\right)$, a contrast that differs in perceived race alone, so the NREC collapses into the AEE among compliers. Recovery needs no such condition, because, given conjunct \textit{(i)} --- $\bm{v}_i(0) = \bm{v}_i(1)$ for every non-complier --- the ratio in \eqref{eq: NREC iden ratio} equals the NREC with or without conjunct \textit{(ii)} (\Cref{prop: cue stability decomposition}). The distinction between selecting and recovering the estimand is the one that \Cref{sec: excludability} of the main text draws from the decomposition of cue stability, and the classical setting could leave the distinction implicit because a one-dimensional treatment leaves conjunct \textit{(ii)} nothing to specify.

\subsection{The mixed contrast} \label{sec: mixed contrast formalism}

\Cref{sec: choosing contrast} of the main text introduces the mixed contrast as one that holds some nonracial features fixed across racial conditions and lets the rest vary by race. This subsection defines the partition of the nonracial features into a held-fixed set and a free set. From the joint PMF $q(z, \bm{v})$ and the common weighting $\rho$ of \Cref{sec: anatomy estimand} of the main text, the subsection then derives the two distributions a mixed contrast needs, namely the conditional distribution of the free features given the held-fixed features within each racial condition and the marginal distribution of the held-fixed features.

Partition the indices of the nonracial features, $\{1, \ldots, L\} = \mathcal{I}_{\mathrm{fix}} \cup \mathcal{I}_{\mathrm{free}}$ with $\mathcal{I}_{\mathrm{fix}} \cap \mathcal{I}_{\mathrm{free}} = \emptyset$. Write $\bm{v} = (\bm{v}_{\mathrm{fix}}, \bm{v}_{\mathrm{free}})$, with $\bm{v}_{\mathrm{fix}} \in \mathcal{V}_{\mathrm{fix}} \coloneqq \prod_{\ell \in \mathcal{I}_{\mathrm{fix}}} \mathcal{V}_\ell$ collecting the features held fixed and $\bm{v}_{\mathrm{free}} \in \mathcal{V}_{\mathrm{free}} \coloneqq \prod_{\ell \in \mathcal{I}_{\mathrm{free}}} \mathcal{V}_\ell$ collecting the free features. The AEE is the case $\mathcal{I}_{\mathrm{free}} = \emptyset$ (every feature held fixed) and the AEWR the case $\mathcal{I}_{\mathrm{fix}} = \emptyset$ (no feature held fixed).

When $\mathcal{I}_{\mathrm{fix}}$ and $\mathcal{I}_{\mathrm{free}}$ are both nonempty, the weighting follows the partition, as \Cref{sec: choosing contrast} of the main text describes. A distribution common to the two racial conditions weights the held-fixed features, as the common weighting $\rho$ weights all of $\bm{v}$ for the AEE, and race-specific distributions weight the free features given the held-fixed value, as the $q_z$ weight all of $\bm{v}$ for the AEWR. The construction below gives the two distributions and then assembles the mixed effect in a single expression.

Within racial condition $z$ at held-fixed value $\bm{v}_{\mathrm{fix}}$, the free features follow the conditional distribution
\begin{align} \label{eq: qz cond vE}
q_z(\bm{v}_{\mathrm{free}} \given \bm{v}_{\mathrm{fix}}) \coloneqq \dfrac{q(z, \bm{v}_{\mathrm{fix}}, \bm{v}_{\mathrm{free}})}{\sum_{\bm{v}_{\mathrm{free}}^{\prime} \in \mathcal{V}_{\mathrm{free}}} q(z, \bm{v}_{\mathrm{fix}}, \bm{v}_{\mathrm{free}}^{\prime})},
\end{align}
which the joint PMF $q(z, \bm{v})$ of \Cref{sec: anatomy estimand} of the main text induces, and which is defined at each pair of a racial condition and a held-fixed value for which the denominator of \eqref{eq: qz cond vE} is positive. The held-fixed features play the role that the whole profile plays in the AEE, so their weights must be common to the two racial conditions. The aggregation over the held-fixed values uses the marginal distribution of the held-fixed features under the common weighting $\rho$ of \Cref{sec: anatomy estimand} of the main text,
\begin{align} \label{eq: vE marginal dist}
\rho(\bm{v}_{\mathrm{fix}}) \coloneqq \sum_{\bm{v}_{\mathrm{free}} \in \mathcal{V}_{\mathrm{free}}} \rho(\bm{v}_{\mathrm{fix}}, \bm{v}_{\mathrm{free}}).
\end{align}

For unit $i$, the race effect under the mixed contrast combines the two distributions in one expression,
\begin{equation} \label{eq: contrasts part margin}
\begin{split}
\sum_{\bm{v}_{\mathrm{fix}} \in \mathcal{V}_{\mathrm{fix}}} \Bigg[\sum_{\bm{v}_{\mathrm{free}} \in \mathcal{V}_{\mathrm{free}}} & y_i(1, \bm{v}_{\mathrm{fix}}, \bm{v}_{\mathrm{free}})\, q_1(\bm{v}_{\mathrm{free}} \given \bm{v}_{\mathrm{fix}}) \\
{} - \sum_{\bm{v}_{\mathrm{free}} \in \mathcal{V}_{\mathrm{free}}} & y_i(0, \bm{v}_{\mathrm{fix}}, \bm{v}_{\mathrm{free}})\, q_0(\bm{v}_{\mathrm{free}} \given \bm{v}_{\mathrm{fix}})\Bigg] \rho(\bm{v}_{\mathrm{fix}}).
\end{split}
\end{equation}
Read the expression from the inside out. Each inner sum averages unit $i$'s potential outcomes over the free features within one racial condition at the held-fixed value $\bm{v}_{\mathrm{fix}}$ with the weights in \eqref{eq: qz cond vE}. The difference between the two inner sums contrasts the racial conditions at that common held-fixed value, so the contrast is all-else-equal on the held-fixed features and within-race on the free features. The outer sum aggregates the contrasts over the held-fixed values with the common weights in \eqref{eq: vE marginal dist}.

Averaging the race effect \eqref{eq: contrasts part margin} over the $N$ units gives the mixed effect, and collecting the weight that multiplies each potential outcome $y_i(z, \bm{v}_{\mathrm{fix}}, \bm{v}_{\mathrm{free}})$ shows that the mixed effect is the member of the estimand family with the mixed weights
\begin{align} \label{eq: mixed weights}
g_z(\bm{v}) = \rho(\bm{v}_{\mathrm{fix}})\, q_z(\bm{v}_{\mathrm{free}} \given \bm{v}_{\mathrm{fix}}),
\end{align}
which are common on the held-fixed features and race-specific on the free features given the held-fixed value. \Cref{sec: gap} substitutes the mixed weights in \eqref{eq: mixed weights} into the gap identity \eqref{eq: gap identity} to obtain the gap between a mixed effect and any other member of the estimand family --- the AEE, the AEWR, or a mixed effect built on a different partition of the nonracial features.

\subsection{The structural tradeoff} \label{sec: tradeoff}

\Cref{sec: substantive motivations} of the main text argues that, when racial classifications index different regions of the attribute space, no common all-else-equal weighting can both restrict attention to profiles typical of both groups and represent the typical members of each group. This subsection makes that claim precise and proves it.

Fix $\varepsilon \in (0, 1)$. For each race $z \in \{0, 1\}$, call a set $\mathcal{R}_z \subseteq \mathcal{V}$ a \textit{typical region for race $z$} if
\begin{align*}
q_z\left(\mathcal{R}_z\right) \coloneqq \sum_{\bm{v} \in \mathcal{R}_z} q_z(\bm{v}) \geq 1 - \varepsilon,
\end{align*}
so that the profiles of at least a $1 - \varepsilon$ share of individuals with race $z$ lie in $\mathcal{R}_z$. On the thick constructivist conception of \Cref{sec: debate} of the main text, a racial category indexes a distribution of nonracial attributes; the typical region for race $z$ is the region of the attribute space on which the indexed distribution $q_z$ concentrates. The intersection $\mathcal{R}_0 \cap \mathcal{R}_1$ is the \textit{both-typical region}.

\begin{prop}[Structural tradeoff] \label{prop: tradeoff}
Fix $\varepsilon \in (0, 1)$ and typical regions $\mathcal{R}_0, \mathcal{R}_1 \subseteq \mathcal{V}$. Suppose that for at least one race $z \in \{0, 1\}$,
\begin{align} \label{eq: tradeoff hypothesis}
q_z\left(\mathcal{R}_0 \cap \mathcal{R}_1\right) < 1 - \varepsilon.
\end{align}
Then no weighting distribution $\rho$ on $\mathcal{V}$, common across the two racial conditions, satisfies both of the following constraints:
\begin{enumerate}
\item[\textnormal{(i)}] $\operatorname{supp}(\rho) \subseteq \mathcal{R}_0 \cap \mathcal{R}_1$, so that $\rho$ places all of its weight on both-typical profiles; and
\item[\textnormal{(ii)}] $q_z\left(\operatorname{supp}(\rho)\right) \geq 1 - \varepsilon$ for each race $z \in \{0, 1\}$, so that the support of $\rho$ retains at least a $1 - \varepsilon$ share of each group.
\end{enumerate}
\end{prop}

\begin{proof}
The proof supposes constraint (i) and derives the negation of constraint (ii) from the proposition's hypothesis, the inequality \eqref{eq: tradeoff hypothesis}.

Let $z^{\star} \in \{0, 1\}$ be a race for which \eqref{eq: tradeoff hypothesis} holds, and suppose a common weighting distribution $\rho$ satisfies constraint (i), so that $\operatorname{supp}(\rho) \subseteq \mathcal{R}_0 \cap \mathcal{R}_1$. For any set of profiles $\mathcal{A} \subseteq \mathcal{V}$, the quantity $q_{z^{\star}}(\mathcal{A}) = \sum_{\bm{v} \in \mathcal{A}} q_{z^{\star}}(\bm{v})$ is the share of individuals with race $z^{\star}$ whose profiles lie in $\mathcal{A}$. Each term of the sum is nonnegative, so enlarging $\mathcal{A}$ to any set $\mathcal{B}$ with $\mathcal{A} \subseteq \mathcal{B} \subseteq \mathcal{V}$ adds terms and removes none; hence $q_{z^{\star}}\left(\mathcal{A}\right) \leq q_{z^{\star}}\left(\mathcal{B}\right)$ (monotonicity with respect to set inclusion).

Take $\mathcal{A} = \operatorname{supp}(\rho)$ and $\mathcal{B} = \mathcal{R}_0 \cap \mathcal{R}_1$; the inclusion $\mathcal{A} \subseteq \mathcal{B}$ is then exactly constraint (i). Monotonicity with respect to set inclusion at these two sets, $q_{z^{\star}}\left(\operatorname{supp}(\rho)\right) \leq q_{z^{\star}}\left(\mathcal{R}_0 \cap \mathcal{R}_1\right)$, gives the weak inequality, and the hypothesis \eqref{eq: tradeoff hypothesis}, $q_{z^{\star}}\left(\mathcal{R}_0 \cap \mathcal{R}_1\right) < 1 - \varepsilon$, gives the strict inequality, in the chain of inequalities
\begin{align} \label{eq: tradeoff chain}
q_{z^{\star}}\left(\operatorname{supp}(\rho)\right) \;\leq\; q_{z^{\star}}\left(\mathcal{R}_0 \cap \mathcal{R}_1\right) \;<\; 1 - \varepsilon.
\end{align}

In the chain \eqref{eq: tradeoff chain}, a weak inequality is followed by a strict inequality, and the two inequalities compose by transitivity into the strict inequality between the ends, $q_{z^{\star}}\left(\operatorname{supp}(\rho)\right) < 1 - \varepsilon$. The strict inequality between the ends states, in words, that the support of $\rho$ retains less than a $1 - \varepsilon$ share of the individuals with race $z^{\star}$, and the statement is the negation of constraint (ii) at $z = z^{\star}$. Hence no common weighting distribution $\rho$ satisfies constraints (i) and (ii) together.
\end{proof}

\Cref{prop: tradeoff} admits two logically equivalent statements --- the proposition as stated and the proposition's contrapositive --- and each statement shows what imposing one of the two constraints forces the researcher to give up. As stated, the proposition says that if the common weighting $\rho$ places all of its weight on both-typical profiles, as constraint (i) requires, then the support of $\rho$ retains less than a $1 - \varepsilon$ share of at least one group, so constraint (ii) fails. In the contrapositive, the proposition says that if the common weighting $\rho$ retains at least a $1 - \varepsilon$ share of each group, as constraint (ii) requires, then $\rho$ must place positive weight on profiles outside the typical region for at least one race, so constraint (i) fails.

\textit{Worked example.} A stylized version of the running example shows, in numbers, what each constraint of \Cref{prop: tradeoff} forces the researcher to give up. The example follows the shape that \Cref{fig: structural tradeoff} of the main text depicts: two distributions of neighborhoods and prior records that overlap only slightly (\Cref{sec: substantive motivations} of the main text). Code the two nonracial features of the running example as binary --- $v_1 = 1$ for a disadvantaged neighborhood and $v_2 = 1$ for an extensive prior record --- so that the profiles $\bm{v} = (v_1, v_2)$ take four possible values. Suppose that, among arrestees, Black, relative to White, indexes a greater likelihood of arrest in a disadvantaged neighborhood and a greater likelihood of an extensive prior record. The distribution $q_1$ of profiles among Black arrestees then places its largest share on the profile with both features, smaller shares on the two profiles with one feature, and its smallest share on the profile with neither:
\begin{align*}
q_1(1,1) = .60, \quad q_1(1,0) = .20, \quad q_1(0,1) = .15, \quad q_1(0,0) = .05.
\end{align*}
The distribution $q_0$ of profiles among White arrestees concentrates on the profile with neither feature and places the same shares as $q_1$ on the single-feature profiles: $q_0(0,0) = .60$, $q_0(1,0) = .20$, $q_0(0,1) = .15$, and $q_0(1,1) = .05$. For both groups, the profile with an extensive prior record but no disadvantaged neighborhood is the rarer of the two single-feature profiles, $q_z(0,1) = .15 < .20 = q_z(1,0)$, because an extensive prior record arises partly from the greater policing of disadvantaged neighborhoods.

Fix $\varepsilon = .10$, so that a typical region must hold at least a $.90$ share of its group. The set $\mathcal{R}_1 = \{(1,1), (1,0), (0,1)\}$ is a typical region for race $1$, with $q_1(\mathcal{R}_1) = .95 \geq .90$, and the set $\mathcal{R}_0 = \{(0,0), (0,1), (1,0)\}$ is a typical region for race $0$, with $q_0(\mathcal{R}_0) = .95$. The both-typical region is the pair of single-feature profiles, $\mathcal{R}_0 \cap \mathcal{R}_1 = \{(1,0), (0,1)\}$. The both-typical region holds only a $.35$ share of each group, $q_z\left(\mathcal{R}_0 \cap \mathcal{R}_1\right) = .35 < .90$, so the hypothesis \eqref{eq: tradeoff hypothesis} holds for both races.

A common weighting $\rho$ that satisfies constraint (i) places all of its weight on the two single-feature profiles, so the support of $\rho$ retains at most a $.35$ share of each group --- far below the $.90$ share that constraint (ii) demands. A common weighting that satisfies constraint (ii) must retain at least a $.90$ share of group $1$, and the support of any such weighting includes the profile $(1,1)$, because the three profiles other than $(1,1)$ together hold only a $.40$ share of group $1$. The profile $(1,1)$ carries a $.05$ share of group $0$, so the weighting places weight on a profile rare under group $0$'s distribution. The same argument with the groups exchanged forces weight on $(0,0)$, a profile rare under group $1$'s distribution.

Whenever the hypothesis \eqref{eq: tradeoff hypothesis} holds, so that the both-typical region contains less than a $1 - \varepsilon$ share of at least one group, constraints (i) and (ii) of \Cref{prop: tradeoff} cannot hold together, and the researcher must give up either the restriction to both-typical profiles or the representation of each group's typical members.

\subsection{Distribution specifications} \label{app: distribution formalism}

\Cref{sec: weighting substantive} of the main text motivates specifications of $q_z$ and $\rho$; this subsection develops the specifications in full. The three parts below take up the restriction of the support to exclude implausible profiles, the information that calibrating the weighting distributions to a reference population requires, and the translation of graded typicality into a distribution that concentrates around each group's prototype.

\subsubsection{Restricting the support}

Restrictions on implausible attribute combinations operate directly on the joint PMF $q(z,\bm{v}_{\mathrm{fix}},\bm{v}_{\mathrm{free}})$ of \Cref{sec: anatomy estimand} of the main text, defined on the space $\mathcal{T}$ of configurations. \citet[pp.~13--16]{hainmuelleretal2014} build randomization restrictions for implausible combinations into the canonical conjoint framework and note that causal effects are then defined only over the combinations the design retains (p.~20).

The researcher judges some configurations in $\mathcal{T}$ implausible and retains the rest in the comparison. Let $\mathcal{T}^{\ast} \subseteq \mathcal{T}$ denote the set of retained configurations. Excluding the implausible configurations from the comparison is a choice about the estimand's weighting, not a claim that the excluded configurations cannot occur: The new joint PMF $q^{\ast}$ places zero weight on every configuration outside $\mathcal{T}^{\ast}$ and rescales the weights that $q$ places on the configurations in $\mathcal{T}^{\ast}$ so that the rescaled weights sum to one. Formally, define
\begin{align} \label{eq: q star def}
q^{\ast}(z,\bm{v}_{\mathrm{fix}},\bm{v}_{\mathrm{free}}) \coloneqq\dfrac{q(z,\bm{v}_{\mathrm{fix}},\bm{v}_{\mathrm{free}}) \mathbbm{1}\{(z,\bm{v}_{\mathrm{fix}},\bm{v}_{\mathrm{free}}) \in \mathcal{T}^{\ast}\}}{\sum_{(z^{\prime},\bm{v}_{\mathrm{fix}}^{\prime},\bm{v}_{\mathrm{free}}^{\prime}) \in \mathcal{T}^{\ast}} q(z^{\prime},\bm{v}_{\mathrm{fix}}^{\prime},\bm{v}_{\mathrm{free}}^{\prime})},
\end{align}
whenever the denominator is positive.

For a configuration retained in $\mathcal{T}^{\ast}$, the indicator in the numerator of \eqref{eq: q star def} equals one, so $q^{\ast}$ returns the weight $q(z,\bm{v}_{\mathrm{fix}},\bm{v}_{\mathrm{free}})$ divided by the total weight that $q$ places on $\mathcal{T}^{\ast}$. For a configuration outside $\mathcal{T}^{\ast}$, the indicator equals zero, so $q^{\ast}$ returns zero. The values of $q^{\ast}$ are therefore nonnegative and sum to one over $\mathcal{T}^{\ast}$, so $q^{\ast}$ is a PMF on $\mathcal{T}^{\ast}$, proportional to $q$ on the configurations retained in $\mathcal{T}^{\ast}$. All subsequent quantities --- including the conditional distribution $q_z(\bm{v}_{\mathrm{free}} \given \bm{v}_{\mathrm{fix}})$ in \eqref{eq: qz cond vE} and the marginal distribution $\rho(\bm{v}_{\mathrm{fix}})$ in \eqref{eq: vE marginal dist} --- are then computed from $q^{\ast}$ rather than from $q$.

Replacing $q$ with $q^{\ast}$ leaves the formal structure of the mixed effect in \eqref{eq: contrasts part margin} unchanged. The replacement can nonetheless leave parts of the mixed effect with nothing to average or to compare, because each inner average in \eqref{eq: contrasts part margin} draws its terms from the configurations retained in $\mathcal{T}^{\ast}$ alone. To state when each part of the mixed effect exists, collect, for racial condition $z$ at held-fixed value $\bm{v}_{\mathrm{fix}}$, the free-feature values that the restriction to $\mathcal{T}^{\ast}$ leaves available,
\begin{align} \label{eq: retained free values}
\mathcal{V}^{\ast}_{\mathrm{free}}(z, \bm{v}_{\mathrm{fix}}) \coloneqq \left\{\bm{v}_{\mathrm{free}} \in \mathcal{V}_{\mathrm{free}} : (z, \bm{v}_{\mathrm{fix}}, \bm{v}_{\mathrm{free}}) \in \mathcal{T}^{\ast}\right\}.
\end{align}

Two requirements for the existence of the mixed effect follow. First, under $q^{\ast}$, the inner average over the free features within racial condition $z$ at held-fixed value $\bm{v}_{\mathrm{fix}}$ becomes
\begin{align*}
\sum_{\bm{v}_{\mathrm{free}} \in \mathcal{V}^{\ast}_{\mathrm{free}}(z, \bm{v}_{\mathrm{fix}})} y_i(z, \bm{v}_{\mathrm{fix}}, \bm{v}_{\mathrm{free}})\, q^{\ast}_z(\bm{v}_{\mathrm{free}} \given \bm{v}_{\mathrm{fix}}),
\end{align*}
because $q^{\ast}_z(\bm{v}_{\mathrm{free}} \given \bm{v}_{\mathrm{fix}})$ is positive only at the free-feature values in $\mathcal{V}^{\ast}_{\mathrm{free}}(z, \bm{v}_{\mathrm{fix}})$. The inner average therefore exists only when $\mathcal{V}^{\ast}_{\mathrm{free}}(z, \bm{v}_{\mathrm{fix}}) \neq \emptyset$. Second, the contrast at held-fixed value $\bm{v}_{\mathrm{fix}}$ subtracts the inner average for racial condition $0$ from the inner average for racial condition $1$. The contrast therefore exists only when both sets of free-feature values are nonempty, $\mathcal{V}^{\ast}_{\mathrm{free}}(0, \bm{v}_{\mathrm{fix}}) \neq \emptyset$ and $\mathcal{V}^{\ast}_{\mathrm{free}}(1, \bm{v}_{\mathrm{fix}}) \neq \emptyset$.

If $\mathcal{V}^{\ast}_{\mathrm{free}}(z, \bm{v}_{\mathrm{fix}})$ is empty for exactly one of the two racial conditions --- as when a combination of features is judged plausible for one group alone --- the contrast at that held-fixed value is undefined. The outer summation in \eqref{eq: contrasts part margin} weights each held-fixed value by the restricted marginal distribution $\rho^{\ast}(\bm{v}_{\mathrm{fix}})$, which places positive weight only on the held-fixed values that appear in at least one configuration retained in $\mathcal{T}^{\ast}$. The mixed effect under $q^{\ast}$ therefore exists when the second requirement holds at every held-fixed value on which $\rho^{\ast}$ places positive weight.

The existence condition is membership in the estimand family, whose weights are probability distributions (\Cref{sec: anatomy estimand} of the main text). For each racial condition $z$, the mixed weights under the restriction to $\mathcal{T}^{\ast}$ sum to
\begin{align*}
\sum_{\bm{v}_{\mathrm{fix}} \in \mathcal{V}_{\mathrm{fix}}} \rho^{\ast}(\bm{v}_{\mathrm{fix}}) \sum_{\bm{v}_{\mathrm{free}} \in \mathcal{V}_{\mathrm{free}}} q^{\ast}_z(\bm{v}_{\mathrm{free}} \given \bm{v}_{\mathrm{fix}}).
\end{align*}
The inner sum equals one at each held-fixed value where the conditional distribution $q^{\ast}_z(\bm{v}_{\mathrm{free}} \given \bm{v}_{\mathrm{fix}})$ exists, so the mixed weights are well defined and sum to one exactly when the conditional distribution exists at every held-fixed value on which $\rho^{\ast}$ places positive weight.

\subsubsection{Calibration to a reference population}

The construction of the mixed contrast in \Cref{sec: mixed contrast formalism} partitions the nonracial features into a held-fixed set and a free set. Under the partition, calibrating the joint PMF $q(z, \bm{v}_{\mathrm{fix}}, \bm{v}_{\mathrm{free}})$ to match a reference population requires more information than the analogous calibration in \citet{delacuestaetal2022}. The reason lies in which pieces of the joint PMF the mixed effect uses. Written with the mixed weights of \eqref{eq: mixed weights}, the race effect under the mixed contrast for unit $i$ in \eqref{eq: contrasts part margin} is
\begin{align*}
\sum_{\bm{v} \in \mathcal{V}} \left[y_i(1, \bm{v})\, \rho(\bm{v}_{\mathrm{fix}})\, q_1(\bm{v}_{\mathrm{free}} \given \bm{v}_{\mathrm{fix}}) - y_i(0, \bm{v})\, \rho(\bm{v}_{\mathrm{fix}})\, q_0(\bm{v}_{\mathrm{free}} \given \bm{v}_{\mathrm{fix}})\right],
\end{align*}
so the mixed effect uses the joint PMF only through the marginal distribution $\rho(\bm{v}_{\mathrm{fix}})$ of the held-fixed features and the conditional distributions $q_z(\bm{v}_{\mathrm{free}} \given \bm{v}_{\mathrm{fix}})$ of the free features within each racial condition. Calibrating the mixed effect means supplying the marginal distribution and the conditional distributions from data on the reference population.

The lighter requirement in \citet[][p.~21]{delacuestaetal2022} --- each attribute's marginal distribution alone --- rests on the assumption that the attributes have no three-way or higher-order interactions, and supplying partial joint distributions relaxes the assumption \citep[][pp.~21, 31]{delacuestaetal2022}. The assumption restricts the outcomes, and under that restriction the estimand of \citet{delacuestaetal2022} needs no more from the profile distribution than the marginal distributions provide. No restriction on the outcomes does the same for the mixed effect, because the weights of the mixed effect are themselves built from the joint PMF. A researcher could instead restrict the reference population --- supposing, for example, that the attributes are independent within each racial condition --- and calibrate from marginal distributions alone. The supposition would then fix the weights, and with the weights the estimand, by assumption rather than by data.

A reference population is often summarized by \textit{single-attribute margins} --- for each attribute separately, the share of the population at each value of that attribute, with no record of how the attributes combine. Many joint distributions share the same single-attribute margins while disagreeing about the dependence across attributes. Both distributions that the calibration needs are properties of the dependence: The marginal distribution $\rho(\bm{v}_{\mathrm{fix}})$ assigns probability to whole bundles of held-fixed values, and the conditional distribution $q_z(\bm{v}_{\mathrm{free}} \given \bm{v}_{\mathrm{fix}})$ records how the free features vary with the racial condition and the held-fixed value.

\textit{Worked example.} The example exhibits two populations that share every single-attribute margin yet differ in the conditional distribution the calibration needs. Take the two nonracial features of the running example of the main text --- the location of the arrest and the extent of the prior record --- coded as binary: $v_1 = 1$ for an arrest in a disadvantaged neighborhood and $v_1 = 0$ otherwise, and $v_2 = 1$ for an extensive prior record and $v_2 = 0$ otherwise. Consider a stylized reference population of $100$ arrestees that splits evenly on each feature, $\Pr(v_1 = 1) = \Pr(v_2 = 1) = 1/2$, and two ways the features can combine within the stylized population (\Cref{tab: margins example}). In the left population of the table, the two features are independent, and each of the four combinations of $(v_1, v_2)$ holds $25$ arrestees. In the right population, the two features coincide, with the matching combinations $(1, 1)$ and $(0, 0)$ holding $50$ arrestees each and the combinations $(1, 0)$ and $(0, 1)$ holding none. The conditional distribution of prior record given neighborhood is even in the left population, $\Pr(v_2 = 1 \given v_1) = 1/2$, and degenerate in the right population, $\Pr(v_2 = v_1 \given v_1) = 1$. The shading in \Cref{tab: margins example} marks what the two populations share: The shaded row and column totals --- the single-attribute margins --- are identical, while the unshaded interior cells differ.

\begin{table}[H]
\vspace{10pt}
\centering
\begin{tabular}{lccc}
\multicolumn{4}{c}{\textbf{Independent features}} \\
\toprule
& \multicolumn{2}{c}{$v_2$} & \\
\cmidrule(lr){2-3}
$v_1$ & $1$ & $0$ & Total \\
\midrule
$1$ & 25 & 25 & \cellcolor{gray!15}50 \\
$0$ & 25 & 25 & \cellcolor{gray!15}50 \\
Total & \cellcolor{gray!15}50 & \cellcolor{gray!15}50 & \cellcolor{gray!15}100 \\
\bottomrule
\end{tabular}
\qquad
\begin{tabular}{lccc}
\multicolumn{4}{c}{\textbf{Coinciding features}} \\
\toprule
& \multicolumn{2}{c}{$v_2$} & \\
\cmidrule(lr){2-3}
$v_1$ & $1$ & $0$ & Total \\
\midrule
$1$ & 50 & 0 & \cellcolor{gray!15}50 \\
$0$ & 0 & 50 & \cellcolor{gray!15}50 \\
Total & \cellcolor{gray!15}50 & \cellcolor{gray!15}50 & \cellcolor{gray!15}100 \\
\bottomrule
\end{tabular}
\caption{Two stylized populations of $100$ arrestees, one per table. Cells count arrestees by neighborhood $v_1$ (rows) and prior record $v_2$ (columns); because each population holds $100$ arrestees, dividing any cell by $100$ gives the corresponding population proportion. The shaded cells --- the row and column totals, which are the single-attribute margins --- are identical across the two populations; the unshaded interior cells differ.}
\label{tab: margins example}
\end{table}

The conditional distributions $q_z(\bm{v}_{\mathrm{free}} \given \bm{v}_{\mathrm{fix}})$ are underdetermined in the same way as the conditional distribution of the worked example, and more severely. The conditioning includes the racial condition $z$, so the calibration also needs the dependence between the nonracial features and race. Identifying the marginal distribution $\rho(\bm{v}_{\mathrm{fix}})$ and the conditional distributions $q_z(\bm{v}_{\mathrm{free}} \given \bm{v}_{\mathrm{fix}})$ therefore requires joint information: cross-tabulations of the free features against the held-fixed features and race in the reference population, or individual-level data on the members of the reference population, from which any such cross-tabulation can be computed.

\subsubsection{Typicality} \label{app: typicality construction}

The construction sketched in \Cref{sec: weighting substantive} of the main text translates graded typicality into the conditional distribution $q_z(\bm{v}_{\mathrm{free}})$. Three elements of the construction come from the theory of conceptual spaces \citep{gardenfors2000,gardenfors2014}: the representation of objects as points in a space of \textit{quality dimensions} --- scales on which the objects can be compared \citep[][p.~2]{gardenfors2000} --- with similarity decreasing in distance \citep[][p.~44]{gardenfors2000}; the representation of a category by a prototype, the category's most typical point \citep[][pp.~87--91]{gardenfors2000}; and a distance that weights each dimension by the dimension's salience \citep[][p.~104]{gardenfors2000}. The exponential decay of similarity in distance follows \citet{nosofsky1986} and \citet{shepard1987}, whom \citet[][p.~20]{gardenfors2000} also follows. The assembly is this paper's: group-specific prototypes, salience weights, and concentration parameters, applied to the finite set $\mathcal{V}_{\mathrm{free}}$ and normalized into a weighting distribution for the estimand family, together with the conditional-prototype variant below.

The construction needs three ingredients.
\begin{itemize}
\item An embedding $\phi: \mathcal{V}_{\mathrm{free}} \to \R^D$ --- a function that represents the elements of one set as points of a space --- maps each free-feature profile to a point of a $D$-dimensional conceptual space, with each coordinate one quality dimension. The number of dimensions $D$ may be smaller than the number of free features, as when several features collapse into a single scale, or larger, as when one feature occupies several scales --- neighborhood disadvantage in the running example, for instance, may occupy one coordinate for income and one for intensity of policing.
\item A racial prototype $\bm{p}_z \in \R^D$ locates the most typical bundle of free features for group $z$ in the conceptual space.
\item Salience weights $\bm{a}_z \in [0,1]^D$ with $\sum_{k = 1}^D a_{z,k} = 1$ record how diagnostic each dimension of the conceptual space is for membership in group $z$. A difference along a highly salient dimension reduces typicality more than the same difference along a less salient dimension.
\end{itemize}
Typicality is similarity to the prototype, and similarity in a conceptual space decreases with distance, so grading typicality requires a distance from the prototype in which differences along more salient dimensions matter more. The embedding, the prototype, and the salience weights combine into such a distance, the weighted Euclidean distance from the prototype,
\begin{align} \label{eq: typicality distance}
d_z(\bm{v}_{\mathrm{free}}, \bm{p}_z) \coloneqq \sqrt{\left(\phi(\bm{v}_{\mathrm{free}}) - \bm{p}_z\right)^{\top} \diag{\bm{a}_z} \left(\phi(\bm{v}_{\mathrm{free}}) - \bm{p}_z\right)},
\end{align}
where $\diag{\bm{a}_z}$ denotes the $D \times D$ diagonal matrix with the salience weights $a_{z,1}, \ldots, a_{z,D}$ on its diagonal. The distance is zero exactly when $\phi(\bm{v}_{\mathrm{free}}) = \bm{p}_z$ and positive otherwise.

Exponentiate the negative squared distance, scaled by a concentration parameter $\lambda_z \geq 0$ that sets how fast a profile's weight declines with distance from the prototype, and normalize over $\mathcal{V}_{\mathrm{free}}$. The result is a conditional distribution that concentrates around the prototype \citep{nosofsky1986,shepard1987}:
\begin{align} \label{eq: typicality dist}
q_z(\bm{v}_{\mathrm{free}}; \lambda_z) \coloneqq \dfrac{\exp\left(-\lambda_z \, d_z(\bm{v}_{\mathrm{free}}, \bm{p}_z)^2\right)}{\sum_{\bm{v}_{\mathrm{free}}^{\prime} \in \mathcal{V}_{\mathrm{free}}} \exp\left(-\lambda_z \, d_z(\bm{v}_{\mathrm{free}}^{\prime}, \bm{p}_z)^2\right)}.
\end{align}
In words, a profile's weight under \eqref{eq: typicality dist} declines exponentially in the profile's weighted squared distance from the prototype, so profiles near the prototype carry most of the weight and profiles far from the prototype contribute little weight.

The two extremes of the concentration parameter show its role: At $\lambda_z = 0$, every profile in $\mathcal{V}_{\mathrm{free}}$ receives equal weight; as $\lambda_z \to \infty$, the distribution collapses to a point mass on the profile closest to the prototype $\bm{p}_z$. The two groups may differ in their concentration parameters $\lambda_0$ and $\lambda_1$ --- the formal counterpart of the example in \Cref{sec: weighting substantive} of the main text, where DeShawn may register as more atypically White than Greg does atypically Black because typicality declines at different rates for the two groups.

\textit{Worked example.} The application of \Cref{sec: app exposure} of the main text supplies a deliberately simple instance --- two profiles and one dimension --- in which the typicality construction reproduces the application's weights exactly and shows what each parameter contributes. Among the candidates who campaigned in Spanish in the content analysis of \citet{zarateetal2024}, the within-race accent distributions over $\mathcal{V}_{\mathrm{free}} = \{\textrm{Non-Native Spanish}, \textrm{Native Spanish}\}$ are $q_0 = (.96, .04)$ for Anglo and $q_1 = (.17, .83)$ for Hispanic candidates (\Cref{tab: lang dist} of the main text). Take a single quality dimension recording the nativeness of the accent, $D = 1$, with the embedding $\phi(\textrm{Non-Native Spanish}) = 0$ and $\phi(\textrm{Native Spanish}) = 1$; with one dimension, the salience weights are trivial, $a_{z,1} = 1$ for both groups, and salience plays no role until a second dimension enters. Set each group's prototype at the accent typical of the group's own campaigners, $\bm{p}_0 = 0$ (non-native) and $\bm{p}_1 = 1$ (native), so the distance \eqref{eq: typicality distance} between a profile and a prototype equals $0$ or $1$.

With the distances so arranged, the distribution \eqref{eq: typicality dist} places weight $1/(1 + e^{-\lambda_z})$ on the accent at group $z$'s prototype and weight $e^{-\lambda_z}/(1 + e^{-\lambda_z})$ on the other accent, so the concentration parameter equals the logarithm of the group's odds of its typical accent. The exact odds come from the counts of the content analysis: $47$ non-native and $2$ native Spanish speakers among the Anglo campaigners, and $13$ non-native and $64$ native among the Hispanic campaigners. Setting $\lambda_1 = \ln(64/13) \approx 1.59$ recovers $q_1$ exactly, and setting $\lambda_0 = \ln(47/2) \approx 3.16$ recovers $q_0$ exactly. The larger $\lambda_0$ says that Anglo campaigners concentrate more tightly on their typical accent than Hispanic campaigners do on theirs. The equality between the typicality weighting and the campaign distributions is the case that \Cref{sec: app exposure} of the main text describes, in which typicality is defined by the target population itself.

\textit{Conditional-prototype variant.} The most typical bundle of free features for a group may differ across values of the features that the mixed contrast holds fixed. For each pair of a group $z$ and a value $\bm{v}_{\mathrm{fix}}$ of the held-fixed features, a conditional prototype $\bm{p}_{z,\bm{v}_{\mathrm{fix}}} \in \R^D$ locates the most typical bundle of free features for group $z$ among individuals whose held-fixed features equal $\bm{v}_{\mathrm{fix}}$. Replacing the prototype $\bm{p}_z$ with the conditional prototype $\bm{p}_{z,\bm{v}_{\mathrm{fix}}}$ in the distance \eqref{eq: typicality distance} and the distribution \eqref{eq: typicality dist} yields
\begin{align} \label{eq: typicality dist conditional}
q_z(\bm{v}_{\mathrm{free}} \given \bm{v}_{\mathrm{fix}}; \lambda_z) \coloneqq
\dfrac{\exp\left(-\lambda_z \, d_z(\bm{v}_{\mathrm{free}},\bm{p}_{z,\bm{v}_{\mathrm{fix}}})^2\right)}
{\sum_{\bm{v}_{\mathrm{free}}^{\prime} \in \mathcal{V}_{\mathrm{free}}}
\exp\left(-\lambda_z \, d_z(\bm{v}_{\mathrm{free}}^{\prime},\bm{p}_{z,\bm{v}_{\mathrm{fix}}})^2\right)}.
\end{align}
For example, if $\bm{v}_{\mathrm{fix}}$ contains the neighborhood of the running example and $\bm{v}_{\mathrm{free}}$ contains the prior record, the most typical prior record for group $z$ among arrestees from a disadvantaged neighborhood may differ from the most typical prior record for group $z$ overall, so the conditional prototype shifts with the neighborhood.

\subsection{Randomization-based properties of the plug-in estimator} \label{app: randomization properties}

This subsection establishes the randomization-based properties of the plug-in estimator of the estimand family, $\widehat{\theta}$ of \eqref{eq: family estimator main} of the main text, restated with its ingredients in \eqref{eq: theta estimator} below. The estimator targets the finite-population estimand
\begin{align} \label{eq: theta estimand app}
\theta \coloneqq \sum_{\bm{v} \in \mathcal{V}} g_1(\bm{v})\, \mu(1, \bm{v}) - \sum_{\bm{v} \in \mathcal{V}} g_0(\bm{v})\, \mu(0, \bm{v}),
\end{align}
where $\mu(z, \bm{v}) \coloneqq \dfrac{1}{N} \sum_{i = 1}^N y_i(z, \bm{v})$ denotes the mean of the potential outcomes at profile $(z, \bm{v})$ over all $N$ units.

The mean $\mu(z, \bm{v})$ equals the average potential outcome $\bar{y}(z, \bm{v})$ of \Cref{prop: gap}; this subsection writes $\mu(z, \bm{v})$ so that each population mean pairs with its estimator, the cell mean $\hat{\mu}(z, \bm{v})$ of \eqref{eq: cell mean estimator} of the main text. Exchanging the finite sums over $i$ and $\bm{v}$ shows that $\theta$ coincides with $\theta(g_0, g_1)$ of \Cref{prop: gap}.

The AEE ($g_0 = g_1 = \rho$), the AEWR ($g_z = q_z$), and the mixed contrasts ($g_z$ as in \eqref{eq: mixed weights}) are the special cases obtained by fixing the weights $g_0$ and $g_1$, so the results below cover the whole estimand family.

Let $n_{z, \bm{v}}$ be fixed positive integers for all $(z, \bm{v}) \in \mathcal{T}$ satisfying $\sum_{\bm{v} \in \mathcal{V}} n_{0, \bm{v}} + \sum_{\bm{v} \in \mathcal{V}} n_{1, \bm{v}} = N$. For each $(z, \bm{v}) \in \mathcal{T}$, refer to the set of units assigned to profile $(z, \bm{v})$ as the \textit{cell} corresponding to profile $(z, \bm{v})$; the integer $n_{z, \bm{v}}$ fixes the size of that cell.

Define the set of configuration assignments consistent with these fixed cell sizes as
\begin{align}
\Omega \coloneqq \bigcap_{(z, \bm{v}) \in \mathcal{T}} \left\{\bm{t} \in \mathcal{T}^N: \sum \limits_{i = 1}^N \mathbbm{1}\{\bm{t}_i = (z, \bm{v})\} = n_{z, \bm{v}}\right\}.
\end{align}
Thus, $\Omega$ consists of all configuration assignments $\bm{t} \in \mathcal{T}^N$ such that, for every profile $(z, \bm{v}) \in \mathcal{T}$, exactly $n_{z, \bm{v}}$ components of $\bm{t}$ equal $(z, \bm{v})$.

Complete random assignment (\Cref{assm: CRA}) makes the random configuration assignment $\bm{T}$ uniformly distributed on $\Omega$, so $\Pr(\bm{T} = \bm{t}) = 1/\abs{\Omega}$ for all $\bm{t} \in \Omega$, where $\abs{\Omega}$, the number of assignments in $\Omega$, equals $N! / \prod_{(z, \bm{v}) \in \mathcal{T}} n_{z, \bm{v}}!$.

To construct the plug-in estimator $\widehat{\theta}$, define the observed outcome for unit $i$ as $Y_i \coloneqq y_i(\bm{T})$, a random quantity through $\bm{T}$. For each $(z, \bm{v}) \in \mathcal{T}$, define the mean of the observed outcomes within the cell corresponding to profile $(z, \bm{v})$ as
\begin{align} \label{eq: profile-cell estimator I}
\hat{\mu}(z, \bm{v}) \coloneqq \dfrac{1}{n_{z, \bm{v}}} \sum \limits_{i = 1}^N Y_i \,\mathbbm{1}\left\{\bm{T}_i = (z, \bm{v})\right\},
\end{align}
which is a function of the random configuration assignment $\bm{T}$. The plug-in estimator $\widehat{\theta}$ is then
\begin{align} \label{eq: theta estimator}
\widehat{\theta} \coloneqq \sum \limits_{\bm{v} \in \mathcal{V}} g_1(\bm{v}) \hat{\mu}(1, \bm{v}) - \sum \limits_{\bm{v} \in \mathcal{V}} g_0(\bm{v}) \hat{\mu}(0, \bm{v}).
\end{align}
All randomness in $\widehat{\theta}$ arises from the random assignment $\bm{T}$, which enters through the cell means $\hat{\mu}(z, \bm{v})$.

\begin{lem}[Randomization moments] \label{lem: randomization moms}
Under complete random assignment (\Cref{assm: CRA}), for every unit $i \in \{1, \ldots, N\}$ and every profile $(z, \bm{v}) \in \mathcal{T}$,
\begin{align}
\E\left[\mathbbm{1}\{\bm{T}_i = (z, \bm{v})\}\right] & = \dfrac{n_{z, \bm{v}}}{N} \label{eq: assign EV} \\
\Var\left[\mathbbm{1}\{\bm{T}_i = (z, \bm{v})\}\right] & = \dfrac{n_{z, \bm{v}}(N - n_{z, \bm{v}})}{N^2} \label{eq: assign Var}
\end{align}
and, for any units $i$ and $j$ and any profiles $(z, \bm{v}), (z^{\prime},\bm{v}^{\prime}) \in \mathcal{T}$,
\begin{align}
\Cov\left[\mathbbm{1}\left\{\bm{T}_i = (z, \bm{v})\right\}, \mathbbm{1}\left\{\bm{T}_j = (z, \bm{v})\right\} \right] & = - \dfrac{n_{z, \bm{v}}(N - n_{z, \bm{v}})}{N^2 (N-1)} \text{ for } i \neq j \label{eq: assign Cov I} \\ 
\Cov\left[\mathbbm{1}\left\{\bm{T}_i = (z, \bm{v})\right\}, \mathbbm{1}\left\{\bm{T}_j = (z^{\prime},\bm{v}^{\prime})\right\}\right] & = \dfrac{n_{z, \bm{v}} n_{z^{\prime},\bm{v}^{\prime}}}{N^2 (N - 1)} \text{ for } i \neq j \text{ and } (z, \bm{v}) \neq (z^{\prime},\bm{v}^{\prime}) \label{eq: assign Cov II} \\
\Cov\left[\mathbbm{1}\{\bm{T}_i = (z, \bm{v})\}, \mathbbm{1}\{\bm{T}_i = (z^{\prime},\bm{v}^{\prime})\}\right] & = - \dfrac{n_{z, \bm{v}} n_{z^{\prime},\bm{v}^{\prime}}}{N^2} \text{ for } (z, \bm{v}) \neq (z^{\prime},\bm{v}^{\prime}). \label{eq: assign Cov III}
\end{align}
\end{lem}
\begin{proof}
The proof computes each moment by counting assignments. The counting is justified by the design: Complete random assignment places equal probability $1/\abs{\Omega}$ on every assignment in $\Omega$, so the probability of any event equals the number of assignments in the event divided by $\abs{\Omega}$. Every moment in the lemma reduces to such a probability because the expectation of a product of indicators is the probability that all of the indicated events occur. Each count then follows from fixing the coordinates of the units that the event names and distributing the remaining units over the remaining places in each cell.

Fix a unit $i \in \{1,\ldots,N\}$ and a profile $(z, \bm{v}) \in \mathcal{T}$. Then
\begin{align*}
\E\left[\mathbbm{1}\{\bm{T}_i = (z, \bm{v})\}\right] & = \Pr\left(\bm{T}_i = (z, \bm{v})\right) \\ 
& = \sum \limits_{\bm{t} \in \Omega} \mathbbm{1}\{\bm{t}_i = (z, \bm{v})\} \Pr(\bm{T} = \bm{t}).
\end{align*}
Under complete random assignment (\Cref{assm: CRA}), $\Pr(\bm{T} = \bm{t}) = 1/\abs{\Omega}$ for all $\bm{t} \in \Omega$. Therefore,
\begin{align*}
\Pr\left(\bm{T}_i = (z, \bm{v})\right) = \dfrac{1}{\abs{\Omega}} \sum \limits_{\bm{t} \in \Omega} \mathbbm{1}\{\bm{t}_i = (z, \bm{v})\}.
\end{align*}
To evaluate the number of assignments in $\Omega$ for which $\bm{t}_i = (z, \bm{v})$, note that fixing $\bm{t}_i = (z, \bm{v})$ reduces the remaining counts to $n_{z, \bm{v}} - 1$ in cell $(z, \bm{v})$ and $n_{z^{\prime},\bm{v}^{\prime}}$ in every other cell $(z^{\prime},\bm{v}^{\prime}) \neq (z, \bm{v})$. Thus,
\begin{align*}
\sum_{\bm{t} \in \Omega} \mathbbm{1}\{\bm{t}_i = (z, \bm{v})\} & = \dfrac{(N - 1)!}{(n_{z, \bm{v}} - 1)! \prod_{(z^{\prime},\bm{v}^{\prime}) \neq (z, \bm{v})} n_{z^{\prime},\bm{v}^{\prime}}!}.
\end{align*}
Writing the size of $\Omega$, stated in the introduction to this subsection, with the factor for cell $(z, \bm{v})$ separated from the product of factorials,
\begin{align*}
\abs{\Omega} & = \dfrac{N!}{n_{z, \bm{v}}! \prod_{(z^{\prime},\bm{v}^{\prime}) \neq (z, \bm{v})} n_{z^{\prime},\bm{v}^{\prime}}!},
\end{align*}
which implies that
\begin{align*}
\Pr\left(\bm{T}_i = (z, \bm{v})\right) & = \dfrac{1}{\abs{\Omega}} \sum \limits_{\bm{t} \in \Omega} \mathbbm{1}\{\bm{t}_i = (z, \bm{v})\} \\
& = \dfrac{(N-1)!}{(n_{z, \bm{v}}-1)! \prod_{(z^{\prime},\bm{v}^{\prime}) \neq (z, \bm{v})} n_{z^{\prime},\bm{v}^{\prime}}!} \dfrac{n_{z, \bm{v}}! \prod_{(z^{\prime},\bm{v}^{\prime}) \neq (z, \bm{v})} n_{z^{\prime}, \bm{v}^{\prime}}!}{N!} \\
& = \dfrac{n_{z, \bm{v}}}{N},
\end{align*}
which establishes \eqref{eq: assign EV}.

Because an indicator equals its own square, the definition of variance then gives
\begin{align*}
\Var\left[\mathbbm{1}\{\bm{T}_i = (z, \bm{v})\}\right]
& = \E\left[\mathbbm{1}\{\bm{T}_i = (z, \bm{v})\}^2\right]
   - \E\left[\mathbbm{1}\{\bm{T}_i = (z, \bm{v})\}\right]^2 \\
& = \E\left[\mathbbm{1}\{\bm{T}_i = (z, \bm{v})\}\right] 
   - \E\left[\mathbbm{1}\{\bm{T}_i = (z, \bm{v})\}\right]^2 \\ 
& = \dfrac{n_{z, \bm{v}}}{N}
   - \left(\dfrac{n_{z, \bm{v}}}{N}\right)^2 \\
& = \dfrac{n_{z, \bm{v}}(N - n_{z, \bm{v}})}{N^2},
\end{align*}
which establishes \eqref{eq: assign Var}.

Next, fix distinct units $i \neq j$. For the same profile $(z, \bm{v})$, the product $\mathbbm{1}\{\bm{T}_i = (z, \bm{v})\}\, \mathbbm{1}\{\bm{T}_j = (z, \bm{v})\}$ equals $1$ exactly when both units fall in cell $(z, \bm{v})$, so its expectation is the joint probability $\Pr\left(\bm{T}_i = (z, \bm{v}), \bm{T}_j = (z, \bm{v})\right)$. The count parallels the derivation of \eqref{eq: assign EV}: Fixing $\bm{t}_i = \bm{t}_j = (z, \bm{v})$ reduces the remaining counts to $n_{z, \bm{v}} - 2$ in cell $(z, \bm{v})$ and $n_{z^{\prime},\bm{v}^{\prime}}$ in every other cell $(z^{\prime},\bm{v}^{\prime}) \neq (z, \bm{v})$, so
\small
\begin{align*}
\Pr\left(\bm{T}_i = (z, \bm{v}), \bm{T}_j = (z, \bm{v})\right) & = \dfrac{(N - 2)!}{(n_{z, \bm{v}} - 2)! \prod_{(z^{\prime},\bm{v}^{\prime}) \neq (z, \bm{v})} n_{z^{\prime},\bm{v}^{\prime}}!} \dfrac{n_{z, \bm{v}}! \prod_{(z^{\prime},\bm{v}^{\prime}) \neq (z, \bm{v})} n_{z^{\prime},\bm{v}^{\prime}}!}{N!} \\
& = \dfrac{(N - 2)!}{N!} \dfrac{n_{z, \bm{v}}!}{(n_{z, \bm{v}} - 2)!} \\
& = \dfrac{n_{z, \bm{v}}(n_{z, \bm{v}} - 1)}{N(N - 1)}.
\end{align*}
\normalsize
The count assumes $n_{z, \bm{v}} \geq 2$; when $n_{z, \bm{v}} = 1$, no assignment places two distinct units in the cell, so the joint probability equals $0$, which the final expression also gives. Therefore,
\small
\begin{align*}
\Cov\left[\mathbbm{1}\{\bm{T}_i = (z, \bm{v})\},
          \mathbbm{1}\{\bm{T}_j = (z, \bm{v})\}\right]
&= \E\left[\mathbbm{1}\{\bm{T}_i = (z, \bm{v})\}
           \mathbbm{1}\{\bm{T}_j = (z, \bm{v})\}\right] \\
& - \E\left[\mathbbm{1}\{\bm{T}_i = (z, \bm{v})\}\right]
        \E\left[\mathbbm{1}\{\bm{T}_j = (z, \bm{v})\}\right] \\
& = \dfrac{n_{z, \bm{v}}(n_{z, \bm{v}} - 1)}{N(N - 1)}
   - \dfrac{n_{z, \bm{v}}^2}{N^2} \\
& = - \dfrac{n_{z, \bm{v}}(N - n_{z, \bm{v}})}{N^2(N - 1)},
\end{align*}
\normalsize
which establishes \eqref{eq: assign Cov I}.

For $(z, \bm{v}) \neq (z^{\prime}, \bm{v}^{\prime})$, still with $i \neq j$, the count is analogous: Fixing $\bm{t}_i = (z, \bm{v})$ and $\bm{t}_j = (z^{\prime}, \bm{v}^{\prime})$ reduces the remaining counts to $n_{z, \bm{v}} - 1$ and $n_{z^{\prime},\bm{v}^{\prime}} - 1$ in those two cells, so
\begin{align*}
\E\left[\mathbbm{1}\{\bm{T}_i = (z, \bm{v})\}\, \mathbbm{1}\{\bm{T}_j = (z^{\prime}, \bm{v}^{\prime})\}\right] = \Pr\left(\bm{T}_i = (z, \bm{v}), \bm{T}_j = (z^{\prime}, \bm{v}^{\prime})\right) = \dfrac{n_{z, \bm{v}}}{N} \dfrac{n_{z^{\prime}, \bm{v}^{\prime}}}{N - 1}.
\end{align*}
Therefore,
\small
\begin{align*}
\Cov\left[\mathbbm{1}\{\bm{T}_i = (z, \bm{v})\},
          \mathbbm{1}\{\bm{T}_j = (z^{\prime}, \bm{v}^{\prime})\}\right]
& = \dfrac{n_{z, \bm{v}}}{N}
   \dfrac{n_{z^{\prime}, \bm{v}^{\prime}}}{N - 1}
   - \dfrac{n_{z, \bm{v}}}{N}
     \dfrac{n_{z^{\prime},\bm{v}^{\prime}}}{N} \\
& = \dfrac{n_{z, \bm{v}} n_{z^{\prime},\bm{v}^{\prime}}}{N^2 (N - 1)},
\end{align*}
\normalsize
which establishes \eqref{eq: assign Cov II}.

Finally, fix a single unit $i$ and two distinct profiles $(z, \bm{v}) \neq (z^{\prime},\bm{v}^{\prime})$. A unit receives exactly one profile, so the two indicator events are disjoint and the product of the two indicators equals zero at every assignment in $\Omega$:
\begin{align*}
\mathbbm{1}\{\bm{T}_i = (z, \bm{v})\} \mathbbm{1}\{\bm{T}_i = (z^{\prime},\bm{v}^{\prime})\} = 0.
\end{align*}
Therefore,
\begin{align*}
\Cov\left[\mathbbm{1}\{\bm{T}_i = (z, \bm{v})\},
          \mathbbm{1}\{\bm{T}_i = (z^{\prime},\bm{v}^{\prime})\}\right]
& = 0 
   - \dfrac{n_{z, \bm{v}}}{N}
     \dfrac{n_{z^{\prime},\bm{v}^{\prime}}}{N} \\
& = - \dfrac{n_{z, \bm{v}} n_{z^{\prime},\bm{v}^{\prime}}}{N^2},
\end{align*}
which establishes \eqref{eq: assign Cov III}.
\end{proof}

The first proposition of this subsection establishes part \textit{(i)} of \Cref{prop: two blindnesses} of the main text --- the plug-in estimator $\widehat{\theta}$ is unbiased for the estimand $\theta$ --- the part whose proof \Cref{app: double duty proof} defers to this subsection.

\begin{prop}[Unbiasedness of the plug-in estimator] \label{prop: theta EV}
Under \Cref{assm: SUTVA}, together with \Cref{assm: CRA} of the main text, $\widehat{\theta}$ in \eqref{eq: theta estimator} is unbiased for $\theta$ in \eqref{eq: theta estimand app}:
\begin{align}
\E\left[\widehat{\theta}\right] = \theta,
\end{align}
where the expectation $\E$ is taken with respect to the randomization distribution of the configuration assignment $\bm{T}$, which \Cref{assm: CRA} fixes as the uniform distribution on $\Omega$.
\end{prop}
\begin{proof}
The proof writes each cell mean as a sum of fixed potential outcomes multiplied by random assignment indicators and applies the expectation of each indicator from \Cref{lem: randomization moms}.

Under SUTVA (\Cref{assm: SUTVA}), $Y_i = y_i(\bm{T}) = y_i(\bm{T}_i)$, so
\begin{align*}
\hat{\mu}(z, \bm{v}) = \dfrac{1}{n_{z, \bm{v}}} \sum \limits_{i = 1}^N y_i(\bm{T}_i) \mathbbm{1}\{\bm{T}_i = (z, \bm{v})\}.
\end{align*}
Each summand is nonzero only when $\bm{T}_i = (z, \bm{v})$, so replacing the random argument $\bm{T}_i$ of $y_i$ with the fixed profile $(z, \bm{v})$ leaves every summand unchanged:
\begin{align} \label{eq: profile-cell estimator II}
\hat{\mu}(z, \bm{v}) & = \dfrac{1}{n_{z, \bm{v}}} \sum \limits_{i = 1}^N y_i(z, \bm{v}) \mathbbm{1}\{\bm{T}_i = (z, \bm{v})\}.
\end{align}
Since $g_1(\bm{v})$ and $g_0(\bm{v})$ are fixed constants for each attribute profile $\bm{v} \in \mathcal{V}$, the linearity of expectation gives
\begin{align} \label{eq: profile-cell estimator III}
\E\left[\widehat{\theta}\right] & = \sum \limits_{\bm{v} \in \mathcal{V}} g_1(\bm{v}) \E\left[\hat{\mu}(1, \bm{v})\right] - \sum \limits_{\bm{v} \in \mathcal{V}} g_0(\bm{v}) \E\left[\hat{\mu}(0, \bm{v})\right].
\end{align}
Under complete random assignment (\Cref{assm: CRA}), the cell sizes $n_{z, \bm{v}}$ are fixed across all $\bm{t} \in \Omega$, and the potential outcomes $y_i(z, \bm{v})$ are fixed features of the finite population. Substituting \eqref{eq: profile-cell estimator II} into \eqref{eq: profile-cell estimator III} and again applying the linearity of expectation therefore yields
\begin{equation*}
\begin{split}
\E\left[\widehat{\theta}\right] & = \sum_{\bm{v} \in \mathcal{V}} g_1(\bm{v}) \dfrac{1}{n_{1, \bm{v}}} \sum \limits_{i = 1}^N y_i(1, \bm{v}) \E\left[\mathbbm{1}\{\bm{T}_i = (1, \bm{v})\}\right] \\
& - \sum \limits_{\bm{v} \in \mathcal{V}} g_0(\bm{v}) \dfrac{1}{n_{0, \bm{v}}} \sum \limits_{i = 1}^N y_i(0, \bm{v}) \E\left[\mathbbm{1}\{\bm{T}_i = (0, \bm{v})\}\right].
\end{split}
\end{equation*}
Then, under complete random assignment (\Cref{assm: CRA}), \eqref{eq: assign EV} of \Cref{lem: randomization moms} implies that
\begin{align*}
\E\left[\widehat{\theta}\right] = \sum_{\bm{v} \in \mathcal{V}} g_1(\bm{v}) \dfrac{1}{n_{1, \bm{v}}} \sum \limits_{i = 1}^N y_i(1, \bm{v}) \left(n_{1, \bm{v}} / N\right) - \sum \limits_{\bm{v} \in \mathcal{V}} g_0(\bm{v}) \dfrac{1}{n_{0, \bm{v}}} \sum \limits_{i = 1}^N y_i(0, \bm{v}) \left(n_{0, \bm{v}} / N\right).
\end{align*}
Since $\dfrac{1}{n_{z, \bm{v}}} \left(n_{z, \bm{v}} / N\right) = \dfrac{1}{N}$, each inner sum reduces to $\dfrac{1}{N} \sum_{i = 1}^N y_i(z, \bm{v}) = \mu(z, \bm{v})$, so
\begin{align*}
\E\left[\widehat{\theta}\right] & = \sum_{\bm{v} \in \mathcal{V}} g_1(\bm{v}) \mu(1, \bm{v}) - \sum \limits_{\bm{v} \in \mathcal{V}} g_0(\bm{v}) \mu(0, \bm{v}),
\end{align*}
which equals $\theta$ in \eqref{eq: theta estimand app}, thereby completing the proof.
\end{proof}

The second proposition of this subsection derives the exact variance of the plug-in estimator $\widehat{\theta}$. The variance involves two population quantities. For each profile $(z, \bm{v}) \in \mathcal{T}$, define the finite-population variance of the potential outcomes at $(z, \bm{v})$ as
\begin{align} \label{eq: S2 def}
S^2\left(z, \bm{v}\right) & \coloneqq \dfrac{1}{N - 1} \sum \limits_{i = 1}^N \left[y_i(z, \bm{v}) - \mu(z, \bm{v})\right]^2,
\end{align}
and, for any two profiles $(z, \bm{v}), (z^{\prime}, \bm{v}^{\prime}) \in \mathcal{T}$, define the finite-population covariance of the potential outcomes at $(z, \bm{v})$ and $(z^{\prime}, \bm{v}^{\prime})$ as
\begin{align} \label{eq: S cov def}
S\left(z, \bm{v}; z^{\prime}, \bm{v}^{\prime}\right) & \coloneqq \dfrac{1}{N - 1} \sum \limits_{i = 1}^N \left[y_i(z, \bm{v}) - \mu(z, \bm{v})\right] \left[y_i(z^{\prime}, \bm{v}^{\prime}) - \mu(z^{\prime}, \bm{v}^{\prime})\right].
\end{align}

\begin{prop}[Exact variance of the plug-in estimator] \label{prop: theta Var}
Under \Cref{assm: SUTVA}, together with \Cref{assm: CRA} of the main text,
\begin{align*}
\Var\left[\widehat{\theta}\right] & = \sum \limits_{\bm{v} \in \mathcal{V}} g_1(\bm{v})^2 \dfrac{N - n_{1, \bm{v}}}{N} \dfrac{S^2(1, \bm{v})}{n_{1, \bm{v}}} + \sum \limits_{\bm{v} \in \mathcal{V}} g_0(\bm{v})^2 \dfrac{N - n_{0, \bm{v}}}{N} \dfrac{S^2(0, \bm{v})}{n_{0, \bm{v}}} \\ 
& - \dfrac{1}{N} \sum \limits_{\bm{v} \neq \bm{v}^{\prime}} g_1(\bm{v}) g_1(\bm{v}^{\prime}) S(1, \bm{v}; 1, \bm{v}^{\prime}) - \dfrac{1}{N} \sum \limits_{\bm{v} \neq \bm{v}^{\prime}} g_0(\bm{v}) g_0(\bm{v}^{\prime}) S(0, \bm{v}; 0, \bm{v}^{\prime}) \\ 
& + \dfrac{2}{N} \sum \limits_{\bm{v} \in \mathcal{V}} \sum \limits_{\bm{v}^{\prime} \in \mathcal{V}} g_1(\bm{v}) g_0(\bm{v}^{\prime}) S(1, \bm{v} ; 0, \bm{v}^{\prime}),
\end{align*}
where the variance $\Var$ is taken with respect to the same randomization distribution of $\bm{T}$ as in \Cref{prop: theta EV}.
\end{prop}
\begin{proof}
The proof decomposes the variance of $\widehat{\theta}$ into three components --- the variance of the weighted sum under each racial condition and the covariance between the two weighted sums --- and evaluates each component with the moments of the indicators from \Cref{lem: randomization moms}.

The estimator $\widehat{\theta}$ in \eqref{eq: theta estimator} is a linear combination of the cell means $\hat{\mu}(z, \bm{v})$ with fixed weights $g_1(\bm{v})$ and $-g_0(\bm{v})$. For fixed constants $\alpha_k$ and $\beta_l$ and random variables $X_k$ and $Y_l$, covariance distributes over weighted sums: $\Cov\left[\sum_k \alpha_k X_k, \, \sum_l \beta_l Y_l\right] = \sum_k \sum_l \alpha_k \beta_l\, \Cov\left[X_k, Y_l\right]$ (the bilinearity of covariance). The variance of a random variable is the covariance of the random variable with itself. Writing $\Var[\widehat{\theta}]$ as the covariance of $\widehat{\theta}$ with itself and applying the bilinearity of covariance to the difference of the two weighted sums in \eqref{eq: theta estimator},
\begin{equation} \label{eq: theta var decomp}
\begin{split}
\Var\left[\widehat{\theta}\right] & = \Var\left[\sum \limits_{\bm{v} \in \mathcal{V}} g_1(\bm{v}) \hat{\mu}(1, \bm{v})\right] + \Var\left[\sum \limits_{\bm{v} \in \mathcal{V}} g_0(\bm{v}) \hat{\mu}(0, \bm{v})\right] \\ 
& - 2 \Cov\left[ \sum \limits_{\bm{v} \in \mathcal{V}} g_1(\bm{v}) \hat{\mu}(1, \bm{v}), \, \sum \limits_{\bm{v} \in \mathcal{V}} g_0(\bm{v}) \hat{\mu}(0, \bm{v})\right]. 
\end{split}
\end{equation}
Take the three components of \eqref{eq: theta var decomp} one at a time. For the first two, again by the bilinearity of covariance,
\small
\begin{align}
\Var\left[\sum \limits_{\bm{v} \in \mathcal{V}} g_1(\bm{v}) \hat{\mu}(1, \bm{v})\right] & = \sum_{\bm{v} \in \mathcal{V}} g_1(\bm{v})^2 \Var\left[\hat{\mu}(1, \bm{v})\right] + \sum_{\bm{v} \neq \bm{v}^{\prime}} g_1(\bm{v}) g_1(\bm{v}^{\prime}) \Cov\left[\hat{\mu}(1, \bm{v}), \hat{\mu}(1, \bm{v}^{\prime})\right] \label{eq: theta var decomp A} \\
\Var\left[\sum \limits_{\bm{v} \in \mathcal{V}} g_0(\bm{v}) \hat{\mu}(0, \bm{v})\right] & = \sum_{\bm{v} \in \mathcal{V}} g_0(\bm{v})^2 \Var\left[\hat{\mu}(0, \bm{v})\right] + \sum_{\bm{v} \neq \bm{v}^{\prime}}  g_0(\bm{v}) g_0(\bm{v}^{\prime})\Cov\left[\hat{\mu}(0, \bm{v}), \hat{\mu}(0, \bm{v}^{\prime})\right]. \label{eq: theta var decomp B}
\end{align}
\normalsize
Under SUTVA (\Cref{assm: SUTVA}), the cell mean $\hat{\mu}(z, \bm{v})$ takes the form \eqref{eq: profile-cell estimator II}, so
\begin{align*}
\Var\left[\hat{\mu}(z, \bm{v})\right] & = \Var\left[\dfrac{1}{n_{z, \bm{v}}} \sum \limits_{i = 1}^N y_i(z, \bm{v}) \mathbbm{1}\{\bm{T}_i = (z, \bm{v})\}\right],
\end{align*}
which, by the bilinearity of covariance, is
\small
\begin{align*}
\Var\left[\hat{\mu}(z, \bm{v})\right] & = \dfrac{1}{n_{z, \bm{v}}^2} \sum \limits_{i = 1}^N y_i(z, \bm{v})^2 \Var\left[\mathbbm{1}\{\bm{T}_i = (z, \bm{v})\}\right] \\ 
& \qquad + \dfrac{1}{n_{z, \bm{v}}^2} \sum \limits_{i \neq j} y_i(z, \bm{v}) y_j(z, \bm{v}) \Cov\left[\mathbbm{1}\{\bm{T}_i = (z, \bm{v})\}, \mathbbm{1}\{\bm{T}_j = (z, \bm{v})\}\right].
\end{align*}
Then, under complete random assignment (\Cref{assm: CRA}), \eqref{eq: assign Var} and \eqref{eq: assign Cov I} of \Cref{lem: randomization moms} imply that
\begin{align*}
\Var\left[\hat{\mu}(z, \bm{v})\right] & = \dfrac{1}{n_{z, \bm{v}}^2} \sum \limits_{i = 1}^N y_i(z, \bm{v})^2 \left(\dfrac{n_{z, \bm{v}}(N - n_{z, \bm{v}})}{N^2}\right) \\
& \qquad + \dfrac{1}{n_{z, \bm{v}}^2} \sum \limits_{i \neq j} y_i(z, \bm{v}) y_j(z, \bm{v}) \left(- \dfrac{n_{z, \bm{v}}(N - n_{z, \bm{v}})}{N^2(N - 1)}\right).
\end{align*}
\normalsize
Factoring $\dfrac{n_{z, \bm{v}}(N - n_{z, \bm{v}})}{n_{z, \bm{v}}^2 N^2} = \dfrac{N - n_{z, \bm{v}}}{n_{z, \bm{v}} N^2}$ out of both terms of the display above, so that the two sums collect into a single factor in brackets,
\begin{align} \label{eq: mu var bracket}
\Var\left[\hat{\mu}(z, \bm{v})\right] & = \dfrac{N - n_{z, \bm{v}}}{n_{z, \bm{v}} N^2} \left[\sum \limits_{i = 1}^N y_i(z, \bm{v})^2 - \dfrac{1}{N - 1} \sum \limits_{i \neq j} y_i(z, \bm{v}) y_j(z, \bm{v})\right].
\end{align}
The remaining steps rewrite the factor in brackets in \eqref{eq: mu var bracket} as a multiple of a finite-population variance. Splitting the square of the sum $\left(\sum_{i = 1}^N y_i(z, \bm{v})\right)^2 = \sum_{i = 1}^N y_i(z, \bm{v})^2 + \sum_{i \neq j} y_i(z, \bm{v}) y_j(z, \bm{v})$ into its $i = j$ and $i \neq j$ terms and using $\sum_{i = 1}^N y_i(z, \bm{v}) = N \mu(z, \bm{v})$,
\begin{align*}
\sum \limits_{i \neq j} y_i(z, \bm{v}) y_j(z, \bm{v}) & = N^2 \mu(z, \bm{v})^2 - \sum \limits_{i = 1}^N y_i(z, \bm{v})^2.
\end{align*}
Substituting this expression for $\sum_{i \neq j} y_i(z, \bm{v}) y_j(z, \bm{v})$ into the factor in brackets in \eqref{eq: mu var bracket} and combining the two terms of that factor over the common denominator $N - 1$,
\small
\begin{align*}
\sum \limits_{i = 1}^N y_i(z, \bm{v})^2 - \dfrac{N^2 \mu(z, \bm{v})^2 - \sum_{i = 1}^N y_i(z, \bm{v})^2}{N - 1}
& = \dfrac{N \left[\sum_{i = 1}^N y_i(z, \bm{v})^2 - N \mu(z, \bm{v})^2\right]}{N - 1}.
\end{align*}
\normalsize
Expanding the square in $\sum_{i = 1}^N \left[y_i(z, \bm{v}) - \mu(z, \bm{v})\right]^2$ and again using $\sum_{i = 1}^N y_i(z, \bm{v}) = N \mu(z, \bm{v})$ shows that $\sum_{i = 1}^N \left[y_i(z, \bm{v}) - \mu(z, \bm{v})\right]^2 = \sum_{i = 1}^N y_i(z, \bm{v})^2 - N \mu(z, \bm{v})^2$.
The factor in brackets in \eqref{eq: mu var bracket} therefore equals $\dfrac{N}{N - 1} \sum_{i = 1}^N \left[y_i(z, \bm{v}) - \mu(z, \bm{v})\right]^2$, which is $N S^2(z, \bm{v})$ by the definition of the finite-population variance in \eqref{eq: S2 def}, so that
\begin{align}
\Var\left[\hat{\mu}(z, \bm{v})\right] & = \dfrac{N - n_{z, \bm{v}}}{n_{z, \bm{v}} N^2} \cdot \dfrac{N}{N - 1} \sum \limits_{i = 1}^N \left[y_i(z, \bm{v}) - \mu(z, \bm{v})\right]^2 = \dfrac{N - n_{z, \bm{v}}}{N} \dfrac{S^2(z, \bm{v})}{n_{z, \bm{v}}}.
\end{align}

For the covariance terms in \eqref{eq: theta var decomp A} and \eqref{eq: theta var decomp B}, fix $z$ and $\bm{v} \neq \bm{v}^{\prime}$. SUTVA (\Cref{assm: SUTVA}) and the bilinearity of covariance imply that
\small
\begin{align*}
\Cov\left[\hat{\mu}(z, \bm{v}), \hat{\mu}(z, \bm{v}^{\prime})\right] & = \Cov\left[
\dfrac{1}{n_{z, \bm{v}}} \sum \limits_{i = 1}^N y_i(z, \bm{v}) \mathbbm{1}\{\bm{T}_i = (z, \bm{v})\}, \, \dfrac{1}{n_{z, \bm{v}^{\prime}}} \sum \limits_{j=1}^N y_j(z, \bm{v}^{\prime}) \mathbbm{1}\{\bm{T}_j = (z, \bm{v}^{\prime})\}\right] \\
& = \dfrac{1}{n_{z, \bm{v}} n_{z, \bm{v}^{\prime}}} \sum \limits_{i = 1}^N \sum_{j = 1}^N y_i(z, \bm{v}) y_j(z, \bm{v}^{\prime}) \Cov\left[\mathbbm{1}\{\bm{T}_i = (z, \bm{v})\}, \mathbbm{1}\{\bm{T}_j = (z, \bm{v}^{\prime})\}\right].
\end{align*}
\normalsize
Then \eqref{eq: assign Cov II} and \eqref{eq: assign Cov III} of \Cref{lem: randomization moms} imply that
\small
\begin{align}
\Cov\left[\hat{\mu}(z, \bm{v}), \hat{\mu}(z, \bm{v}^{\prime})\right] & = \dfrac{1}{n_{z, \bm{v}} n_{z, \bm{v}^{\prime}}}
\Bigg[\sum \limits_{i = 1}^N y_i(z, \bm{v}) y_i(z, \bm{v}^{\prime}) \left(- \dfrac{n_{z, \bm{v}} n_{z, \bm{v}^{\prime}}}{N^2}\right) \nonumber \\
& \qquad \qquad + \sum \limits_{i \neq j} y_i(z, \bm{v}) y_j(z, \bm{v}^{\prime}) \dfrac{n_{z, \bm{v}} n_{z, \bm{v}^{\prime}}}{N^2(N - 1)}\Bigg] \nonumber \\
& = \dfrac{1}{N^2} \left[- \sum \limits_{i = 1}^N y_i(z, \bm{v}) y_i(z, \bm{v}^{\prime}) + \dfrac{1}{N - 1} \sum \limits_{i \neq j} y_i(z, \bm{v}) y_j(z, \bm{v}^{\prime})\right]. \label{eq: cov bracket}
\end{align}
\normalsize
The remaining steps rewrite the factor in brackets in \eqref{eq: cov bracket} as a multiple of a finite-population covariance. Splitting the product of the two sums into $i = j$ and $i \neq j$ terms and using $\sum_{i = 1}^N y_i(z, \bm{v}) = N \mu(z, \bm{v})$,
\begin{align*}
\sum \limits_{i \neq j} y_i(z, \bm{v}) y_j(z, \bm{v}^{\prime}) & = \left(\sum \limits_{i = 1}^N y_i(z, \bm{v})\right) \left(\sum \limits_{j = 1}^N y_j(z, \bm{v}^{\prime})\right) - \sum \limits_{i = 1}^N y_i(z, \bm{v}) y_i(z, \bm{v}^{\prime}) \\
& = N^2 \mu(z, \bm{v})\, \mu(z, \bm{v}^{\prime}) - \sum \limits_{i = 1}^N y_i(z, \bm{v}) y_i(z, \bm{v}^{\prime}).
\end{align*}
Substituting this expression for $\sum_{i \neq j} y_i(z, \bm{v}) y_j(z, \bm{v}^{\prime})$ into the factor in brackets in \eqref{eq: cov bracket} and combining the two terms of that factor over the common denominator $N - 1$, which gives the product term $\sum_{i = 1}^N y_i(z, \bm{v}) y_i(z, \bm{v}^{\prime})$ the coefficient $-(N - 1) - 1 = -N$,
\small
\begin{align*}
& - \sum \limits_{i = 1}^N y_i(z, \bm{v}) y_i(z, \bm{v}^{\prime}) + \dfrac{N^2 \mu(z, \bm{v})\, \mu(z, \bm{v}^{\prime}) - \sum_{i = 1}^N y_i(z, \bm{v}) y_i(z, \bm{v}^{\prime})}{N - 1} \\
& = - \dfrac{N}{N - 1} \left[\sum \limits_{i = 1}^N y_i(z, \bm{v}) y_i(z, \bm{v}^{\prime}) - N \mu(z, \bm{v})\, \mu(z, \bm{v}^{\prime})\right].
\end{align*}
\normalsize
Expanding the product of the two deviations from the cell means $\mu(z, \bm{v})$ and $\mu(z, \bm{v}^{\prime})$ in each summand below and using the identities $\sum_{i = 1}^N y_i(z, \bm{v}) = N \mu(z, \bm{v})$ and $\sum_{i = 1}^N y_i(z, \bm{v}^{\prime}) = N \mu(z, \bm{v}^{\prime})$ shows that
\begin{align*}
\sum \limits_{i = 1}^N \left[y_i(z, \bm{v}) - \mu(z, \bm{v})\right] \left[y_i(z, \bm{v}^{\prime}) - \mu(z, \bm{v}^{\prime})\right] = \sum \limits_{i = 1}^N y_i(z, \bm{v}) y_i(z, \bm{v}^{\prime}) - N \mu(z, \bm{v})\, \mu(z, \bm{v}^{\prime}).
\end{align*}
By the definition of the finite-population covariance in \eqref{eq: S cov def}, the left-hand side of the display above equals $(N - 1)\, S\left(z, \bm{v}; z, \bm{v}^{\prime}\right)$, so the factor in brackets in \eqref{eq: cov bracket} equals $- \dfrac{N}{N - 1} \cdot (N - 1)\, S\left(z, \bm{v}; z, \bm{v}^{\prime}\right) = - N S\left(z, \bm{v}; z, \bm{v}^{\prime}\right)$. Restoring the leading factor $1/N^2$ that multiplies the brackets in \eqref{eq: cov bracket},
\begin{align*}
\Cov\left[\hat{\mu}(z, \bm{v}), \hat{\mu}(z, \bm{v}^{\prime})\right] & = - \dfrac{1}{N} S\left(z, \bm{v}; z, \bm{v}^{\prime}\right).
\end{align*}

Consequently, the first two components of $\Var[\widehat{\theta}]$ in \eqref{eq: theta var decomp A} and \eqref{eq: theta var decomp B}, respectively, are
\small
\begin{align}
\Var\left[\sum \limits_{\bm{v} \in \mathcal{V}} g_1(\bm{v}) \hat{\mu}(1, \bm{v})\right] & = \sum_{\bm{v} \in \mathcal{V}} g_1(\bm{v})^2 \dfrac{N - n_{1, \bm{v}}}{N} \dfrac{S^2(1, \bm{v})}{n_{1, \bm{v}}} - \sum_{\bm{v} \neq \bm{v}^{\prime}} g_1(\bm{v}) g_1(\bm{v}^{\prime}) \dfrac{1}{N} S\left(1, \bm{v}; 1, \bm{v}^{\prime}\right) \label{eq: theta var decomp A final} \\
\Var\left[\sum \limits_{\bm{v} \in \mathcal{V}} g_0(\bm{v}) \hat{\mu}(0, \bm{v})\right] & = \sum_{\bm{v} \in \mathcal{V}} g_0(\bm{v})^2 \dfrac{N - n_{0, \bm{v}}}{N} \dfrac{S^2(0, \bm{v})}{n_{0, \bm{v}}} - \sum_{\bm{v} \neq \bm{v}^{\prime}}  g_0(\bm{v}) g_0(\bm{v}^{\prime}) \dfrac{1}{N} S\left(0, \bm{v}; 0, \bm{v}^{\prime}\right). \label{eq: theta var decomp B final}
\end{align}
\normalsize

For the third component of $\Var[\widehat{\theta}]$ in \eqref{eq: theta var decomp},
\begin{align} \label{eq: theta var decomp C}
- 2 \Cov\left[ \sum \limits_{\bm{v} \in \mathcal{V}} g_1(\bm{v}) \hat{\mu}(1, \bm{v}), \, \sum \limits_{\bm{v} \in \mathcal{V}} g_0(\bm{v}) \hat{\mu}(0, \bm{v})\right],
\end{align}
the bilinearity of covariance implies that 
\small
\begin{align}
\Cov\left[\sum \limits_{\bm{v} \in \mathcal{V}} g_1(\bm{v}) \hat{\mu}(1, \bm{v}), \, \sum \limits_{\bm{v}^{\prime} \in \mathcal{V}} g_0(\bm{v}^{\prime}) \hat{\mu}(0, \bm{v}^{\prime})\right] & = \sum \limits_{\bm{v} \in \mathcal{V}} \sum_{\bm{v}^{\prime} \in \mathcal{V}} g_1(\bm{v}) g_0(\bm{v}^{\prime}) \Cov\left[\hat{\mu}(1, \bm{v}), \hat{\mu}(0, \bm{v}^{\prime})\right]. \label{eq: theta var decomp C cov}
\end{align}
\normalsize
Then SUTVA (\Cref{assm: SUTVA}), together with \eqref{eq: assign Cov II} and \eqref{eq: assign Cov III} of \Cref{lem: randomization moms}, yields
\begin{align*}
\Cov\left[\hat{\mu}(1, \bm{v}), \hat{\mu}(0, \bm{v}^{\prime})\right] & = \dfrac{1}{n_{1, \bm{v}} n_{0, \bm{v}^{\prime}}} \Bigg[\sum \limits_{i = 1}^N y_i(1, \bm{v}) y_i(0, \bm{v}^{\prime}) \left(- \dfrac{n_{1, \bm{v}} n_{0, \bm{v}^{\prime}}}{N^2}\right) \\
& + \sum \limits_{i \neq j} y_i(1, \bm{v}) y_j(0, \bm{v}^{\prime}) \dfrac{n_{1, \bm{v}} n_{0, \bm{v}^{\prime}}}{N^2 (N-1)} \Bigg] \\
& = \dfrac{1}{N^2} \left[- \sum \limits_{i=1}^N y_i(1, \bm{v}) y_i(0, \bm{v}^{\prime})
+ \dfrac{1}{N-1} \sum \limits_{i \neq j} y_i(1, \bm{v}) y_j(0, \bm{v}^{\prime})
\right],
\end{align*}
which reduces by the same three steps as the covariance $\Cov\left[\hat{\mu}(z, \bm{v}), \hat{\mu}(z, \bm{v}^{\prime})\right]$ above --- splitting the product of the two sums, combining over the common denominator $N - 1$, and recognizing the finite-population covariance \eqref{eq: S cov def} --- now with $y_i(1, \bm{v})$ and $y_i(0, \bm{v}^{\prime})$ in place of $y_i(z, \bm{v})$ and $y_i(z, \bm{v}^{\prime})$.
Restoring the leading factor $1/N^2$ that multiplies the brackets in the display above then yields
\begin{align} \label{eq: cov across z}
\Cov\left[\hat{\mu}(1, \bm{v}), \hat{\mu}(0, \bm{v}^{\prime})\right] & = - \dfrac{1}{N} S\left(1, \bm{v}; 0, \bm{v}^{\prime}\right).
\end{align}

Substituting \eqref{eq: cov across z} into \eqref{eq: theta var decomp C cov} then yields
\begin{align*}
\Cov\left[\sum \limits_{\bm{v} \in \mathcal{V}} g_1(\bm{v}) \hat{\mu}(1, \bm{v}), \, \sum \limits_{\bm{v}^{\prime} \in \mathcal{V}} g_0(\bm{v}^{\prime}) \hat{\mu}(0, \bm{v}^{\prime})\right]
& = - \dfrac{1}{N} \sum_{\bm{v} \in \mathcal{V}} \sum_{\bm{v}^{\prime} \in \mathcal{V}} g_1(\bm{v}) g_0(\bm{v}^{\prime}) S(1, \bm{v}; 0, \bm{v}^{\prime}).
\end{align*}
Multiplying the covariance by $-2$ shows that the third component \eqref{eq: theta var decomp C} of the variance decomposition in \eqref{eq: theta var decomp} equals
\begin{align} \label{eq: theta var decomp C final}
\dfrac{2}{N} \sum_{\bm{v} \in \mathcal{V}} \sum_{\bm{v}^{\prime} \in \mathcal{V}} g_1(\bm{v}) g_0(\bm{v}^{\prime}) S(1, \bm{v}; 0, \bm{v}^{\prime}).
\end{align}

Finally, putting \eqref{eq: theta var decomp A final}, \eqref{eq: theta var decomp B final}, and \eqref{eq: theta var decomp C final} together yields
\begin{equation}
\begin{split}
\Var\left[\widehat{\theta}\right] & = \sum \limits_{\bm{v} \in \mathcal{V}} g_1(\bm{v})^2 \dfrac{N - n_{1, \bm{v}}}{N} \dfrac{S^2(1, \bm{v})}{n_{1, \bm{v}}} + \sum \limits_{\bm{v} \in \mathcal{V}} g_0(\bm{v})^2 \dfrac{N - n_{0, \bm{v}}}{N} \dfrac{S^2(0, \bm{v})}{n_{0, \bm{v}}} \\ 
& - \dfrac{1}{N} \sum \limits_{\bm{v} \neq \bm{v}^{\prime}} g_1(\bm{v}) g_1(\bm{v}^{\prime}) S(1, \bm{v}; 1, \bm{v}^{\prime}) - \dfrac{1}{N} \sum \limits_{\bm{v} \neq \bm{v}^{\prime}} g_0(\bm{v}) g_0(\bm{v}^{\prime}) S(0, \bm{v}; 0, \bm{v}^{\prime}) \\ 
& + \dfrac{2}{N} \sum \limits_{\bm{v} \in \mathcal{V}} \sum \limits_{\bm{v}^{\prime} \in \mathcal{V}} g_1(\bm{v}) g_0(\bm{v}^{\prime}) S(1, \bm{v} ; 0, \bm{v}^{\prime}),
\end{split}
\end{equation}
which completes the proof.
\end{proof}

The exact variance in \Cref{prop: theta Var} cannot be estimated. Each unit reveals the potential outcome of at most one cell --- the fundamental problem of causal inference \citep{holland1986} --- so no assignment identifies the cross-cell covariances $S(z, \bm{v}; z^{\prime}, \bm{v}^{\prime})$, which pair potential outcomes that no unit exhibits together. \Cref{cor: theta Var upper} therefore replaces the exact variance with an upper bound whose every ingredient, including the correction term $\Gamma$, is a within-cell variance $S^2(z, \bm{v})$ or its square root, each identified by the units assigned to the corresponding cell. Confidence intervals computed from the upper bound are conservative rather than exact.

\begin{cor} \label{cor: theta Var upper}
Under \Cref{assm: SUTVA}, together with \Cref{assm: CRA} of the main text,
\begin{align*}
\Var\left[\widehat{\theta}\right] & \leq \overline{\Var}\left[\widehat{\theta}\right] = \sum_{z \in \{0, 1\}} \sum_{\bm{v} \in \mathcal{V}} \dfrac{g_z(\bm{v})^2 S^2(z, \bm{v})}{n_{z, \bm{v}}} - \dfrac{1}{N} \max\left\{0, \, \Gamma\right\},
\end{align*}
where $S(z, \bm{v}) \coloneqq \sqrt{S^2(z, \bm{v})}$ and the correction term $\Gamma$ is
\begin{align*}
\Gamma \coloneqq 2 \sum_{z \in \{0, 1\}} \sum_{\bm{v} \in \mathcal{V}} g_z(\bm{v})^2 S^2(z, \bm{v}) - \left(\sum_{\bm{v} \in \mathcal{V}} g_1(\bm{v}) S\left(1, \bm{v}\right) + \sum_{\bm{v} \in \mathcal{V}} g_0(\bm{v}) S\left(0, \bm{v}\right)\right)^2.
\end{align*}
\end{cor}
\begin{proof}
From the expression for the variance in \Cref{prop: theta Var}, $\Var[\widehat{\theta}]$ is
\begin{align}
\Var\left[\widehat{\theta}\right]
& = \sum \limits_{\bm{v} \in \mathcal{V}} g_1(\bm{v})^2 \dfrac{N - n_{1, \bm{v}}}{N} \dfrac{S^2(1, \bm{v})}{n_{1, \bm{v}}} - \dfrac{1}{N} \sum \limits_{\bm{v} \neq \bm{v}^{\prime}} g_1(\bm{v}) g_1(\bm{v}^{\prime}) S(1, \bm{v}; 1, \bm{v}^{\prime}) \label{eq: theta Var line I} \\ 
& + \sum \limits_{\bm{v} \in \mathcal{V}} g_0(\bm{v})^2 \dfrac{N - n_{0, \bm{v}}}{N} \dfrac{S^2(0, \bm{v})}{n_{0, \bm{v}}} - \dfrac{1}{N} \sum \limits_{\bm{v} \neq \bm{v}^{\prime}} g_0(\bm{v}) g_0(\bm{v}^{\prime}) S(0, \bm{v}; 0, \bm{v}^{\prime}) \label{eq: theta Var line II} \\ 
& + \dfrac{2}{N} \sum \limits_{\bm{v} \in \mathcal{V}} \sum \limits_{\bm{v}^{\prime} \in \mathcal{V}} g_1(\bm{v}) g_0(\bm{v}^{\prime}) S(1, \bm{v} ; 0, \bm{v}^{\prime}). \label{eq: theta Var line III}
\end{align}
The proof proceeds in three steps. Step 1 rewrites the exact variance so that every unidentified covariance is collected into a single nonnegative quantity, the variance of the unit-level contrasts; dropping that quantity gives the first upper bound. Step 2 instead bounds each unidentified covariance by Cauchy-Schwarz, which gives the second upper bound and the correction term $\Gamma$. Step 3 takes the smaller of the two upper bounds, which is the bound the corollary states.

\textit{Step 1: an exact rewriting of the variance.} For each unit $i$, define the unit-level contrast
\begin{align*}
\psi_i \coloneqq \sum_{\bm{v} \in \mathcal{V}} g_1(\bm{v})\, y_i(1, \bm{v}) - \sum_{\bm{v} \in \mathcal{V}} g_0(\bm{v})\, y_i(0, \bm{v}).
\end{align*}
Exchanging the finite sums over $i$ and $\bm{v}$ shows that the mean $\bar{\psi} \coloneqq (1/N) \sum_{i=1}^N \psi_i$ equals $\theta$ in \eqref{eq: theta estimand app}, so the finite-population variance $S^2(\psi) \coloneqq \frac{1}{N - 1} \sum_{i=1}^N (\psi_i - \bar{\psi})^2$ measures the spread of the unit-level contrasts around the estimand. By the bilinearity of the finite-population covariance, substituting the definition of $\psi_i$ into $S^2(\psi)$ and multiplying out the squared difference of the two weighted sums yields one term per pair of cells, with diagonal terms $S(z, \bm{v}; z, \bm{v}) = S^2(z, \bm{v})$:
\begin{equation} \label{eq: S2 delta expansion}
\begin{split}
S^2(\psi) & = \sum_{z \in \{0, 1\}} \sum_{\bm{v} \in \mathcal{V}} g_z(\bm{v})^2 S^2(z, \bm{v}) + \sum_{z \in \{0, 1\}} \sum_{\bm{v} \neq \bm{v}^{\prime}} g_z(\bm{v}) g_z(\bm{v}^{\prime}) S(z, \bm{v}; z, \bm{v}^{\prime}) \\
& \qquad - 2 \sum_{\bm{v} \in \mathcal{V}} \sum_{\bm{v}^{\prime} \in \mathcal{V}} g_1(\bm{v}) g_0(\bm{v}^{\prime}) S(1, \bm{v}; 0, \bm{v}^{\prime}).
\end{split}
\end{equation}
Since $(N - n_{z, \bm{v}}) / N n_{z, \bm{v}} = 1/n_{z, \bm{v}} - 1/N$, the first term of \eqref{eq: theta Var line I} and of \eqref{eq: theta Var line II} decomposes as
\begin{align} \label{eq: cell variance split}
\sum \limits_{\bm{v} \in \mathcal{V}} g_z(\bm{v})^2 \dfrac{N - n_{z, \bm{v}}}{N} \dfrac{S^2(z, \bm{v})}{n_{z, \bm{v}}} & = \sum \limits_{\bm{v} \in \mathcal{V}} \dfrac{g_z(\bm{v})^2 S^2(z, \bm{v})}{n_{z, \bm{v}}} - \dfrac{1}{N} \sum_{\bm{v} \in \mathcal{V}} g_z(\bm{v})^2 S^2(z, \bm{v}),
\end{align}
so \eqref{eq: theta Var line I}--\eqref{eq: theta Var line III} regroup as the cell-variance term $\sum_{z \in \{0, 1\}} \sum_{\bm{v} \in \mathcal{V}} g_z(\bm{v})^2 S^2(z, \bm{v}) / n_{z, \bm{v}}$ minus $1/N$ times the right-hand side of \eqref{eq: S2 delta expansion}. The regrouping matches term by term: The decomposition \eqref{eq: cell variance split} supplies the $-1/N$ multiples of the diagonal terms $g_z(\bm{v})^2 S^2(z, \bm{v})$; the second terms of \eqref{eq: theta Var line I} and \eqref{eq: theta Var line II} are the $-1/N$ multiples of the within-condition covariances; and \eqref{eq: theta Var line III}, with its sign flipped by the leading minus, supplies the $+2/N$ multiples of the cross-condition covariances. The regrouping yields the identity
\begin{align} \label{eq: neyman form identity}
\Var\left[\widehat{\theta}\right] & = \sum_{z \in \{0, 1\}} \sum_{\bm{v} \in \mathcal{V}} \dfrac{g_z(\bm{v})^2 S^2(z, \bm{v})}{n_{z, \bm{v}}} - \dfrac{S^2(\psi)}{N},
\end{align}
the multi-cell analogue of the classical decomposition of the variance of a difference in means, in which $S^2(\psi)$ plays the role of the variance of the unit-level effects. Every unidentified covariance now enters the variance only through $S^2(\psi)$. As a sum of squared deviations divided by $N - 1$, $S^2(\psi) \geq 0$, so dropping the term $- S^2(\psi)/N$ from the identity \eqref{eq: neyman form identity} can only increase the right-hand side. Dropping the term gives the first upper bound,
\begin{align} \label{eq: theta Var LD bound}
\Var\left[\widehat{\theta}\right] & \leq \sum_{z \in \{0, 1\}} \sum_{\bm{v} \in \mathcal{V}} \dfrac{g_z(\bm{v})^2 S^2(z, \bm{v})}{n_{z, \bm{v}}}.
\end{align}

\textit{Step 2: the Cauchy-Schwarz bound.} The second route bounds the unidentified covariances in \eqref{eq: theta Var line I}--\eqref{eq: theta Var line III} one pair of cells at a time, replacing each covariance by a product of identified standard deviations. Fix $z \in \{0, 1\}$. By Cauchy-Schwarz,
\begin{align*}
\abs{S(z, \bm{v}; z, \bm{v}^{\prime})} \leq \sqrt{S^2(z, \bm{v}) S^2(z, \bm{v}^{\prime})},
\end{align*}
and in particular
\begin{align*}
S(z, \bm{v}; z, \bm{v}^{\prime}) \geq - \sqrt{S^2(z, \bm{v}) S^2(z, \bm{v}^{\prime})}.
\end{align*}
The lower bound is the direction the proof needs because the within-condition covariances enter \eqref{eq: theta Var line I} and \eqref{eq: theta Var line II} with a minus sign. Multiplying the lower bound by the nonnegative product of weights $g_z(\bm{v})\, g_z(\bm{v}^{\prime})$ preserves the bound's direction, negating turns the lower bound into an upper bound, and summing over the pairs $\bm{v} \neq \bm{v}^{\prime}$ preserves the upper bound, so
\begin{align} \label{eq: within cov bound}
- \dfrac{1}{N} \sum_{\bm{v} \neq \bm{v}^{\prime}} 
g_z(\bm{v}) g_z(\bm{v}^{\prime}) S(z, \bm{v}; z, \bm{v}^{\prime})
& \leq \dfrac{1}{N} \sum_{\bm{v} \neq \bm{v}^{\prime}} 
g_z(\bm{v}) g_z(\bm{v}^{\prime}) 
\sqrt{S^2(z, \bm{v}) S^2(z, \bm{v}^{\prime})} \nonumber \\
& = \dfrac{1}{N} \sum_{\bm{v} \neq \bm{v}^{\prime}}
g_z(\bm{v}) g_z(\bm{v}^{\prime}) S(z, \bm{v}) S(z, \bm{v}^{\prime}),
\end{align}
where the equality uses the definition $S(z, \bm{v}) = \sqrt{S^2(z, \bm{v})}$ from the statement of the corollary.

Then, by the decomposition \eqref{eq: cell variance split} of the first term of \eqref{eq: theta Var line I} and \eqref{eq: theta Var line II}, together with the identity
\begin{align} \label{eq: off diagonal square identity}
\sum_{\bm{v} \neq \bm{v}^{\prime}} g_z(\bm{v}) S(z, \bm{v}) g_z(\bm{v}^{\prime}) S(z, \bm{v}^{\prime}) = \left(\sum_{\bm{v}} g_z(\bm{v}) S(z, \bm{v})\right)^2 - \sum_{\bm{v}} g_z(\bm{v})^2 S^2(z, \bm{v}),
\end{align}
it follows that
\begin{align*}
\dfrac{1}{N} \sum_{\bm{v} \neq \bm{v}^{\prime}} g_z(\bm{v}) g_z(\bm{v}^{\prime}) S(z, \bm{v}) S(z, \bm{v}^{\prime}) & = \dfrac{1}{N} \left[\left(\sum_{\bm{v} \in \mathcal{V}} g_z(\bm{v}) S(z, \bm{v})\right)^2 - \sum \limits_{\bm{v} \in \mathcal{V}} g_z(\bm{v})^2 S^2(z, \bm{v})\right].
\end{align*}
Combining, for fixed $z$, the decomposition \eqref{eq: cell variance split}, the covariance bound \eqref{eq: within cov bound}, and the identity \eqref{eq: off diagonal square identity} bounds each of \eqref{eq: theta Var line I} and \eqref{eq: theta Var line II} by identified quantities alone:
\small
\begin{align*}
& \sum \limits_{\bm{v} \in \mathcal{V}} g_z(\bm{v})^2 \dfrac{N - n_{z, \bm{v}}}{N} \dfrac{S^2(z, \bm{v})}{n_{z, \bm{v}}} - \dfrac{1}{N} \sum \limits_{\bm{v} \neq \bm{v}^{\prime}} g_z(\bm{v}) g_z(\bm{v}^{\prime}) S(z, \bm{v}; z, \bm{v}^{\prime}) \\
& \leq \sum \limits_{\bm{v} \in \mathcal{V}} \dfrac{g_z(\bm{v})^2 S^2(z, \bm{v})}{n_{z, \bm{v}}} - \dfrac{1}{N} \sum_{\bm{v} \in \mathcal{V}} g_z(\bm{v})^2 S^2(z, \bm{v}) + \dfrac{1}{N} \left(\sum_{\bm{v} \in \mathcal{V}} g_z(\bm{v}) S(z, \bm{v})\right)^2 - \dfrac{1}{N} \sum \limits_{\bm{v} \in \mathcal{V}} g_z(\bm{v})^2 S^2(z, \bm{v}) \\
& = \sum \limits_{\bm{v} \in \mathcal{V}} \dfrac{g_z(\bm{v})^2 S^2(z, \bm{v})}{n_{z, \bm{v}}} - \dfrac{2}{N} \sum_{\bm{v} \in \mathcal{V}} g_z(\bm{v})^2 S^2(z, \bm{v}) + \dfrac{1}{N} \left(\sum_{\bm{v} \in \mathcal{V}} g_z(\bm{v}) S(z, \bm{v})\right)^2,
\end{align*}
\normalsize
where the equality collects the two $- \frac{1}{N} \sum_{\bm{v} \in \mathcal{V}} g_z(\bm{v})^2 S^2(z, \bm{v})$ terms. Therefore, at $z = 1$ and $z = 0$ respectively, \eqref{eq: theta Var line I} and \eqref{eq: theta Var line II} are bounded above by
\begin{align}
\sum \limits_{\bm{v} \in \mathcal{V}} \dfrac{g_1(\bm{v})^2 S^2(1, \bm{v})}{n_{1, \bm{v}}} - \dfrac{2}{N} \sum_{\bm{v} \in \mathcal{V}} g_1(\bm{v})^2 S^2(1, \bm{v}) + \dfrac{1}{N} \left(\sum \limits_{\bm{v} \in \mathcal{V}} g_1(\bm{v}) S(1, \bm{v})\right)^2 \label{eq: theta Var line I UB}
\end{align}
and
\begin{align}
\sum \limits_{\bm{v} \in \mathcal{V}} \dfrac{g_0(\bm{v})^2 S^2(0, \bm{v})}{n_{0, \bm{v}}} - \dfrac{2}{N} \sum_{\bm{v} \in \mathcal{V}} g_0(\bm{v})^2 S^2(0, \bm{v}) + \dfrac{1}{N} \left(\sum \limits_{\bm{v} \in \mathcal{V}} g_0(\bm{v}) S(0, \bm{v})\right)^2, \label{eq: theta Var line II UB}
\end{align}
respectively.

For \eqref{eq: theta Var line III}, Cauchy-Schwarz implies that
\begin{align*}
\abs{S(1, \bm{v}; 0, \bm{v}^{\prime})} \leq \sqrt{S^2(1, \bm{v}) S^2(0, \bm{v}^{\prime})} = S(1, \bm{v}) S(0, \bm{v}^{\prime}),
\end{align*}
and, since $g_1(\bm{v}), g_0(\bm{v}^{\prime}) \geq 0$,
\begin{align}
\dfrac{2}{N} \sum_{\bm{v},\bm{v}^{\prime}} g_1(\bm{v}) g_0(\bm{v}^{\prime}) S(1, \bm{v}; 0, \bm{v}^{\prime}) & \leq \dfrac{2}{N} \sum_{\bm{v},\bm{v}^{\prime}} g_1(\bm{v}) g_0(\bm{v}^{\prime}) S(1, \bm{v}) S(0, \bm{v}^{\prime}) \notag \\
& = \dfrac{2}{N} \left(\sum_{\bm{v} \in \mathcal{V}} g_1(\bm{v}) S(1, \bm{v})\right)
\left(\sum_{\bm{v}^{\prime} \in \mathcal{V}} g_0(\bm{v}^{\prime}) S(0, \bm{v}^{\prime})\right). \label{eq: theta Var line III UB}
\end{align}

Summing the three upper bounds --- \eqref{eq: theta Var line I UB} for \eqref{eq: theta Var line I}, \eqref{eq: theta Var line II UB} for \eqref{eq: theta Var line II}, and \eqref{eq: theta Var line III UB} for \eqref{eq: theta Var line III} --- bounds the whole variance:
\begin{align*}
\Var\left[\widehat{\theta}\right] & \leq \sum_{z \in \{0, 1\}} \sum_{\bm{v} \in \mathcal{V}} \dfrac{g_z(\bm{v})^2 S^2(z, \bm{v})}{n_{z, \bm{v}}} - \dfrac{2}{N} \sum_{z \in \{0, 1\}} \sum_{\bm{v} \in \mathcal{V}} g_z(\bm{v})^2 S^2(z, \bm{v}) \\
& + \dfrac{1}{N}  \left(\sum_{\bm{v} \in \mathcal{V}} g_1(\bm{v}) S(1, \bm{v})\right)^2 + \dfrac{1}{N} \left(\sum_{\bm{v} \in \mathcal{V}} g_0(\bm{v}) S(0, \bm{v})\right)^2 \\
& + \dfrac{2}{N}\left(\sum_{\bm{v} \in \mathcal{V}} g_1(\bm{v}) S(1, \bm{v})\right) \left(\sum_{\bm{v}^{\prime} \in \mathcal{V}} g_0(\bm{v}^{\prime}) S(0, \bm{v}^{\prime})\right).
\end{align*}
Finally, the three square terms in the bound above collect into a single square,
\begin{align*}
& \dfrac{1}{N} \left(\sum_{\bm{v} \in \mathcal{V}} g_1(\bm{v}) S(1, \bm{v})\right)^2 + \dfrac{1}{N}  \left(\sum_{\bm{v} \in \mathcal{V}} g_0(\bm{v}) S(0, \bm{v})\right)^2 \\
& + \dfrac{2}{N} \left(\sum_{\bm{v} \in \mathcal{V}} g_1(\bm{v}) S(1, \bm{v})\right) \left(\sum_{\bm{v}^{\prime} \in \mathcal{V}} g_0(\bm{v}^{\prime}) S(0, \bm{v}^{\prime})\right) \\
& = \dfrac{1}{N} \left(\sum_{\bm{v} \in \mathcal{V}} g_1(\bm{v}) S(1, \bm{v}) + \sum_{\bm{v} \in \mathcal{V}} g_0(\bm{v}) S(0, \bm{v})\right)^2.
\end{align*}
By the definition of $\Gamma$ in the statement of the corollary, the subtracted diagonal term $- \frac{2}{N} \sum_{z} \sum_{\bm{v}} g_z(\bm{v})^2 S^2(z, \bm{v})$ and the collected square together equal $- \Gamma / N$. The substitution expresses the bound through $\Gamma$,
\begin{align} \label{eq: theta Var CS bound}
\Var\left[\widehat{\theta}\right] & \leq \sum_{z \in \{0, 1\}} \sum_{\bm{v} \in \mathcal{V}} \dfrac{g_z(\bm{v})^2 S^2(z, \bm{v})}{n_{z, \bm{v}}} - \dfrac{\Gamma}{N}.
\end{align}
By the identity \eqref{eq: neyman form identity}, the Cauchy-Schwarz bound \eqref{eq: theta Var CS bound} is equivalent to the lower bound $S^2(\psi) \geq \Gamma$. The two steps thus bound the same unidentified quantity $S^2(\psi)$ from below, once by $0$ and once by $\Gamma$, and neither lower bound is always the larger of the two because $\Gamma$ may take either sign.

\textit{Step 3: combining the two bounds.} Steps 1 and 2 each establish a valid upper bound on $\Var[\widehat{\theta}]$ --- \eqref{eq: theta Var LD bound} and \eqref{eq: theta Var CS bound} --- so the variance is at most the smaller of the two upper bounds. Subtracting $\max\{0, \Gamma\}/N$ from the cell-variance term subtracts the larger of the two correction terms, $0$ or $\Gamma$, and therefore selects the smaller upper bound:
\begin{align*}
\Var\left[\widehat{\theta}\right] & \leq \sum_{z \in \{0, 1\}} \sum_{\bm{v} \in \mathcal{V}} \dfrac{g_z(\bm{v})^2 S^2(z, \bm{v})}{n_{z, \bm{v}}} - \dfrac{1}{N} \max\left\{0, \, \Gamma\right\},
\end{align*}
which completes the proof.
\end{proof}

\noindent The identity \eqref{eq: neyman form identity} is the multi-cell form of a decomposition that goes back to \citet{neyman1923} and appears, for two arms, as Theorem~6.2 of \citet[][p.~89]{imbensrubin2015} and, for a general contrast of the means of the arms of a multi-arm experiment, as Theorem~3 of \citet{liding2017}. Discarding the variance of the unit-level contrasts outright, as in \eqref{eq: theta Var LD bound}, is the standard route to a conservative variance in randomization-based inference: For two arms, the resulting bound is what the Neyman variance estimator estimates \citep[][Theorem~6.3 and pp.~92--93]{imbensrubin2015}, and for the general contrast, the bound and its plug-in estimator appear in \citet[][Proposition~3]{liding2017}. Bounding unidentified covariances by Cauchy-Schwarz likewise goes back to \citet{neyman1923}: Given only the identified second moments, the Cauchy-Schwarz bounds on an unidentified covariance are the sharpest available, and sharper bounds require bringing the marginal distributions of the potential outcomes to bear \citep[][pp.~852--854]{aronowetal2014}.

How the two branches relate depends on the number of cells. With a single profile in $\mathcal{V}$ --- the classical two-arm experiment --- the correction term reduces to $\Gamma = \left[g_1(\bm{v}) S(1, \bm{v}) - g_0(\bm{v}) S(0, \bm{v})\right]^2 \geq 0$. The Cauchy-Schwarz branch is therefore never looser than the branch that discards $S^2(\psi)$, the two branches coincide exactly when $g_1(\bm{v}) S(1, \bm{v}) = g_0(\bm{v}) S(0, \bm{v})$, and the maximum always selects the Cauchy-Schwarz branch, which is the classical Neyman bound for a two-arm experiment \citep[][p.~853]{aronowetal2014}. The two classical routes to a conservative variance do coincide in the two-arm case once Cauchy-Schwarz is completed, as in \citet{neyman1923}, with the inequality of arithmetic and geometric means: The completion bounds $2 g_1(\bm{v}) g_0(\bm{v}) S(1, \bm{v}) S(0, \bm{v})$ by $g_1(\bm{v})^2 S^2(1, \bm{v}) + g_0(\bm{v})^2 S^2(0, \bm{v})$, and the resulting bound is exactly \eqref{eq: theta Var LD bound}, the bound that discarding $S^2(\psi)$ produces.

With two or more profiles, no ordering between the branches holds: $\Gamma$ may take either sign, so either branch can be the tighter one, and the maximum genuinely chooses. \Cref{ex: branch divergence} shows the two branches taking different values on one population, with the true variance below both.

\begin{example}[The two branches diverging in a four-prosecutor population] \label{ex: branch divergence}
Simplify the charging example of \Cref{fig: assignment arrays} in the main text to four prosecutors, indexed $i \in \{1, 2, 3, 4\}$. Each case file pairs a race, White ($z = 0$) or Black ($z = 1$), with a single nonracial feature, neighborhood disadvantage $\bm{v} \in \{0, 1\}$, coded $\bm{v} = 1$ for a disadvantaged neighborhood as in the worked examples above. The pairing gives four configurations. The design assigns each of the four configurations to exactly one of the four prosecutors, uniformly at random, so every cell count is $n_{z, \bm{v}} = 1$.

Cells of size one are permissible here because the example evaluates the exact variance and the two bounds, which are population quantities defined for any positive cell counts. The requirement $n_{z, \bm{v}} \geq 2$ applies only to the plug-in estimator of the variance bound, stated in \eqref{eq: plug in var est} below, which the example does not compute.

The two racial conditions share the uniform weighting of panel (a) of \Cref{fig: assignment arrays}, $g_1(\bm{v}) = g_0(\bm{v}) = 1/2$ at both values of $\bm{v}$, so the target is the AEE. The potential outcomes record whether each prosecutor would charge ($1$) or decline ($0$) each file: Prosecutor 1 would charge every file, prosecutors 2 and 3 would charge exactly the Black-marked file from the disadvantaged neighborhood, the configuration $(1, 1)$, and prosecutor 4 would charge no file.
\begin{center}
{\small
\begin{tabular}{@{}lcccc@{}}
\toprule
& \multicolumn{2}{c}{White ($z = 0$)} & \multicolumn{2}{c}{Black ($z = 1$)} \\
\cmidrule(lr){2-3} \cmidrule(lr){4-5}
$y_i(z, \bm{v})$ & $\bm{v} = 0$ & $\bm{v} = 1$ & $\bm{v} = 0$ & $\bm{v} = 1$ \\
\midrule
Prosecutor 1 & 1 & 1 & 1 & 1 \\
Prosecutor 2 & 0 & 0 & 0 & 1 \\
Prosecutor 3 & 0 & 0 & 0 & 1 \\
Prosecutor 4 & 0 & 0 & 0 & 0 \\
\midrule
Cell mean $\mu(z, \bm{v})$ & $1/4$ & $1/4$ & $1/4$ & $3/4$ \\
Cell variance $S^2(z, \bm{v})$ & $1/4$ & $1/4$ & $1/4$ & $1/4$ \\
\bottomrule
\end{tabular}}
\end{center}
Each column contains either one or three ones, so every cell mean is $1/4$ or $3/4$, every cell variance is $S^2(z, \bm{v}) = 1/4$, and $S(z, \bm{v}) = 1/2$ for every cell.

The two branches of the variance bound in \Cref{cor: theta Var upper} now separate. The branch that discards $S^2(\psi)$ gives the upper bound on $\Var[\widehat{\theta}]$
\begin{align*}
\sum_{z \in \{0, 1\}} \sum_{\bm{v}} \dfrac{g_z(\bm{v})^2 S^2(z, \bm{v})}{n_{z, \bm{v}}} = 4 \cdot \left(\dfrac{1}{2}\right)^2 \cdot \dfrac{1/4}{1} = \dfrac{1}{4}.
\end{align*}
The correction term is
\begin{align*}
\Gamma = 2 \sum_{z \in \{0, 1\}} \sum_{\bm{v}} g_z(\bm{v})^2 S^2(z, \bm{v}) - \left(\sum_{z \in \{0, 1\}} \sum_{\bm{v}} g_z(\bm{v}) S(z, \bm{v})\right)^2 = 2 \cdot \dfrac{1}{4} - \left(4 \cdot \dfrac{1}{2} \cdot \dfrac{1}{2}\right)^2 = - \dfrac{1}{2},
\end{align*}
so the Cauchy-Schwarz branch gives the upper bound $\frac{1}{4} - \left(-\frac{1}{2}\right)/4 = \frac{1}{4} + \frac{1}{8} = \frac{3}{8}$ on $\Var[\widehat{\theta}]$.

For the exact variance, evaluate the unit-level contrasts $\psi_i = \frac{1}{2}\left[y_i(1, 0) + y_i(1, 1)\right] - \frac{1}{2}\left[y_i(0, 0) + y_i(0, 1)\right]$. The four values are $(0, \frac{1}{2}, \frac{1}{2}, 0)$, with mean $1/4$; the four deviations from the mean are $\pm 1/4$, so the variance of the unit-level contrasts is $S^2(\psi) = \frac{1}{3} \cdot 4 \cdot \frac{1}{16} = \frac{1}{12}$. The identity \eqref{eq: neyman form identity} then gives
\begin{align*}
\Var\left[\widehat{\theta}\right] = \dfrac{1}{4} - \dfrac{1/12}{4} = \dfrac{11}{48}.
\end{align*}

The same value follows from direct calculation over the set of assignments. The direct calculation uses only the definition of the variance over equally likely assignments, so agreement with the value above checks the identity \eqref{eq: neyman form identity} and the expression for the variance in \Cref{prop: theta Var} on this population. Under \Cref{assm: CRA} of the main text, the set of assignments $\Omega$ contains the $4! = 24$ ways to assign the four configurations to the four prosecutors, each with probability $1/24$.

Prosecutors 2 and 3 have identical potential outcomes, so $\widehat{\theta}$ depends on the assignment only through the configurations that prosecutors 1 and 4 read. Prosecutor 1 reads any one of the four configurations, and prosecutor 4 then reads any one of the remaining three, so the pair of configurations takes $4 \times 3 = 12$ equally likely values, each value collecting two of the $24$ assignments. Each cell mean is the potential outcome of the lone prosecutor assigned to that configuration: The mean $\hat{\mu}(1, 1)$ equals $1$ unless prosecutor 4 reads the file $(1, 1)$, and each of the other three cell means equals $1$ exactly when prosecutor 1 reads the corresponding file. Therefore
\begin{align*}
\widehat{\theta} = \dfrac{1}{2} \Big[ & \mathbbm{1}\{\text{prosecutor 4 does not read } (1, 1)\} + \mathbbm{1}\{\text{prosecutor 1 reads } (1, 0)\} \\
& - \mathbbm{1}\{\text{prosecutor 1 reads } (0, 1)\} - \mathbbm{1}\{\text{prosecutor 1 reads } (0, 0)\} \Big],
\end{align*}
and evaluating the four indicators at each of the twelve values of the pair of configurations gives the value of $\widehat{\theta}$ at every assignment:
\begin{center}
{\small
\begin{tabular}{@{}lcccc@{}}
\toprule
& \multicolumn{4}{c}{Configuration read by prosecutor 4} \\
\cmidrule(lr){2-5}
Configuration read by prosecutor 1 & $(0, 0)$ & $(0, 1)$ & $(1, 0)$ & $(1, 1)$ \\
\midrule
$(0, 0)$ & --- & $0$ & $0$ & $-1/2$ \\
$(0, 1)$ & $0$ & --- & $0$ & $-1/2$ \\
$(1, 0)$ & $1$ & $1$ & --- & $1/2$ \\
$(1, 1)$ & $1/2$ & $1/2$ & $1/2$ & --- \\
\bottomrule
\end{tabular}}
\end{center}
The twelve equally likely entries give the distribution of $\widehat{\theta}$: the value $1$ with probability $2/12$, the value $1/2$ with probability $4/12$, the value $0$ with probability $4/12$, and the value $-1/2$ with probability $2/12$. The mean over the set of assignments is
\begin{align*}
\E\left[\widehat{\theta}\right] = \dfrac{2 \cdot 1 + 4 \cdot \frac{1}{2} + 4 \cdot 0 + 2 \cdot \left(-\frac{1}{2}\right)}{12} = \dfrac{3}{12} = \dfrac{1}{4},
\end{align*}
which equals the estimand $\theta = \frac{1}{2}\left[\mu(1, 0) + \mu(1, 1)\right] - \frac{1}{2}\left[\mu(0, 0) + \mu(0, 1)\right] = \frac{1}{4}$, as the unbiasedness of $\widehat{\theta}$ for $\theta$ in \Cref{prop: theta EV} guarantees. The second moment over the set of assignments is
\begin{align*}
\E\left[\widehat{\theta}^2\right] = \dfrac{2 \cdot 1 + 4 \cdot \frac{1}{4} + 4 \cdot 0 + 2 \cdot \frac{1}{4}}{12} = \dfrac{7/2}{12} = \dfrac{7}{24},
\end{align*}
so the variance calculated directly over the set of assignments is
\begin{align*}
\Var\left[\widehat{\theta}\right] = \dfrac{7}{24} - \left(\dfrac{1}{4}\right)^2 = \dfrac{14}{48} - \dfrac{3}{48} = \dfrac{11}{48},
\end{align*}
in exact agreement with the value that \eqref{eq: neyman form identity} produced above from the expression for the variance in \Cref{prop: theta Var}. The three quantities are ordered as
\begin{align*}
\Var\left[\widehat{\theta}\right] = \dfrac{11}{48} \; < \; \dfrac{12}{48} = \dfrac{1}{4} \;\, \text{(discarding $S^2(\psi)$)} \; < \; \dfrac{18}{48} = \dfrac{3}{8} \;\, \text{(Cauchy-Schwarz)}.
\end{align*}
Both branches bound the true variance, and the two upper bounds are not equal. Because $\Gamma < 0$, the maximum in \Cref{cor: theta Var upper} equals $\max\{0, \Gamma\} = 0$, so the corollary's bound is the smaller of the two upper bounds --- here $12/48$, from the branch that discards $S^2(\psi)$.

The example reverses the ordering of the two-arm case, in which the Cauchy-Schwarz branch is never the looser branch. The reversal has a plain source. Each branch replaces a quantity that the data cannot identify with the value that makes the resulting bound largest. The Cauchy-Schwarz branch replaces every cross-cell covariance at once: each within-condition covariance with its most negative possible value and each cross-condition covariance with its most positive possible value. Yet no single population of potential outcomes places all of the covariances at those extremes simultaneously when four cells have equal spread, so the Cauchy-Schwarz bound exceeds the exact variance of every such population. The branch that discards $S^2(\psi)$ replaces a single quantity, the variance of the unit-level contrasts, with zero, and a population attains zero whenever the unit-level contrasts $\psi_i$ do not vary.

The Cauchy-Schwarz branch becomes the tighter one under a small change to the potential outcomes. Change the potential outcomes of prosecutor 1 so that prosecutor 1 would charge only the file $(1, 1)$; the other three prosecutors keep their potential outcomes. No prosecutor would charge any file other than $(1, 1)$, and the column of the configuration $(1, 1)$ is unchanged. Every cell except $(1, 1)$ is then constant, with $S^2(z, \bm{v}) = 0$, while $S^2(1, 1) = 1/4$ as before. The branch that discards $S^2(\psi)$ gives the upper bound $(1/2)^2 \cdot \frac{1/4}{1} = \frac{1}{16}$. The correction term is now positive, $\Gamma = 2 \cdot \frac{1}{16} - \left(\frac{1}{2} \cdot \frac{1}{2}\right)^2 = \frac{1}{16}$, so the Cauchy-Schwarz branch gives the upper bound $\frac{1}{16} - \frac{1/16}{4} = \frac{3}{64}$, and the corollary's bound is the Cauchy-Schwarz branch. The unit-level contrasts are $\psi_i = \frac{1}{2}\, y_i(1, 1)$, with variance $S^2(\psi) = \frac{1}{4} \cdot \frac{1}{4} = \frac{1}{16} = \Gamma$, so the Cauchy-Schwarz bound equals the exact variance, $3/64$. The mechanism from the four-varying-cell case runs in reverse: With a single varying cell, every cross-cell covariance is zero, so the Cauchy-Schwarz replacements introduce no slack at all, while the branch that discards $S^2(\psi)$ still replaces the positive variance $S^2(\psi) = 1/16$ with zero.
\end{example}

Now define a plug-in estimator of the variance bound, requiring $n_{z, \bm{v}} \geq 2$ for every $(z, \bm{v}) \in \mathcal{T}$ so that the sample variance defined below exists for every cell:
\begin{align}
\widehat{\overline{\Var}}\left[\widehat{\theta}\right] & = \sum_{z \in \{0, 1\}} \sum_{\bm{v} \in \mathcal{V}} \dfrac{g_z(\bm{v})^2 \hat{S}^2(z, \bm{v})}{n_{z, \bm{v}}} - \dfrac{1}{N} \max\left\{0, \, \widehat{\Gamma}\right\}, \label{eq: plug in var est} \\
\widehat{\Gamma} & \coloneqq 2 \sum_{z \in \{0, 1\}} \sum_{\bm{v} \in \mathcal{V}} g_z(\bm{v})^2 \hat{S}^2(z, \bm{v}) - \left(\sum_{\bm{v} \in \mathcal{V}} g_1(\bm{v}) \hat{S}\left(1, \bm{v}\right) + \sum_{\bm{v} \in \mathcal{V}} g_0(\bm{v}) \hat{S}\left(0, \bm{v}\right)\right)^2, \label{eq: Gamma hat def}
\end{align}
where, for each profile $(z, \bm{v})$,
\begin{align*}
\hat{S}^2\left(z, \bm{v}\right) & \coloneqq \dfrac{1}{n_{z, \bm{v}} - 1} \sum \limits_{i = 1}^N \mathbbm{1}\{\bm{T}_i = (z, \bm{v})\} \left[Y_i - \hat{\mu}(z, \bm{v})\right]^2,
\end{align*}
and $\hat{S}(z, \bm{v}) \coloneqq \sqrt{\hat{S}^2(z, \bm{v})}$.

Within each cell, the units assigned to the cell form a simple random sample without replacement from the $N$ units, so the sample variance $\hat{S}^2(z, \bm{v})$ is unbiased for $S^2(z, \bm{v})$ by the classical result for sampling without replacement \citep[][Theorem~2.4]{cochran1977}. The plug-in estimator of the variance bound as a whole is nonetheless not unbiased for $\overline{\Var}[\widehat{\theta}]$ in finite samples: The square root and the maximum are nonlinear, so Jensen's inequality introduces bias through each nonlinearity. The property established below for the plug-in estimator of the variance bound is therefore consistency, not unbiasedness.

The remaining results establish, over a sequence of finite populations of increasing size, the consistency of the plug-in estimator $\widehat{\theta}$ for the limiting parameter, the consistency of the plug-in estimator of the variance bound for the scaled limiting bound, and a central limit theorem for $\widehat{\theta}$. The following regularity conditions govern the sequence.
\begin{assm}[Asymptotic regularity conditions] \label{assm: asymp regularity}
Consider a sequence of finite populations indexed by $N$, with potential outcomes
$\{y_i(z, \bm{v}) : i = 1,\dots,N; \, z \in \{0, 1\}, \bm{v} \in \mathcal{V}\}$
and treatment profiles assigned under complete random assignment with fixed cell counts
$n_{z, \bm{v}}$ for $z \in \{0, 1\}$ and $\bm{v} \in \mathcal{V}$, as in \Cref{assm: CRA}. The potential outcomes, the cell counts, the moments below, and the estimator $\widehat{\theta}$ all depend on $N$. The notation suppresses the index except where an argument compares quantities across population sizes, as with the estimand $\theta_N$, the limiting parameter $\theta_\infty$, and the scaled variance $\sigma^2_{\theta, N}$ below. Throughout the asymptotic results, a plain arrow $\to$ denotes the convergence of a nonrandom sequence of real numbers as $N \to \infty$.
Suppose:
\begin{enumerate}
\item[(a)] The set of attribute profiles $\mathcal{V}$ is finite, and for each profile $(z, \bm{v})$ there exists $\pi_{z, \bm{v},\infty} \in (0,1)$ such that
\begin{align*}
\dfrac{n_{z, \bm{v}}}{N} \to \pi_{z, \bm{v},\infty} \quad \text{as } N \to \infty.
\end{align*}

\item[(b)] For each profile $(z, \bm{v})$, the finite-population second moments are uniformly bounded:
\begin{align*}
\dfrac{1}{N} \sum_{i=1}^N y_i(z, \bm{v})^2 \leq C_2
\end{align*}
for some constant $C_2 < \infty$ that does not depend on $N$. The finite-population means $\mu(z, \bm{v})$, the finite-population variances $S^2(z, \bm{v})$ of \eqref{eq: S2 def}, and the finite-population covariances $S(z, \bm{v}; z^{\prime}, \bm{v}^{\prime})$ of \eqref{eq: S cov def} converge to finite limits, denoted
\begin{align*}
\mu_{z, \bm{v},\infty} \in \R,
\qquad
S^2_{z, \bm{v},\infty} \in [0,\infty),
\qquad
S_{z, \bm{v}; z^{\prime}, \bm{v}^{\prime},\infty} \in \R,
\end{align*}
and we write $S_{z, \bm{v},\infty} \coloneqq \sqrt{S^2_{z, \bm{v},\infty}}$. The Lindeberg-type condition
\begin{align*}
\max_{1 \leq i \leq N} \dfrac{\left[y_i(z, \bm{v}) - \mu(z, \bm{v})\right]^2}{N} \to 0
\end{align*}
also holds for each profile $(z, \bm{v})$.

\item[(c)] The weights $g_z(\bm{v})$, for $z \in \{0, 1\}$ and $\bm{v} \in \mathcal{V}$, do not vary with $N$.

\item[(d)] (Nondegeneracy of the limiting variance.) Under (a)--(c), the scaled variance $N \, \Var\left[\widehat{\theta}\right]$ converges to a finite limit. By the expression for the variance of the plug-in estimator in \eqref{eq: neyman form identity},
\begin{align*}
N \, \Var\left[\widehat{\theta}\right] = \sum_{z \in \{0,1\}} \sum_{\bm{v} \in \mathcal{V}} \dfrac{N}{n_{z, \bm{v}}}\, g_z(\bm{v})^2 S^2(z, \bm{v}) - S^2(\psi),
\end{align*}
and each term on the right-hand side converges under (a)--(c): the ratios $N / n_{z, \bm{v}}$ to $1 / \pi_{z, \bm{v},\infty}$ under (a), the variances $S^2(z, \bm{v})$ to $S^2_{z, \bm{v},\infty}$ under (b), and $S^2(\psi)$ --- a combination of the variances and covariances with weights fixed under (c), as \eqref{eq: S2 delta expansion} shows --- to the same combination of the limits in (b). Condition (d) itself is the assumption that the limit is strictly positive,
\begin{align*}
\sigma^2_{\theta,\infty} \coloneqq \lim_{N \to \infty} N \, \Var\left[\widehat{\theta}\right] \in (0,\infty).
\end{align*}
\end{enumerate}
The convergence statements for random quantities in the results below --- consistency and the central limit theorem --- are with respect to the randomization distribution that complete random assignment induces, conditional on the potential outcomes.
\end{assm}

Define the scaled variance and the scaled limiting bound,
\begin{align*}
\sigma^2_{\theta,N} \coloneqq N \, \Var\left[\widehat{\theta}\right],
\qquad
\overline{\sigma}^2_{\theta,\infty} \coloneqq \lim_{N \to \infty} N \, \overline{\Var}\left[\widehat{\theta}\right],
\end{align*}
where the limit defining $\overline{\sigma}^2_{\theta,\infty}$ exists and is finite by Assumption~\ref{assm: asymp regularity}(a)--(c) and \Cref{cor: theta Var upper}. Explicitly,
\begin{align}
\overline{\sigma}^2_{\theta,\infty}
& =
\sum_{z \in \{0, 1\}} \sum_{\bm{v} \in \mathcal{V}}
\dfrac{g_z(\bm{v})^2 S^2_{z, \bm{v},\infty}}{\pi_{z, \bm{v},\infty}}
- \max\left\{0, \, \Gamma_{\infty}\right\}, \label{eq: scaled limiting bound}
\\
\Gamma_{\infty}
& \coloneqq
2 \sum_{z \in \{0, 1\}} \sum_{\bm{v} \in \mathcal{V}}
g_z(\bm{v})^2 S^2_{z, \bm{v},\infty}
- \left(
\sum_{\bm{v} \in \mathcal{V}} g_1(\bm{v}) S_{1, \bm{v},\infty}
+
\sum_{\bm{v} \in \mathcal{V}} g_0(\bm{v}) S_{0, \bm{v},\infty}
\right)^2, \label{eq: Gamma infinity def}
\end{align}
which follows from Assumption~\ref{assm: asymp regularity}(a)--(c) by taking the limit of each term of the scaled bound $N \, \overline{\Var}[\widehat{\theta}]$: The finite sums converge term by term, with $N / n_{z, \bm{v}} \to 1 / \pi_{z, \bm{v},\infty}$ and $S(z, \bm{v}) \to S_{z, \bm{v},\infty}$, so the correction term converges, $\Gamma \to \Gamma_{\infty}$. The function $x \mapsto \max\{0, x\}$ is continuous at every point of $\R$: The function equals $x$ above zero and $0$ below, and the two pieces agree at zero. A continuous function preserves the limit of a convergent nonrandom sequence, so $\max\{0, \Gamma\} \to \max\{0, \Gamma_{\infty}\}$.

The plug-in estimator $\widehat{\theta}$ is consistent for the limiting parameter.

\begin{prop}[Consistency of the plug-in estimator] \label{prop: theta consistency}
Under Assumptions \ref{assm: SUTVA} and \ref{assm: asymp regularity}, together with \Cref{assm: CRA} of the main text, define the finite-population estimand
\begin{align*}
\theta_N \coloneqq \sum_{\bm{v} \in \mathcal{V}} g_1(\bm{v}) \mu(1, \bm{v}) - \sum_{\bm{v} \in \mathcal{V}} g_0(\bm{v}) \mu(0, \bm{v})
\end{align*}
--- the estimand $\theta$ of \eqref{eq: theta estimand app}, written here with the subscript $N$ to make explicit its dependence on the population size --- and the limiting parameter
\begin{align*}
\theta_\infty \coloneqq \sum_{\bm{v} \in \mathcal{V}} g_1(\bm{v}) \mu_{1, \bm{v},\infty} - \sum_{\bm{v} \in \mathcal{V}} g_0(\bm{v}) \mu_{0, \bm{v},\infty}.
\end{align*}
Then
\begin{align*}
\widehat{\theta} \overset{p}{\longrightarrow} \theta_\infty
\quad \text{as } N \to \infty,
\end{align*}
where $\overset{p}{\longrightarrow}$ denotes convergence in probability under the randomization distribution induced by complete random assignment, conditional on the potential outcomes.
\end{prop}

\begin{proof}
The proof proceeds in four steps. Step 1 establishes the consistency of each cell mean for the corresponding population cell mean. Step 2 combines the cell means into the plug-in estimator and establishes the convergence in probability $\widehat{\theta} - \theta_N \overset{p}{\longrightarrow} 0$. Step 3 establishes the nonrandom convergence $\theta_N \to \theta_\infty$. Step 4 combines the convergence in probability from Step 2 with the nonrandom convergence from Step 3.

\textit{Step 1: consistency of the cell means.} Fix $(z, \bm{v}) \in \{0, 1\} \times \mathcal{V}$. Under complete random assignment with fixed counts $n_{z, \bm{v}}$, the set $\{y_i(z, \bm{v}) : \bm{T}_i = (z, \bm{v})\}$ is a simple random sample without replacement of size $n_{z, \bm{v}}$ from the finite population $\{y_i(z, \bm{v}) : i = 1,\dots,N\}$, and $\hat{\mu}(z, \bm{v})$ is the corresponding sample mean. From the expression for the variance of $\hat{\mu}(z, \bm{v})$ derived in the proof of \Cref{prop: theta Var},
\begin{align*}
\Var\left[\hat{\mu}(z, \bm{v})\right]
= \dfrac{N - n_{z, \bm{v}}}{N} \, \dfrac{S^2(z, \bm{v})}{n_{z, \bm{v}}}.
\end{align*}
By Assumption~\ref{assm: asymp regularity}(a), $n_{z, \bm{v}}/N \to \pi_{z, \bm{v},\infty} \in (0,1)$, so both $n_{z, \bm{v}}$ and $N - n_{z, \bm{v}}$ diverge and $(N - n_{z, \bm{v}})/N \to 1 - \pi_{z, \bm{v},\infty} \in (0,1)$. By Assumption~\ref{assm: asymp regularity}(b), $S^2(z, \bm{v})$ converges to the finite limit $S^2_{z, \bm{v},\infty}$. There is therefore a constant $C_{z, \bm{v}} < \infty$ and an integer $N_0$ such that for all $N \geq N_0$,
\begin{align*}
\Var\left[\hat{\mu}(z, \bm{v})\right] \leq \dfrac{C_{z, \bm{v}}}{n_{z, \bm{v}}}.
\end{align*}
The right-hand side tends to $0$ because $n_{z, \bm{v}} \to \infty$, and the variance is nonnegative, so $\Var\left[\hat{\mu}(z, \bm{v})\right] \to 0$. Under complete random assignment, $\hat{\mu}(z, \bm{v})$ is unbiased for $\mu(z, \bm{v})$ (\Cref{prop: theta EV}), so by Chebyshev's inequality, for any $\eta > 0$,
\begin{align*}
\Pr\left(\left|\hat{\mu}(z, \bm{v}) - \mu(z, \bm{v})\right| > \eta\right)
\leq \dfrac{\Var\left[\hat{\mu}(z, \bm{v})\right]}{\eta^2} \to 0,
\end{align*}
and therefore $\hat{\mu}(z, \bm{v}) - \mu(z, \bm{v}) \overset{p}{\longrightarrow} 0$ for each profile $(z, \bm{v})$.

\textit{Step 2: consistency for the finite-population estimand.} Write
\begin{align*}
\widehat{\theta} - \theta_N
= \sum_{\bm{v} \in \mathcal{V}} g_1(\bm{v}) \left\{\hat{\mu}(1, \bm{v}) - \mu(1, \bm{v})\right\}
   - \sum_{\bm{v} \in \mathcal{V}} g_0(\bm{v}) \left\{\hat{\mu}(0, \bm{v}) - \mu(0, \bm{v})\right\}.
\end{align*}
Because $\mathcal{V}$ is finite and the weights $g_z(\bm{v})$ are fixed by Assumption~\ref{assm: asymp regularity}(c), the right-hand side of the display above is a finite linear combination, with fixed coefficients, of terms that converge in probability to $0$ by Step~1. A finite linear combination of such terms with fixed coefficients itself converges in probability to $0$; hence
\begin{align*}
\widehat{\theta} - \theta_N \overset{p}{\longrightarrow} 0.
\end{align*}

\textit{Step 3: convergence of the finite-population estimand to its limit.} By Assumption~\ref{assm: asymp regularity}(b), $\mu(z, \bm{v}) \to \mu_{z, \bm{v},\infty}$ for each profile $(z, \bm{v})$. A finite sum of convergent sequences with fixed coefficients converges to the same sum of the limits, so
\begin{align*}
\theta_N \longrightarrow \sum_{\bm{v} \in \mathcal{V}} g_1(\bm{v}) \mu_{1, \bm{v},\infty} - \sum_{\bm{v} \in \mathcal{V}} g_0(\bm{v}) \mu_{0, \bm{v},\infty} = \theta_\infty.
\end{align*}

\textit{Step 4: combining the limits.} Adding and subtracting $\theta_N$ splits the error into a random part and a nonrandom part,
\begin{align*}
\widehat{\theta} - \theta_\infty
= \left(\widehat{\theta} - \theta_N\right) + \left(\theta_N - \theta_\infty\right).
\end{align*}
The first term converges in probability to $0$ by Step~2, and the second term is nonrandom and converges to $0$ by Step~3. The sum therefore converges in probability to $0$, which establishes $\widehat{\theta} \overset{p}{\longrightarrow} \theta_\infty$.
\end{proof}

The plug-in estimator of the variance bound in \eqref{eq: plug in var est} is consistent for the scaled limiting bound $\overline{\sigma}^2_{\theta,\infty}$ in \eqref{eq: scaled limiting bound}.

\begin{prop}[Consistency of the plug-in estimator of the variance bound] \label{prop: theta Var upper consistency}
Under Assumptions \ref{assm: SUTVA} and \ref{assm: asymp regularity}, together with \Cref{assm: CRA} of the main text,
\begin{align*}
N \, \widehat{\overline{\Var}}\left[\widehat{\theta}\right] \overset{p}{\longrightarrow} \overline{\sigma}^2_{\theta,\infty}.
\end{align*}
\end{prop}

\begin{proof}
The proof establishes the consistency of each cell's sample variance and then assembles the scaled plug-in estimator of the variance bound by the continuous mapping theorem, which states that convergence in probability is preserved by continuous functions: If $X_N \overset{p}{\longrightarrow} x$ and the function $h$ is continuous at $x$, then $h(X_N) \overset{p}{\longrightarrow} h(x)$.

Fix $(z, \bm{v})$. Under complete random assignment and SUTVA, the set $\{Y_i : \bm{T}_i = (z, \bm{v})\} = \{y_i(z, \bm{v}) : \bm{T}_i = (z, \bm{v})\}$ is a simple random sample without replacement of size $n_{z, \bm{v}}$ from the finite population $\{y_i(z, \bm{v}) : i=1,\dots,N\}$, and $\hat{S}^2(z, \bm{v})$ is the corresponding sample variance, which is unbiased for $S^2(z, \bm{v})$ \citep[][Theorem~2.4]{cochran1977}. Two cases establish $\hat{S}^2(z, \bm{v}) - S^2(z, \bm{v}) \overset{p}{\longrightarrow} 0$, according to whether the limiting variance $S^2_{z, \bm{v},\infty}$ of Assumption~\ref{assm: asymp regularity}(b) is zero or positive. Both cases are possible: $S^2(z, \bm{v})$ measures the spread of the potential outcomes across units rather than a sampling variance, so the limit need not vanish as $N$ grows, and Assumption~\ref{assm: asymp regularity} permits a zero limit for any single cell --- condition (d) restricts only the limiting variance of $\widehat{\theta}$ as a whole. The zero case needs its own argument because the ratio consistency used in the positive case divides by $S^2(z, \bm{v})$.

If $S^2_{z, \bm{v},\infty} = 0$: Markov's inequality states that, for a nonnegative random variable $X$ and any $\eta > 0$, $\Pr\left(X > \eta\right) \leq \E[X]/\eta$. Applied to the nonnegative random variable $\hat{S}^2(z, \bm{v})$, the inequality gives
\begin{align*}
\Pr\left(\hat{S}^2(z, \bm{v}) > \eta\right) \leq \dfrac{\E\left[\hat{S}^2(z, \bm{v})\right]}{\eta} = \dfrac{S^2(z, \bm{v})}{\eta} \to 0,
\end{align*}
because $\hat{S}^2(z, \bm{v})$ is unbiased for $S^2(z, \bm{v})$ and $S^2(z, \bm{v}) \to 0$. The display gives $\hat{S}^2(z, \bm{v}) \overset{p}{\longrightarrow} 0$, and subtracting the nonrandom sequence $S^2(z, \bm{v}) \to 0$ leaves $\hat{S}^2(z, \bm{v}) - S^2(z, \bm{v}) \overset{p}{\longrightarrow} 0$.

If $S^2_{z, \bm{v},\infty} > 0$: Proposition~1 of \citet{liding2017} gives the consistency of the ratio
\begin{align*}
\hat{S}^2(z, \bm{v}) / S^2(z, \bm{v}) \overset{p}{\longrightarrow} 1,
\end{align*}
provided Condition~(4) of their Theorem~1 holds for the cell, namely
\begin{align} \label{eq: liding condition}
\dfrac{1}{\min\left(n_{z, \bm{v}}, N - n_{z, \bm{v}}\right)} \cdot \dfrac{\max_{1 \leq i \leq N} \left[y_i(z, \bm{v}) - \mu(z, \bm{v})\right]^2}{S^2(z, \bm{v})} \to 0.
\end{align}
To verify the condition \eqref{eq: liding condition}, factor its left-hand side into three pieces,
\small
\begin{align*}
\dfrac{N}{\min\left(n_{z, \bm{v}}, N - n_{z, \bm{v}}\right)} \cdot \dfrac{\max_{1 \leq i \leq N} \left[y_i(z, \bm{v}) - \mu(z, \bm{v})\right]^2}{N} \cdot \dfrac{1}{S^2(z, \bm{v})}.
\end{align*}
\normalsize
The first factor converges to $1/\min(\pi_{z, \bm{v},\infty}, 1 - \pi_{z, \bm{v},\infty})$, a finite limit, by Assumption~\ref{assm: asymp regularity}(a). The second factor converges to $0$ by the Lindeberg-type condition of Assumption~\ref{assm: asymp regularity}(b). The third factor converges to $1/S^2_{z, \bm{v},\infty}$, a finite limit, because $S^2_{z, \bm{v},\infty} > 0$ in the present case. A product of two factors with finite limits and one factor that converges to $0$ converges to $0$, so the condition \eqref{eq: liding condition} holds. Since $S^2_{z, \bm{v},\infty} > 0$, the variance $S^2(z, \bm{v})$ is positive for all large $N$, and multiplying and dividing the estimation error by $S^2(z, \bm{v})$ gives
\begin{align*}
\hat{S}^2(z, \bm{v}) - S^2(z, \bm{v}) = S^2(z, \bm{v}) \left[\dfrac{\hat{S}^2(z, \bm{v})}{S^2(z, \bm{v})} - 1\right].
\end{align*}
The factor in brackets converges in probability to $0$ by the consistency of the ratio, and $S^2(z, \bm{v})$ is a nonrandom sequence with a finite limit, hence bounded. The product of a bounded nonrandom sequence and a sequence that converges in probability to $0$ converges in probability to $0$, so $\hat{S}^2(z, \bm{v}) - S^2(z, \bm{v}) \overset{p}{\longrightarrow} 0$.

In both cases $\hat{S}^2(z, \bm{v}) - S^2(z, \bm{v}) \overset{p}{\longrightarrow} 0$, and $S^2(z, \bm{v}) \to S^2_{z, \bm{v},\infty}$ by Assumption~\ref{assm: asymp regularity}(b), so $\hat{S}^2(z, \bm{v}) \overset{p}{\longrightarrow} S^2_{z, \bm{v},\infty}$. Because the square-root function is continuous on $[0,\infty)$, the continuous mapping theorem then gives $\hat{S}(z, \bm{v}) \overset{p}{\longrightarrow} S_{z, \bm{v},\infty}$.

Because $\mathcal{V}$ is finite, the vector $\big(\hat{S}^2(z, \bm{v}), \hat{S}(z, \bm{v})\big)_{z, \bm{v}}$ has finitely many coordinates, and convergence in probability coordinate by coordinate implies convergence in probability of the whole vector. The vector therefore converges in probability to $\big(S^2_{z, \bm{v},\infty}, S_{z, \bm{v},\infty}\big)_{z, \bm{v}}$. The scaled plug-in estimator of the variance bound
\begin{align*}
N \, \widehat{\overline{\Var}}\left[\widehat{\theta}\right]
& = \sum_{z, \bm{v}} g_z(\bm{v})^2 \hat{S}^2(z, \bm{v}) \dfrac{N}{n_{z, \bm{v}}} - \max\left\{0, \, \widehat{\Gamma}\right\}, \\
\widehat{\Gamma} & = 2 \sum_{z, \bm{v}} g_z(\bm{v})^2 \hat{S}^2(z, \bm{v}) - \left(\sum_{\bm{v}} g_1(\bm{v}) \hat{S}(1, \bm{v}) + \sum_{\bm{v}} g_0(\bm{v}) \hat{S}(0, \bm{v})\right)^2,
\end{align*}
is a finite composition of additions, multiplications by the fixed scalars $g_z(\bm{v})^2$ and the ratios $N/n_{z, \bm{v}} \to 1/\pi_{z, \bm{v},\infty}$, squared finite sums, and the function $x \mapsto \max\{0, x\}$, all of which are continuous operations.
The ratios $N/n_{z, \bm{v}}$ are nonrandom, and a nonrandom sequence that converges to a limit also converges in probability to that limit, so the full vector of inputs --- the sample variances, their square roots, and the ratios --- converges in probability jointly. By the continuous mapping theorem, replacing $\hat{S}^2(z, \bm{v})$ by $S^2_{z, \bm{v},\infty}$, $\hat{S}(z, \bm{v})$ by $S_{z, \bm{v},\infty}$, and $N/n_{z, \bm{v}}$ by $1/\pi_{z, \bm{v},\infty}$ yields
\begin{align*}
N \, \widehat{\overline{\Var}}\left[\widehat{\theta}\right]
\overset{p}{\longrightarrow}
\overline{\sigma}^2_{\theta,\infty},
\end{align*}
which completes the proof.
\end{proof}

Finally, $\widehat{\theta}$ satisfies a finite-population central limit theorem, which together with the consistent plug-in estimator of the variance bound justifies the Normal-approximation confidence intervals reported in the application. The intervals require only the consistency just established. By Slutsky's theorem, standardizing $\widehat{\theta} - \theta_N$ by any scale that converges in probability to a limit at least as large as $\sigma_{\theta,\infty}$ produces a statistic whose limiting distribution is Normal with variance at most one, so a confidence interval built from such a scale covers $\theta_N$ with limiting probability at least the interval's nominal level --- coverage that is asymptotically conservative. Conservative coverage delivers valid inference: A test that rejects a hypothesized value of the estimand exactly when the interval excludes the hypothesized value rejects a true value with limiting probability at most the nominal error rate.

\begin{prop}[Finite population CLT for the plug-in estimator] \label{prop: theta CLT}
Under Assumptions \ref{assm: SUTVA} and \ref{assm: asymp regularity}, together with \Cref{assm: CRA} of the main text,
\begin{align*}
\sqrt{N} \left\{\widehat{\theta} - \theta_N\right\} \, / \, \sigma_{\theta,N} \;\; \overset{d}{\longrightarrow} \;\; \mathcal{N}(0,1),
\end{align*}
where $\mathcal{N}(0, 1)$ denotes the Normal distribution with mean $0$ and variance $1$, and $\overset{d}{\longrightarrow}$ denotes convergence in distribution under the randomization distribution induced by complete random assignment, conditional on the potential outcomes.
The scaled variance $\sigma^2_{\theta,N} = N \, \Var\left[\widehat{\theta}\right]$ is given by \Cref{prop: theta Var}, and $\sigma^2_{\theta,N} \to \sigma^2_{\theta,\infty} \in (0,\infty)$.
\end{prop}

\begin{proof}
The proof recasts the plug-in estimator of the estimand family, $\widehat{\theta}$, as a scalar contrast of the means of the arms of a multi-arm experiment --- the setting of Theorem~5 of \citet{liding2017}, which states a finite-population central limit theorem for a fixed linear contrast of arm means under complete random assignment. Casting $\widehat{\theta}$ in that form lets the theorem apply directly: The cells $(z, \bm{v})$ become the arms, and the estimand weights become the contrast weights. The proof sets up the correspondence and then verifies the theorem's conditions.

To match the notation of \citet{liding2017}, index the treatment profiles $(z, \bm{v}) \in \mathcal{T}$ by $q = 1,\dots,Q$, one arm per cell, and define scalar potential outcomes $Y_i(q) \coloneqq y_i(z, \bm{v})$. Let $\hat{Y}(q)$ and $\bar{Y}(q)$ denote the sample mean and finite-population mean in arm $q$. A contrast weight attaches to each arm the coefficient that the arm's mean carries in the estimand: The cells of racial condition $1$ enter $\theta_N$ with weights $g_1(\bm{v})$, and the cells of racial condition $0$ enter with weights $-g_0(\bm{v})$. Define the contrast weights accordingly,
\begin{align*}
c_q \coloneqq
\begin{cases}
g_1(\bm{v}) & \text{if arm } q \text{ indexes the profile } (1, \bm{v}), \\
- g_0(\bm{v}) & \text{if arm } q \text{ indexes the profile } (0, \bm{v}).
\end{cases}
\end{align*}
Substituting the arm notation into the definitions of $\widehat{\theta}$ in \eqref{eq: theta estimator} and of $\theta_N$ collects each into a single weighted sum over arms,
\begin{align*}
\widehat{\theta} = \sum_{q=1}^Q c_q \hat{Y}(q),
\qquad
\theta_N = \sum_{q=1}^Q c_q \bar{Y}(q),
\end{align*}
which is the scalar contrast of \citet[][Section~2]{liding2017} with a $1 \times Q$ contrast matrix whose entries are the $c_q$. The index $q$ and the potential-outcome notation $Y_i(q)$ are local to this proof and follow \citet{liding2017}. The index $q$ is unrelated to the joint PMF $q(z, \bm{v})$ and the conditionals $q_z$ of the main text, and $Y_i(q)$ denotes the potential outcome $y_i(z, \bm{v})$ rather than the observed outcome $Y_i$.

Assumption~\ref{assm: asymp regularity}(a)--(c) ensure that the conditions of \citet[][Theorem~5]{liding2017} hold. As $N \to \infty$, the arm fractions $n_q/N$ converge to positive limits by (a). The finite-population variances and covariances converge to finite limits as $N \to \infty$ by (b), and the Lindeberg-type condition on the maximum squared deviation also holds by (b). The contrast weights $c_q$ do not vary with $N$ by (c); the theorem takes the contrast matrix as fixed along the sequence, so weights that changed with $N$ would change the estimand from one population to the next and fall outside the theorem. Assumption~\ref{assm: asymp regularity}(d) is the nondegeneracy condition ensuring that the limiting variance $\sigma^2_{\theta,\infty} = \lim_{N \to \infty} N \, \Var[\widehat{\theta}]$ is strictly positive. With the conditions verified, Theorem~5 of \citet{liding2017}, applied to the scalar contrast $\widehat{\theta}$, gives
\begin{align*}
\sqrt{N} \left\{\widehat{\theta} - \theta_N\right\} \overset{d}{\longrightarrow} \mathcal{N}\left(0, \sigma^2_{\theta,\infty}\right).
\end{align*}
The standardized display of the proposition follows by Slutsky's theorem. Write the standardized statistic as a product,
\begin{align} \label{eq: clt standardization}
\dfrac{\sqrt{N} \left\{\widehat{\theta} - \theta_N\right\}}{\sigma_{\theta,N}} = \dfrac{\sqrt{N} \left\{\widehat{\theta} - \theta_N\right\}}{\sigma_{\theta,\infty}} \cdot \dfrac{\sigma_{\theta,\infty}}{\sigma_{\theta,N}}.
\end{align}
The first factor on the right-hand side of \eqref{eq: clt standardization} converges in distribution to $\mathcal{N}(0, 1)$ because dividing the limit $\mathcal{N}(0, \sigma^2_{\theta,\infty})$ by the constant $\sigma_{\theta,\infty} > 0$ rescales the limit to unit variance. The second factor is a nonrandom sequence that converges to $1$: $\sigma^2_{\theta,N} \to \sigma^2_{\theta,\infty} \in (0, \infty)$ by Assumption~\ref{assm: asymp regularity}(d), and the square root and the reciprocal are continuous at the strictly positive limit. Applied to the two factors of \eqref{eq: clt standardization} --- one converging in distribution and one converging to the constant $1$ --- Slutsky's theorem gives
\begin{align*}
\sqrt{N} \left\{\widehat{\theta} - \theta_N\right\} \, / \, \sigma_{\theta,N} \;\; \overset{d}{\longrightarrow} \;\; \mathcal{N}(0,1),
\end{align*}
which is the display of the proposition.
\end{proof}

\subsection{Confidence intervals for the NREC ratio} \label{app: NREC inference}

\Cref{sec: estimation} of the main text estimates the NREC by the instrumental-variables ratio $\widehat{\mathrm{NREC}} = \widehat{\Delta}_Y / \widehat{\Delta}_Z$, the ratio of the reduced-form estimator \eqref{eq: rf estimator} to the first-stage estimator \eqref{eq: fs estimator} of the main text, with the cue as the instrument and perceived race as the treatment taken up; \Cref{sec: excludability} of the main text develops how this instrumental-variables reading relates to the classical setting. A confidence interval for the ratio uses a delta-method approximation, justified for the finite population by \citet{pashley2022} together with a finite-population central limit theorem \citep{liding2017}. The interval is design-based, large-sample, and approximate rather than exact.

The justification runs along a sequence of finite populations of increasing size, indexed by $N$ as in \Cref{assm: asymp regularity}, with the cue in place of the configuration as the randomized assignment: $n_w$ units receive cue value $w \in \{0, 1\}$ under \Cref{assm: CRA cue} of the main text, and the arm shares $n_w / N$ converge to limits in $(0, 1)$, the analogue of \Cref{assm: asymp regularity}(a). Beyond randomization of the cue, the delta-method approximation rests on three conditions, which specialize to the ratio the requirements of the finite-population delta method \citep[Theorem~2]{pashley2022} and of the central limit theorem that the delta method draws on \citep[Theorem~5]{liding2017}:
\begin{itemize}
\item The complier share converges to a positive limit, $\abs{\mathcal{C}} / N \to \pi_{\mathcal{C}} > 0$. Under \Cref{assm: no defiers}, the population first stage equals the complier share, $\Delta_Z = \abs{\mathcal{C}} / N$, as in the proof of \Cref{cor: NREC lower bound}, so the positive limit keeps the denominator of the ratio away from zero along the sequence. The positive-limit condition strengthens \Cref{assm: compliers exist}, which requires only one complier at each population size $N$ and so permits a complier share that vanishes as $N$ grows.
\item Under each value $w$ of the cue, the finite-population means, variances, and covariances of the response $y_i(w)$ and of perceived race $z_i(w)$ converge to finite limits, and the response satisfies $\max_{1 \leq i \leq N} \left[y_i(w) - \frac{1}{N} \sum_{j=1}^N y_j(w)\right]^2 / N \to 0$. The two requirements are the analogues of \Cref{assm: asymp regularity}(b) for the cue design. The maximum-squared-deviation condition needs no separate statement for perceived race: Because $z_i(w) \in \{0, 1\}$, every squared deviation of $z_i(w)$ from its mean is at most $1$, so the maximum divided by $N$ is at most $1/N$, which tends to $0$.
\item The variance of the linearized statistic $\widehat{\Delta}_Y - \mathrm{NREC} \cdot \widehat{\Delta}_Z$, the numerator of \eqref{eq: NREC linearization} below, scaled by $N$, converges to a strictly positive limit, the analogue of \Cref{assm: asymp regularity}(d).
\end{itemize}

The delta method reaches the ratio from a linear statistic. Subtracting the NREC from the ratio and placing the difference over the common denominator $\widehat{\Delta}_Z$ gives the exact identity
\begin{align} \label{eq: NREC linearization}
\widehat{\mathrm{NREC}} - \mathrm{NREC} = \dfrac{\widehat{\Delta}_Y - \mathrm{NREC} \cdot \widehat{\Delta}_Z}{\widehat{\Delta}_Z}.
\end{align}
The numerator is the \textit{linearized statistic}: The numerator is linear in the pair $(\widehat{\Delta}_Y, \widehat{\Delta}_Z)$, and its mean over the randomization distribution of the cue is zero, because $\E[\widehat{\Delta}_Y] = \Delta_Y$ and $\E[\widehat{\Delta}_Z] = \Delta_Z$ --- the two-arm analogue of \Cref{prop: theta EV} --- and $\Delta_Y - \mathrm{NREC} \cdot \Delta_Z = 0$ by \Cref{prop: NREC identification}. Two steps then deliver the Normal approximation for the ratio. First, the pair $(\widehat{\Delta}_Y, \widehat{\Delta}_Z)$ satisfies a joint finite-population central limit theorem \citep[Theorem~5]{liding2017}. Second, the finite-population delta method \citep[Theorem~2]{pashley2022} replaces the random denominator $\widehat{\Delta}_Z$ in \eqref{eq: NREC linearization} by the nonrandom $\Delta_Z$; the replacement is asymptotically negligible because $\widehat{\Delta}_Z$ converges in probability to the positive limit $\pi_{\mathcal{C}}$.

The variance of the linearized statistic has the two-arm Neyman form. For each unit $i$ and cue value $w$, define the \textit{composite outcome} $a_i(w) \coloneqq y_i(w) - \mathrm{NREC} \cdot z_i(w)$: the response minus NREC times perceived race. Under the cue-level SUTVA (\Cref{assm: SUTVA cue}), the observed combination $Y_i - \mathrm{NREC} \cdot Z_i$ equals $a_i(W_i)$. The linearized statistic is therefore the difference of the two cue arms' means of the observed composite outcomes, and the identity \eqref{eq: neyman form identity}, applied with one cell per arm and the composite outcome in place of the potential outcome, gives
\begin{align} \label{eq: NREC linearized var}
\Var\left[\widehat{\Delta}_Y - \mathrm{NREC} \cdot \widehat{\Delta}_Z\right] & = \dfrac{S^2_a(1)}{n_1} + \dfrac{S^2_a(0)}{n_0} - \dfrac{S^2(\tau_a)}{N},
\end{align}
where the variance $\Var$ is taken over the randomization distribution of the cue assignment, $S^2_a(w)$ is the finite-population variance of the composite outcomes $a_i(w)$ under cue value $w$, and $S^2(\tau_a)$ is the finite-population variance of the unit-level differences $\tau_{a, i} \coloneqq a_i(1) - a_i(0)$. As in \Cref{cor: theta Var upper}, the subtracted term is unidentified --- each unit reveals its composite outcome under one cue value only --- and $S^2(\tau_a) \geq 0$, so dropping the subtracted term bounds the variance from above:
\begin{align} \label{eq: NREC var bound}
\Var\left[\widehat{\Delta}_Y - \mathrm{NREC} \cdot \widehat{\Delta}_Z\right] & \leq \dfrac{S^2_a(1)}{n_1} + \dfrac{S^2_a(0)}{n_0}.
\end{align}
The bound in \eqref{eq: NREC var bound} follows \citet[Section~3.1]{pashley2022}, whose conservative variance keeps the identified within-arm variances and drops the subtracted term, and corresponds to the branch of \Cref{cor: theta Var upper} that discards the variance of the unit-level contrasts, \eqref{eq: theta Var LD bound}, with the composite outcome in place of the potential outcome.

The plug-in estimator of the bound replaces each population piece by its sample counterpart, with the unknown NREC in the composite outcome replaced by the ratio $\widehat{\mathrm{NREC}}$:
\begin{align} \label{eq: NREC var bound est}
\widehat{V} \coloneqq \dfrac{\hat{S}^2_{\hat{a}}(1)}{n_1} + \dfrac{\hat{S}^2_{\hat{a}}(0)}{n_0},
\qquad
\hat{S}^2_{\hat{a}}(w) \coloneqq \dfrac{1}{n_w - 1} \sum_{i \, : \, W_i = w} \left(\hat{A}_i - \dfrac{1}{n_w} \sum_{j \, : \, W_j = w} \hat{A}_j\right)^2,
\end{align}
where $\hat{A}_i \coloneqq Y_i - \widehat{\mathrm{NREC}} \cdot Z_i$ is the estimated composite outcome for unit $i$. Under the three conditions, $N \widehat{V}$ is consistent for the limit of $N$ times the bound in \eqref{eq: NREC var bound}. Expanding the square shows that $\hat{S}^2_{\hat{a}}(w) = \hat{S}^2_Y(w) - 2\, \widehat{\mathrm{NREC}}\, \hat{S}_{YZ}(w) + \widehat{\mathrm{NREC}}^{\, 2}\, \hat{S}^2_Z(w)$, where $\hat{S}^2_Y(w)$, $\hat{S}^2_Z(w)$, and $\hat{S}_{YZ}(w)$ are the within-arm sample variances and covariance of the observed outcome and perceived race. Each within-arm sample variance is consistent for its limit by the argument of \Cref{prop: theta Var upper consistency}, and the sample covariance follows from the same argument applied to the sum of the outcome and perceived race; $\widehat{\mathrm{NREC}}$ is consistent for the NREC because its numerator and denominator converge in probability and the limiting denominator $\pi_{\mathcal{C}}$ is positive; and the expansion is a continuous function of the sample moments and the ratio. The reported confidence interval is
\begin{align*}
\widehat{\mathrm{NREC}} \; \pm \; \Phi^{-1}(1 - \alpha/2) \, \dfrac{\sqrt{\widehat{V}}}{\abs{\widehat{\Delta}_Z}},
\end{align*}
where $\Phi^{-1}(1 - \alpha/2)$ is the standard Normal quantile at level $1 - \alpha/2$; the division by $\abs{\widehat{\Delta}_Z}$ restores the denominator of the linearization in \eqref{eq: NREC linearization}, with $\widehat{\Delta}_Z$ in place of the unknown $\Delta_Z$.

Divided by $\Delta_Z^2$, the variance bound in \eqref{eq: NREC var bound} is the finite-population analogue of the usual asymptotic variance of the instrumental-variables estimator \citep[Section~4.2.1, pp.~138--140]{angristpischke2008} --- itself the variance of a linearization, with the first stage in its denominator. The interval therefore degrades when the cue moves perceived race only weakly \citep[Section~4.6.4, pp.~205--216]{angristpischke2008}: A small complier share leaves $\widehat{\Delta}_Z$ small relative to its own sampling variability, and asymptotic standard errors and confidence intervals can then suggest that an unstable estimate is stable \citep[Section~1.4]{imbensrosenbaum2005}. Approaches built for weak instruments include the permutation intervals of \citet{imbensrosenbaum2005} and the methods of \citet{kangpeckkeele2018} and \citet{aronowchanglopatto2026}.

\section{Additional Material} \label{app: additional material}

This section collects the summary table of grounds from \Cref{sec: substantive motivations} of the main text (\Cref{app: grounds table}).

\subsection{Summary of grounds for choosing among members of the estimand family} \label{app: grounds table}

\Cref{tab: grounds} collects the grounds developed in \Cref{sec: substantive motivations} of the main text, ordered by which component of the estimand each ground bears on.

\begin{table}[H]
\vspace{10pt}
\centering
{\singlespacing\small
\renewcommand{\arraystretch}{1.3}
\begin{tabular}{@{}>{\raggedright\arraybackslash}p{5.5cm}>{\raggedright\arraybackslash}p{7.7cm}@{}}
\toprule
\textbf{Ground} & \textbf{Implication for the estimand} \\
\midrule
\multicolumn{2}{@{}l}{\textit{Panel A. The contrast: which nonracial features are held fixed}} \\
\addlinespace[3pt]
Classical conception of the category & Hold every nonracial feature fixed: the all-else-equal effect \\
\addlinespace[3pt]
Thick constructivist conception & Let the indexed features vary by race: the within-race effect \\
\addlinespace[3pt]
Isolating a causal mechanism & Provide the information the decision-maker could otherwise infer and hold it fixed: a mixed effect \\
\addlinespace[3pt]
Normative and legal commitments & Hold the allowable features fixed; let the non-allowable vary: a mixed effect \\
\midrule
\multicolumn{2}{@{}l}{\textit{Panel B. The weighting: how the varying features are weighted}} \\
\addlinespace[3pt]
External validity & Calibrate $\rho$ or the $q_z$ to the target population \\
\addlinespace[3pt]
Typicality & Weight each conditional distribution $q_z$ by closeness to the group's prototype \\
\bottomrule
\end{tabular}

\smallskip
\parbox{13.6cm}{\raggedright\footnotesize \textit{Note}: Panel A grounds select the contrast and, with it, the member type; Panel B grounds set the distributions the chosen contrast uses. Every member also fixes a subset of units, typically all $N$ or, under the perception-based regime, the compliers $\mathcal{C}$ (\Cref{sec: anatomy estimand}).}
}
\caption{Grounds for choosing among the members of the estimand family}
\label{tab: grounds}
\end{table}

\section*{Disclosures}
\addcontentsline{toc}{section}{Disclosures}

\noindent \textbf{Funding.} This research received no external funding.

\noindent \textbf{Conflicts of interest.} The author declares no conflicts of interest.

\noindent \textbf{Use of artificial intelligence.} The author used Claude (Anthropic; Claude Opus 4 models, via Claude Code, \url{https://claude.com/claude-code}) during 2026 as an editing and verification assistant --- revising prose drafted by the author and checking notation, replication code, and figure code. All arguments, results, and interpretations are the author's own.

\printbibliography
\end{refsection}

\end{document}

%% file: estimands_macros.tex
\newcommand{\muHnat}{0.62}
\newcommand{\muAnn}{0.49}
\newcommand{\rhoNat}{0.52}
\newcommand{\AEWRest}{0.10}
\newcommand{\AEWRlo}{0.04}
\newcommand{\AEWRhi}{0.15}
\newcommand{\AEEest}{-0.01}
\newcommand{\AEElo}{-0.05}
\newcommand{\AEEhi}{0.03}

%% file: nrec_macros.tex
\newcommand{\fsEng}{0.86}
\newcommand{\nrecEng}{0.14}
\newcommand{\nrecEngLo}{0.08}
\newcommand{\nrecEngHi}{0.20}

%% file: language_macros.tex
\newcommand{\nSpanAnglo}{49}
\newcommand{\nSpanHisp}{77}
\newcommand{\pctAngloNonNative}{96}
\newcommand{\pctAngloNative}{4}
\newcommand{\pctHispNative}{83}

%% file: perceived_race_macros.tex
\newcommand{\pctAngloHispNat}{35}
\newcommand{\pctAngloHispEng}{8}
\newcommand{\pctHispWhiteNN}{35}
\newcommand{\pctHispWhiteEng}{6}

%% file: abstract-text.tex
Empirical studies of racial discrimination vary race while holding nonracial traits fixed, a design the literature defends as what credible inference requires. This defense bundles two claims: which effect of race a study should target, and whether its design can recover that effect. I separate them. The same randomization that secures credible estimation and inference recovers a family of race estimands, from the all-else-equal effect to a within-race effect that lets associated traits vary with race. Every member of that family is causal rather than descriptive, and the choice among members is a claim about what a racial category is --- a claim extending to ethnicity, religion, and other identity categories that index associated traits. I derive conditions, weaker than the literature's, for credible estimation and inference. Reanalyzing a Spanish-language campaign experiment, I find coethnic preference among Hispanic voters under the within-race effect and none under the all-else-equal effect.

%% file: tab_language_dist.tex
\centering
\begin{tabular}{lcccc}
\toprule
 & \textbf{Non-Native} & \textbf{Native} & & \textbf{Share of} \\
 & \textbf{Spanish} & \textbf{Spanish} & \textbf{Total} & \textbf{candidates} \\
\midrule
Within Anglo, $q_0$ & .96 & .04 & 1.00 & .39 \\
Within Hispanic, $q_1$ & .17 & .83 & 1.00 & .61 \\
\addlinespace
Common weight, $\rho$ & .48 & .52 & 1.00 & \\
\bottomrule
\end{tabular}
\caption{Language weights among U.S.\ House candidates who aired a Spanish-language advertisement, 2010--2018 \citep{zarateetal2024}. The common weight is the share-weighted average $\rho = .39\,q_0 + .61\,q_1$.}
\label{tab: lang dist}

%% file: tab_cell_means.tex
\centering
\begin{tabular}{lcc}
\toprule
\textbf{Candidate} & \textbf{Non-Native Spanish} & \textbf{Native Spanish} \\
\midrule
Anglo & 0.49 & 0.63 \\
Hispanic & 0.49 & 0.62 \\
\bottomrule
\end{tabular}
\caption{Mean evaluation (0--1 index) among Hispanic respondents by candidate race and campaign language, Study~1.}
\label{tab: cell means}